\documentclass[journal]{IEEEtran}
\IEEEoverridecommandlockouts
\usepackage{amsmath,amsfonts}
\usepackage{array}
\usepackage{chapterbib} % 支持每个章节单独的参考文献
\usepackage{textcomp}
\usepackage{stfloats}
\usepackage{url}
\usepackage{verbatim}
\usepackage{graphicx}
\usepackage{cite}
\usepackage[noend]{algpseudocode}
\usepackage[linesnumbered,ruled]{algorithm2e}
\usepackage{upgreek}
\usepackage{subfigure}
\usepackage{color}
\def\BibTeX{{\rm B\kern-.05em{\sc i\kern-.025em b}\kern-.08em
    T\kern-.1667em\lower.7ex\hbox{E}\kern-.125emX}}
    
\usepackage{amsthm}
\usepackage{amssymb}
\usepackage{makecell}
\usepackage{enumitem}
\newtheorem{theorem}{Theorem}
\newtheorem{definition}{Definition}
\theoremstyle{remark}
\theoremstyle{proposition}
\newtheorem{remark}{Remark}
\newtheorem{proposition}{Proposition}
\theoremstyle{Assumption}
\newtheorem{assumption}{Assumption}
\theoremstyle{Lemma}

\newtheorem{corollary}{\textbf{Corollary}}

\ifodd 0
\definecolor{darkgreen}{RGB}{0,200,0}
\newcommand{\rev}[1]{{\color{blue}#1}} %revise of the text
\else
\newcommand{\rev}[1]{#1}
\fi

\newcommand{\myref}[1]{%
    \ifthenelse{\equal{#1}{proof_proposition_submodular}}{A}{%
    \ifthenelse{\equal{#1}{proof_of_NPhard}}{B}{%
    \ifthenelse{\equal{#1}{proof_proposition_eq_u_m}}{C}{%
    \ifthenelse{\equal{#1}{proof_Successive_greedy_optimal}}{D}{%
    \ifthenelse{\equal{#1}{proof_DP_epsilon}}{E}{%
    \ifthenelse{\equal{#1}{proof_theorem_spec_time}}{F}{%
    \ifthenelse{\equal{#1}{proof_corollary_1}}{G}{%
    \ifthenelse{\equal{#1}{proof_theorem_spec_final}}{H}{
    {\textbf{??}}}}}}}}}}%
}

\begin{document}

\graphicspath{{figures/}}

\title{TrimCaching: Parameter-sharing Edge Caching for AI Model Downloading\\
\thanks{
The work was supported in part by the Research Grants Council of Hong Kong under Grant 27213824 and CRS HKU702/2. The work of Fangming Liu was supported in part by the Major Key Project of PCL under Grant PCL2025A10 and PCL2024A06, and in part by the Shenzhen Science and Technology Program under Grant RCJC20231211085918010. The work of Kaibin Huang was supported in part by the Research Grants Council of the Hong Kong Special Administrative Region, China under a fellowship award (HKU RFS2122-7S04), NSFC/RGC CRS (CRS\_HKU702/24), the Areas of Excellence scheme grant (AoE/E-601/22-R), Collaborative Research Fund (C1009-22G), and the Grants 17212423 \& 17304925, and in part by the Croucher Senior Fellowship and Shenzhen-Hong Kong-Macau Technology Research Programme (Type C) (SGDX20230821091559018).  

Guanqiao Qu, Zheng Lin, Qian Chen, Xianhao Chen, and Kaibin Huang are with the Department of Electrical and Computer Engineering, The University of Hong Kong, Pok Fu Lam, Hong Kong SAR, China. (e-mail: gqqu@eee.hku.hk; linzheng@eee.hku.hk; qchen@eee.hku.hk; xchen@eee.hku.hk; huangkb@eee.hku.hk). Jian Li is with the School of Cyber Science and Technology, University of Science and Technology of China, Hefei, Anhui, China. (e-mail: lijian9@ustc.edu.cn). Fangming Liu is with Peng Cheng Laboratory and Huazhong University of Science and Technology, China. (e-mail: fangminghk@gmail.com). \textit{(Corresponding author: Xianhao Chen.)}

A preliminary version of this work \cite{QuICDCS} has been accepted by the 44th IEEE International Conference on Distributed Computing Systems (ICDCS 2024).
}
}
\author{Guanqiao Qu,~\IEEEmembership{Graduate Student Member,~IEEE}, Zheng Lin,~\IEEEmembership{Graduate Student Member,~IEEE}, \\
Qian Chen,~\IEEEmembership{Member,~IEEE}, Jian Li,~\IEEEmembership{Senior Member,~IEEE}, Fangming Liu,~\IEEEmembership{Senior Member,~IEEE}, \\Xianhao Chen,~\IEEEmembership{Member,~IEEE}, Kaibin Huang,~\IEEEmembership{Fellow,~IEEE}
}

%\and
%\and
%\and
%\and
%\and

\maketitle

\begin{abstract}
%Storing AI models in the remote cloud for downloading cannot satisfy the stringent latency requirements of diverse AI inference services. 
Next-generation mobile networks are expected to facilitate fast AI model downloading to end users. By caching models on edge servers, mobile networks can deliver models to end users with low latency, resulting in a paradigm of \textit{edge model caching}. In this paper, we develop a novel model placement framework, called parame\underline{t}e\underline{r}-shar\underline{i}ng \underline{m}odel caching (TrimCaching). TrimCaching exploits the key observation that a wide range of AI models, such as convolutional neural networks or large language models, can share a significant proportion of parameter blocks containing reusable knowledge, thereby improving storage efficiency. To this end, we formulate a parameter-sharing model placement problem to maximize the cache hit ratio in multi-edge wireless networks by balancing the fundamental tradeoff between storage efficiency and service latency. We show that the formulated problem is a submodular maximization problem with submodular constraints, for which no polynomial-time approximation algorithm exists. To tackle this challenge, we study an important special case, where a small fixed number of parameter blocks are shared across models, which often holds in practice. In such a case, a polynomial-time algorithm with a $\left(1-\epsilon\right)/2$-approximation guarantee is developed. Subsequently, we address the original problem for the general case by developing a greedy algorithm. Simulation results demonstrate that the proposed TrimCaching framework significantly improves the cache hit ratio compared with state-of-the-art content caching without exploiting shared parameters in AI models.
\end{abstract}

\section{Introduction}
% \rev{In the era of Artificial Intelligence of Things (AIoT)}, there is an ongoing trend for pushing Artificial Intelligence (AI) services from the cloud to local devices due to the stringent latency and privacy requirements of diverse AI-empowered applications~\cite{deng2022actions,lin2024efficient,lyu2023optimal,lin2024adaptsfl}. For instance, autonomous driving must perceive the surrounding environments in a timely fashion~\cite{9284628,hu2024collaborative,chen2024vehicle}; the emerging large language models (LLMs) often involve privacy-sensitive personal data for human-machine interactions~\cite{9296274,nguyen2022federated}. On the other hand, due to the sheer number of AI applications and the ever-growing size of AI models (e.g., Google's on-device LLM Gemini Nano-2 with 3.25 billion parameters), it is impractical for users to store every AI model locally. For this reason, next-generation mobile networks aim to support rapid model downloading for mobile users to realize ``on-device AI'' with better latency and privacy guarantees. By ensuring efficient model downloading, the 6th generation (6G) mobile networks enable mobile users to enjoy versatile on-device AI services with low latency without exhausting their data storage capacities.
In the era of Artificial Intelligence of Things (AIoT), there is a growing trend of pushing Artificial Intelligence (AI) services from the cloud to local devices, enabling diverse AI-empowered applications~\cite{huang2025d,10.1145/3691620.3695055,10779389}. %~\cite{deng2022actions,lin2024efficient,lyu2023optimal,lin2024adaptsfl}. 
Aligned with this trend, on-device AI inference becomes prevalent due to increasingly powerful mobile AI chips, substantial privacy benefits, and low response time. For example, large language model (LLM)-enabled mobile health requires access to personal health information to generate customized advice~\cite{9296274,nguyen2022federated}; home robots must perceive private indoor environments to make informed decisions~\cite{li2024privacy}. In such scenarios, on-device inference is often preferred, as uploading sensitive personal data to cloud or edge servers may be prohibited due to privacy concerns and data protection regulations~\cite{chen2024fedmeld,10.1145/3548606.3560553}.

To effectively support on-device inference, AI model downloading will become a key component of next-generation mobile networks, as recognized by 3GPP~\cite{3gpp.22.874}. %~\cite{qu2024mobile,10183793,mao2023green,wu2023efficient,qu2024trimcachingICDCS}, 
%For instance, autonomous driving must perceive the surrounding environments in a timely fashion~\cite{9284628,hu2024collaborative,chen2024vehicle}; 
On-device inference demands \textit{real-time}, \textit{frequent} 6G model downloading services for several reasons. First, due to the sheer number of AI applications and the ever-growing size of AI models (e.g., Google’s on-device LLM Gemini Nano-2 with 3.25 billion parameters), it is impractical for users to prestore all needed models on mobile devices. Instead, they can delete infrequently-used models from local storage and only fetch them on demand from 6G mobile networks. Second, real-time model downloading also enables users to access context-aware AI services when entering a new region~\cite{3gpp.22.874}, such as downloading region-specific autonomous driving models or augmented reality models optimized for local conditions~\cite{9162277,3gpp.22.874,ma2025sense4fl,chen2024space}. At last, users should download the latest model versions periodically from cloud/edge learning (e.g., federated learning~\cite{chen2020joint,zhu2024byzantine}), as models continuously adapt to new data. Consequently, supporting real-time model downloading is essential to deliver storage-efficient, customized, and evolving AI for mobile users anywhere and anytime.

Conventional model downloading features the delivery of AI models from the remote cloud to local devices. The excessive downloading latency not only results in poor user experience but also renders it inapplicable for latency-sensitive services~\cite{wang2024ultra,fang2025prioritized}. For example, according to 3GPP, autonomous vehicles and robotic applications require AI model downloading to be completed within one second~\cite{3gpp.22.874}. Given the large sizes of models and the limited bandwidth of remote cloud centers, such fast model downloading may only be achieved by placing AI models in edge networks closer to users, giving rise to the paradigm of \textit{edge model caching}. However, unlike cloud centers, edge servers have limited storage capacity and can cache only a subset of popular models. To enhance model downloading performance, one fundamental research question, therefore, is model placement, e.g., how can we optimally place AI models on edge servers to enhance model downloading performance?

%In this paper, we identify an AI model placement problem by observing the parameter sharing among AI models. Parameter sharing pervades AI models nowadays, which can be exploited to store them more efficiently at the network edge. 

In this paper, we address a novel AI model placement problem by exploiting parameter sharing to improve model storage efficiency at the network edge. A broad range of AI models, such as convolutional neural networks (CNNs) or LLMs, can share a significant proportion of parameters, as they can be derived from the same pre-trained models~\cite{padmanabhan2022gemel,lin2023pushing}. For example, bottom-layer freezing is a classic technique in transfer learning and multi-task learning, since the bottom layers in deep neural networks, such as convolution layers, usually share common knowledge reusable for different downstream tasks~\cite{tenser2023,zhuang2020comprehensive,Guo_2019_CVPR,fang2024pacp}. %In the era of LLMs, parameter-efficient fine-tuning (PEFT), such as LoRA~\cite{hu2022lora}, has emerged as an effective approach to adapt foundation models for downstream applications by freezing a large proportion of parameters (e.g., over 99\% of parameters are frozen in LLMs fine-tuned with LoRA) to save computing and memory resources. 
In the era of LLMs, parameter-efficient fine-tuning (PEFT), such as LoRA~\cite{hu2022lora}, has emerged as an effective approach to adapt foundation models for downstream applications by freezing a large proportion of parameters, thereby saving computing and memory resources. For example, LoRA freezes all parameters in the pre-trained LLMs and introduces only a small number of trainable parameters for fine-tuning, typically less than 1\% of the total, implying over 99\% of the parameters remain unchanged during fine-tuning. 
%Moreover, today's engineers often fine-tune pre-trained models (e.g., pre-trained CNNs available in Pytorch or open-source foundation models such as GPT-2) on task-specific datasets to create new services. 
As shown in Fig. \ref{fig_intro}, the inference accuracy of a fine-tuned ResNet-50 degrades only slightly as the number of bottom layers frozen from the pre-trained model increases. Even when the first 90\% of trainable layers, up to layer 97, are frozen, the average accuracy degradation is only about 4.7\%, compared with the full-layer fine-tuning. As for LLMs, Fig. \ref{fig_intro_lora} illustrates that the BLEU scores of various fine-tuned GPT-2 medium models exhibit slight variation across different frozen parameter ratios. The fine-tuned models maintain satisfactory performance, even when up to 99.97\% of the parameters are frozen from the pre-trained model. In a nutshell, downstream models can share a significant proportion of parameters from the same pre-trained model, inherently making storage more efficient.\par

%This demonstrates the two downstream AI models, e.g., ``animal'' and ``transportation'' models in the figure, can share a significant proportion of layers, making storage more efficient. 
%Nowadays, downstream or personalized models can often have a significant proportion of shared model parameters, given that fine-tuning techniques are widely adopted and are particularly useful for large-sized models.
%For instance, OpenAI provides the API for end users to fine-tune a part parameters of the GPT-3 model based on their local dataset. 
%\par
%Besides, to enhance the inference accuracy, models are trained specifically on different downstream tasks' data sets. There is no one all-powerful model to cover all tasks. Services need to be served by different models. 
%Third, bottom layer sharing exists in early exit technique and backbone model. In the early exit area, data is input in a multi-level model. Each level has a classifier. In the bottom, classifiers can output coarse labesl while they can describe the data more precisely in the top. The classification result can be output early if the inference service does not require a precisie label. The input data passes through the same bottom layers until it reaches the toppest classfier. These bottom layers are shared among different models in the early exist approach.\par 
\begin{figure}[!t]
	\centering
\includegraphics[width=0.3\textwidth]{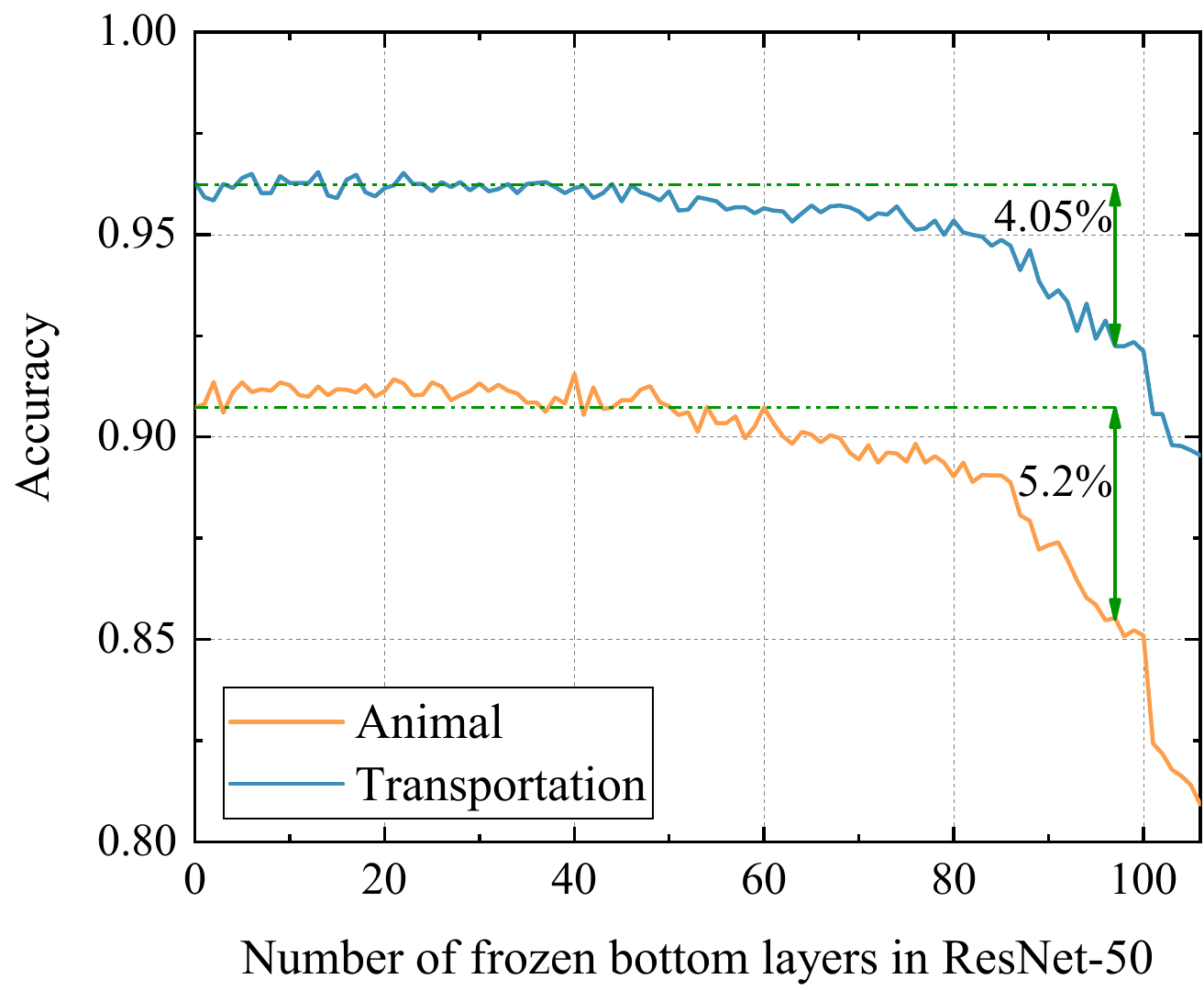}\label{fig_intro_50}
	\caption{Inference accuracy vs. the number of frozen bottom layers in fine-tuned ResNet-50. We fine-tune a ResNet-50~\cite{he2016deep}, pre-trained on CIFAR100~\cite{krizhevsky2009learning}, for two downstream tasks: ``transportation'' classification and ``animal'' classification. Specifically, the classes ``airplane", ``automobile", ``ship", and ``truck" in CIFAR10 \cite{krizhevsky2009learning} are grouped into the superclass ``transportation'', while the classes ``bird", ``cat", ``deer", ``dog", ``frog", and ``horse" are grouped into the superclass ``animal''.}\label{fig_intro}
    % This figure implies that downstream or personalized models can share a significant proportion of parameters, given that fine-tuning techniques are widely adopted nowadays.
\end{figure}
\begin{figure}[!t]
	\centering
\includegraphics[width=0.3\textwidth]{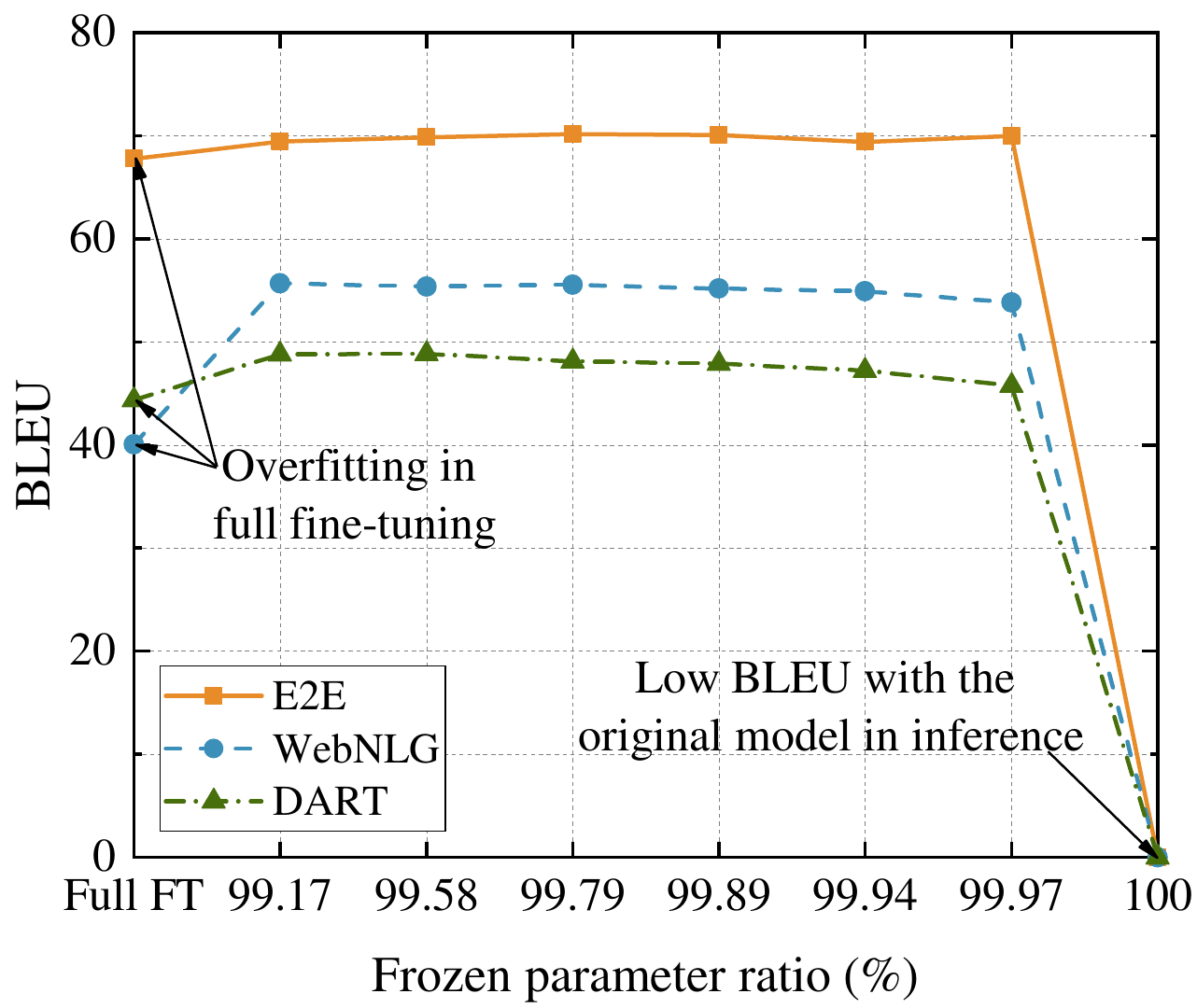}
\caption{BLEU vs. the frozen parameter ratio of fine-tuned GPT-2 medium models. We employ LoRA to fine-tune a pre-trained GPT-2 medium model on various datasets, including E2E~\cite{DBLP:journals/corr/NovikovaDR17}, WebNLG~\cite{gardent2017creating}, and DART~\cite{nan2021dart}. %LoRA freezes all parameters in the pre-trained LLMs and only introduces a small number of trainable parameters for fine-tuning (FT), typically less than 1\% of the total. 
BLEU, a widely used metric in evaluating the quality of the text generated by machine translation, is adopted here, with higher scores indicating that machine translation results are closer to those of professional human translations.}\label{fig_intro_lora}
%
%This underscores the necessity of PEFT and indicates that a large number of parameters, such as the pre-trained parameters of LLMs, can be shared across different LLMs in downstream tasks.
\end{figure}

By taking advantage of the aforementioned salient property, we will design a parame\underline{t}e\underline{r}-shar\underline{i}ng \underline{m}odel caching  (TrimCaching) framework for edge model caching. Specifically, given a set of wireless edge servers, we address the problem of placing models with shared parameters on edge servers to maximize the cache hit ratio for AI model downloading~\cite{6600983,7439797,7797148}. %more: 8169053
While this problem clearly falls into classical content placement problems, the shared parameters distinguish our scheme from traditional placement schemes for general content. In particular, when placing models with more shared parameters on an edge server, the storage becomes more efficient. However, this notable advantage also introduces a \textit{fundamental challenge}: the resulting submodular constraints make existing content placement schemes inapplicable, which lack theoretical guarantees in our scenario. The main contributions of this paper are summarized as follows.
%There is no consideration that elements in contents can be shared with each other in traditional content placement problems. 

\begin{enumerate}
	\item We define the parameter-sharing model caching problem. In multi-edge scenarios, we formulate the model placement problem to maximize the cache hit ratio (e.g., the ratio of downloading requests that can be served within delay requirements) under the storage capacity constraints of edge servers. 
    \item We show that the resulting problem is a submodular maximization problem with submodular constraints. After problem mapping, we conclude that there is no polynomial-time algorithm to solve it with a constant approximation guarantee due to the submodular constraints resulting from shared parameter blocks.
	\item We investigate a special case of the proposed problem, where the number of shared parameter blocks is independent of the problem scale (e.g., the scale of the AI model library). We argue this special case often holds in reality. We develop a successive greedy algorithm and a rounding dynamic programming (DP)-based algorithm to obtain a $\left(1-\epsilon\right)/2$-approximate solution with polynomial time complexity.
 % if transmission latency between edge servers (e.g., through optical fibers) can be ignored.
	\item We design a greedy-based algorithm for the original problem for the general case. Although a constant approximation guarantee cannot be achieved as alluded to earlier, we show that the algorithm is efficient and effective through performance evaluation.
	%\item \textcolor{red}{At last, we investigate the performance of the proposed algorithms in the single-edge scenario and the indentical parameter block sharing scenario in both special case and general case. }
\end{enumerate}

The rest of this paper is organized as follows. Section II reviews the related work. Section III elaborates on the system model and the TrimCaching framework. The problem formulation is presented in Section IV. The model placement algorithm is developed for the special case in Section V and for the general case in Section VI. Section VII presents the simulation results, and Section VIII concludes the paper. 

\section{Related Work}
Content caching in wireless edge networks has attracted significant attention since it brings content closer to end users, leading to the field of edge caching~\cite{8291028,7565183}. %more: ,7565185,10068141
Specifically, popular files can be pre-cached on wireless edge servers, such as small base stations \cite{7572146}, thereby allowing users to download files with reduced latency. 
% Most of the pre-caching strategies are based on file popularity \cite{8964499,8299576}. 
To this end, one fundamental research problem in edge caching is the content placement problem: how to place content on edge servers to enhance the quality of experience (QoE) of end users~\cite{9126265,8422142}.
% to maximize the cache hit ratio \cite{8964499,7572146,7511288} 
Due to the overlapping coverage of wireless edge servers, users can download content from any of the nearby edge servers that cache the requested content. Moreover, since content items are typically assumed to be independent in conventional edge caching, the storage constraints are naturally knapsack constraints. These characteristics result in a class of QoE maximization problems that are typically identified as submodular maximization problems with knapsack constraints~\cite{6600983,7439797,7797148}. %more: 8169053
% A typical polynomial algorithm for this problem is the greedy algorithm \cite{6108218}. 
% In such formulations, since content items are independent, the storage constraints are naturally knapsack constraints. 
Simple greedy methods are usually effective in solving such problems, providing approximation guarantees in polynomial time~\cite{6600983}. However, in our work, due to the parameter sharing across AI models, the shared parameters need to be cached only once on a server, introducing submodularity into the storage constraints. For this reason, when adapting the aforementioned content placement strategies, e.g., greedy algorithms~\cite{6600983}, to our case, no theoretical constant approximation guarantee can be achieved.
%This requires efforts for designing the tailored algorithm.
\par

% By implementing AI in MEC, EI can provide AI services to end users and address some limitations in cloud computing, such as latency reduction and privacy protection \cite{8884164}. 

Caching AI models in distributed wireless edge networks can facilitate model delivery for inference and training~\cite{qu2025partialloading,xu2020survey,xu2024cached,wei2025pipelining,lin2024split}. %more: tang2023energy
By enabling end users to download required models from edge servers, this approach significantly reduces service latency~\cite{9522156, 10183793,qu2024mobile}. Analogous to conventional edge content caching, existing studies on edge model caching primarily focus on optimizing model placement to enhance user QoE under the storage constraints of edge servers. However, these works do not consider parameter sharing among AI models when making model placement decisions. Although Wu et al. exploit shared parameters in AI models to develop a multi-user model downloading method~\cite{wu2024efficient}, this work focuses on designing a multicasting scheme rather than model placement. To the best of our knowledge, this work is the first to define the parameter-sharing model placement problem, identify its mathematical properties, and develop corresponding solution approaches.

\section{Parameter-sharing Model Caching Framework}
To enhance model caching efficiency, this section proposes the TrimCaching framework, including the edge network scenario, parameter-sharing model library, storage constraints, and the formulation of end-to-end (E2E) latency.
\subsection{Network Scenario}
% We consider a typical scenario in wireless edge networks, as shown in Fig. \ref{fig_system_model}. Let the set $\mathcal{M}=\left\{1,2,\dots,m\dots,M\right\}$ denote the $M$ wireless edge servers (e.g., base stations). All edge servers are interconnected. There are a set ${\mathcal{K}} = \left\{1,2,\dots,k,\dots,K\right\}$ of users covered by these edge servers. The network operator provides a library of AI models ${\mathcal{I}} = \left\{1,2,\dots,i,\dots,I\right\}$, which can be downloaded by users for inference services. The request probability and the quality of service (QoS) constraint on end-to-end (E2E) latency of user $k$ for model $i$ are ${p}_{k,i}$ and $\bar{T}_{k,i}$, respectively. The frequently used notations are summarized in Table  \ref{table_notation}.\par 
We consider a multi-edge multi-user scenario in wireless edge networks, as illustrated in Fig. \ref{fig_system_model}. Let $\mathcal{M}=\left\{1,\dots,m,\dots,M\right\}$ denote the set of interconnected wireless edge servers (e.g., base stations), with a set of users ${\mathcal{K}} = \left\{1,\dots,k,\dots,K\right\}$ distributed across their coverage areas. To support fast AI model downloading services for inference, these edge servers collectively cache a set of AI models ${\mathcal{I}} = \left\{1,\dots,i,\dots,I\right\}$. User $k$ requests to download model $i$ for performing local inference with probability $p_{k,i}$\footnote{In practice, $p_{k,i}$ can be obtained by monitoring and analyzing user historical traffic~\cite{8417940}.%more: 8493597
}. We use $\bar{T}_{k,i}$ to denote the E2E latency requirement of user $k$ for consuming model $i$, which consists of model downloading and local inference latency. The frequently used notations are summarized in Table \ref{table_notation}.

\begin{table}[!t]
	\centering
	\caption{Frequently Used Notations}
	\label{table_notation}
	\begin{tabular}{|m{2cm}<{\centering}|m{6cm}<{\centering}|}
		\hline
		\textbf{Symbol} & \textbf{Description}\\ \hline
        $u\left(m,i\right)$, $\dot{u}\left(m,i\right)$ & Number of cache hits and rounded number of cache hits of placing model $i$ on edge server $m$. \\ \hline
		$U\left(\cdot\right)$ & Cache hit ratio of all edge servers. \\ \hline
        $\hat{U}_m\left(\cdot\right)$  & Cache hit ratio of edge server $m$ in the special case. \\ \hline
        % $\hat{U}_m\left(\hat{\bf{X}}_m\right)$, 
        % $\hat{U}_m\left(\hat{\bf{X}}_m\right)$ & Cache hit ratio of edge server $m$ in the special case.\\ \hline
		${\bf{X}}$, ${\bf{X}}_{m}$ & Model placement decision for all edge servers and for edge server $m$ in the general case. \\ \hline
        $\hat{\bf{X}}_m$, $\hat{\bf{X}}=\bigcup\limits_{m\in\mathcal{M}}\hat{\bf{X}}_m$ & Model placement decision for edge server $m$ and for all edge servers in the special case. \\ \hline
        % $\dot{\bf{X}}_m$ & Model placement for edge server $m$ under the TrimCaching Spec algorithm with $\dot{u}\left(m,i\right)$. \\ \hline
        $\bf{X}^*$, $\hat{\bf{X}}^*_m$ & The optimal solution to $\mathcal{P}1.1$ and $\mathcal{P}2.1_m$, respectively. \\ \hline
        % $\bf{V}$    & The universal set of all possible model placements. \\ \hline
		$\mathcal{M}$, $\mathcal{M}_k$ & Set of all edge servers and set of edge servers covering user $k$. \\ \hline
        $\mathcal{K}$ & Set of users. \\ \hline
        $\mathcal{I}$, $\mathcal{I}_j$ & Set of AI models and set of models containing parameter block $j$. \\ \hline
        $\hat{\mathcal{I}}_m$& Set of AI models cached on edge server $m$ in the special case. \\ \hline
        % $\mathcal{J}$, $\mathcal{J}_{i}$, $\mathcal{J}^{\text{sh}}$, $\mathcal{J}^{\text{sh}}_{i}$& Set of parameter blocks, set of parameter blocks of model $i$, set of shared parameter blocks, set of shared parameter blocks of model $i$. \\ \hline
        $\mathcal{J}$, $\mathcal{J}_{i}$, $\mathcal{J}^{\text{sh}}$ & Set of parameter blocks, set of parameter blocks of model $i$, and set of shared parameter blocks. \\ \hline
        ${p}_{k,i}$, $\bar{T}_{k,i}$ & Request probability and E2E latency requirement of user $k$ for model $i$. \\ \hline
        $T_{m,k,i}$ & Latency for user $k$ to download model $i$ from edge server $m$ and perform on-device inference. \\ \hline
        $Q_{m}$, $g_m\left(\cdot\right)$ & Storage capacity and storage usage of edge server $m$.\\ \hline
        $D_i$, $D'_j$ & Size of model $i$ and size of parameter block $j$.\\ \hline
        $\bar{C}_{m,k}$, $C_{m,m'}$ & Expected downloading data rate between edge server $m$ and user $k$, and transmission data rate between edge server $m$ and $m'$. \\ \hline
	\end{tabular}
\end{table} 

\begin{figure}[!t]
	\centerline{\includegraphics[width=0.48\textwidth]{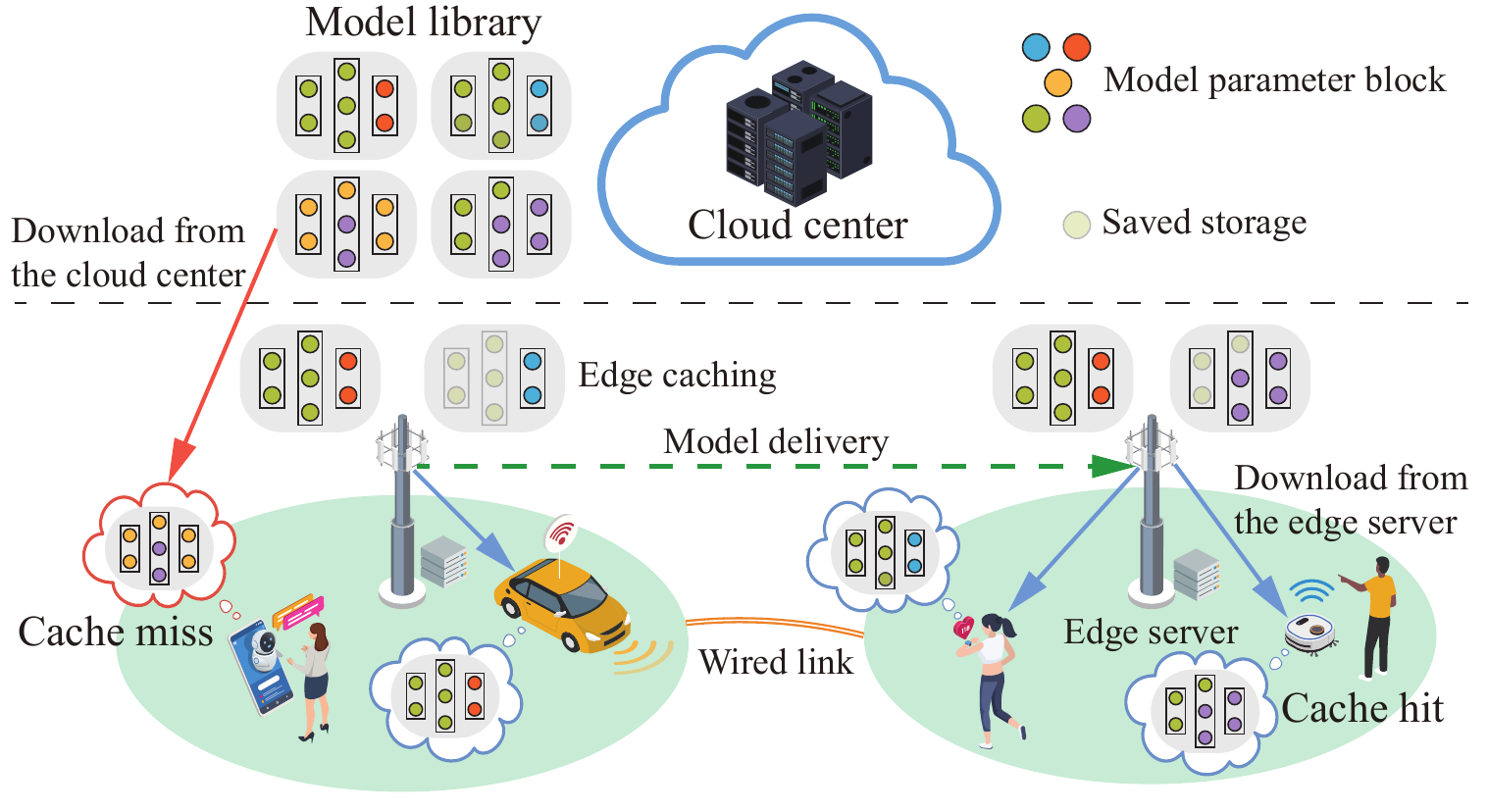}}
	\caption{The TrimCaching framework in a multi-edge scenario.}
	\label{fig_system_model}
\end{figure}

\subsection{Parameter-sharing Model Library}
We consider parameter sharing among models in $\mathcal{I}$. Given that AI models can contain billions of parameters, a \textit{parameter block} refers to a set of parameters, which reduces problem complexity. A parameter block can refer to a layer in a CNN, a block in a transformer, a set of low-dimensional trainable parameters in PEFT (e.g., the LoRA technique), and even an entire backbone network, and so on, depending on how parameters are grouped. Let $\mathcal{J}$ denote the set of parameter blocks of models in $\mathcal{I}$. Furthermore, let $\mathcal{J}_i$ denote the set of parameter blocks of model $i$, and let $\mathcal{I}_j$ denote the set of models in $\mathcal{I}$ that contain parameter block $j$. It is worth noting that a parameter block can be exclusive to one model or shared by multiple models. We formally define the concept of a shared parameter block.
\begin{definition}
    \textbf{Shared parameter block}: A parameter block $j$ is a shared parameter block if it is contained in more than one AI model in model set $\mathcal{I}$, i.e., $\left|\mathcal{I}_j\right|\ge 2$.
\end{definition}
Furthermore, we denote $\mathcal{J}^{\text{sh}}=\left\{j\ \middle|\ \left|\mathcal{I}_j\right|\ge 2 \right\}$ by the set of the shared parameter blocks across the models in $\mathcal{I}$. If a parameter block corresponds to a layer, a shared parameter block is also referred to as a \textit{shared layer}. Conversely, a parameter block is called a \textit{specific parameter block} if $\left|\mathcal{I}_j\right|=1$. %Moreover, we denote $\mathcal{J}^{\text{sh}}=\left\{j\ \middle|\ \left|\mathcal{I}_j\right|\ge 2 \right\}$ and $\mathcal{J}^{\text{sh}}_{i}=\left\{j\ \middle|\ j\in\mathcal{J}_{i},\left|\mathcal{I}_j\right|\ge 2 \right\}$ by the set of the shared parameter blocks across the models in $\mathcal{I}$ and the set of the shared parameter blocks of model $i$, respectively.

\subsection{Storage Constraints}
Let a binary variable $x_{m,i}$ indicate the placement status of model $i$ on edge server $m$, where $x_{m,i} =1$ indicates that model $i$ is cached on edge server $m$. Considering the storage capacity of edge servers, $x_{m,i}$ satisfies that
\begin{equation}\label{const_1}
    g_{m}\left({\bf{X}}_{m}\right)=\sum\limits_{j\in\mathcal{J}} D'_j\left[1-\prod\limits_{i\in\mathcal{I}_j}\left(1-x_{m,i}\right) \right]\le Q_m, \ \forall m\in\mathcal{M},
\end{equation}
where ${\bf{X}}_{m}=\left\{x_{m,i}\ \middle| \ i\in\mathcal{I}\right\}$ denotes the model placement decision for edge server $m$, $D'_j$ is the size of parameter block $j$, and $Q_m$ is the storage capacity of edge server $m$. Here, $1-\prod\limits_{i\in\mathcal{I}_j}\left(1-x_{m,i}\right) = 1$ implies that parameter block $j$ is stored only once on edge server $m$ if it is contained by more than one model cached on the server, thereby enhancing storage efficiency.
\subsection{E2E latency}
%\subsubsection{Model Downloading}%To ensure timely model downloading, the cloud center will push a set of models in the library to wireless edge servers in the offline stage. 
Considering user $k$ requests model $i$, 
%When user $k$ requests model $i$ for inference, it will first send the request to the associated edge servers. User $k$ can directly download the request model $i$, if model $i$ is cached in the associated edge servers. Otherwise, the request model will be migrated from the other edge servers to the request model will not be served to the user if the transmission latency does not meet the QoS requirement. 
two cases of model downloading exist.   \par 
\begin{itemize}
	\item \textbf{Downloading from associated edge servers}: User $k$ first sends a model request to its associated edge servers in $\mathcal{M}_k$, where $\mathcal{M}_k$ represents the set of edge servers covering user $k$. If there exists an edge server $m\in \mathcal{M}_k$ that caches model $i$, and the latency of downloading the model from server $m$ and performing on-device inference does not exceed the latency requirement $\bar{T}_{k,i}$, then a cache hit occurs.
 
	\item\textbf{Downloading from non-associated edge servers}: If none of the edge servers in $\mathcal{M}_k$ caches model $i$, the model will be fetched from another edge server $m$ in $\mathcal{M} \setminus \mathcal{M}_k$ where the model is cached. Specifically, the model is delivered from edge server $m$ to an edge server $m'$ in $\mathcal{M}_k$ and then to user $k$. A cache hit occurs if the E2E latency (including edge-to-edge, edge-to-user, and on-device inference latency) meets the latency requirement $\bar{T}_{k,i}$. 
\end{itemize}\par 
% \subsubsection{E2E latency}

The expected downloading data rate between edge server $m$ and its associated user $k$ is given by 
\begin{gather}\label{eq_communication}
\bar{C}_{m,k} = \bar{B}_{m,k}{\rm{log}}_2\left(1+\frac{\bar{P}_{m,k} \gamma_0  d_{m,k}^{-\alpha_0}}{n_0\bar{B}_{m,k}}\right), 
\end{gather}
where $\gamma_0$ is the antenna-related factor, $\alpha_0$ is the path loss factor, $d_{m,k}$ is the distance between user $k$ and edge server $m$, $\bar{P}_{m,k}$ and $\bar{B}_{m,k}$ are the expected transmit power and spectrum bandwidth of edge server $m$ allocated to user $k$, respectively, and $n_0$ is the power spectral density of the additive white Gaussian noise. Besides, we assume the transmission data rate between edge server $m$ and $m'$ is a constant and is denoted as $C_{m,m'}$. 

We denote $T_{m,k,i}$ as the E2E latency of user $k$, including model downloading latency and on-device inference latency, when downloading model $i$ from edge server $m$, which is given by
\begin{equation}
T_{m,k,i}=
\begin{cases}
    \frac{D_i}{\bar{C}_{m,k}} + t_{k,i}, \text{ if } m\in\mathcal{M}_{k},\\
    \mathop{\min}\limits_{m'\in\mathcal{M}_k}\left(\frac{D_i}{C_{m,m'}}+\frac{D_i}{\bar{C}_{m',k}}\right) +t_{k,i},\text{ if } m\notin\mathcal{M}_{k},\\
\end{cases}
\end{equation}
% If user $k$ can download model $i$ from its associated edge server, we have
% \begin{equation}
% T_{m,k,i}=\frac{D_i}{\bar{C}_{m,k}} + t_{k,i},
% \end{equation}
where $D_i$ is the size of model $i$ and $t_{k,i}$ is the inference latency of user $k$ using model $i$. 
% If edge server $m$, which is not associated with user $k$, supplies model $i$ for user $k$, we have
% \begin{equation}
% T_{m,k,i}=\mathop{\min}\limits_{m'\in\mathcal{M}_k}\left(\frac{D_i}{C_{m,m'}}+\frac{D_i}{\bar{C}_{m',k}}\right) +t_{k,i}.
% \end{equation}

%Therefore, we can calculate the expected communication data rate between users and edge servers as follows. Suppose that the associated user set of edge server $m$ is $\mathcal{K}_m$. The expected bandwidth and transmit power of the sub-channel assigned by edge server $m$ for its associated user $k$ are $B_m = \frac{B}{{p}_\text{A}\left|\mathcal{K}_m\right|}$ and $P_m = \frac{P}{{p}_\text{A}\left|\mathcal{K}_m\right|}$, where $B$ and $P$ are the total bandwidth and transmit power of an edge server. \par 

\subsection{Design Objective}
The TrimCaching framework aims to judiciously place AI models on edge servers to serve as many user requests as possible. For cache misses on edge servers, user requests can be forwarded to the cloud center for fetching the model. However, since downloading models from the cloud can be much slower, our policy design aims to maximize the cache hit ratio, which is a common design objective in the literature~\cite{8374917,dan1997multimedia}. Similar to prior works~\cite{8352848,zhang2015survey}, %more: huang2008mining
the cache hit ratio of model caches across edge servers in $\mathcal{M}$ is given by
\begin{equation}\label{eq_u_X}
    U\left({\bf{X}}\right)=\frac{\sum\limits_{k\in\mathcal{K}}\sum\limits_{i\in\mathcal{I}}{p}_{k,i}\left[1-\prod\limits_{m\in\mathcal{M}}\left(1-x_{m,i}{\mathbb{I}}_{1}\left(m,k,i\right)\right)\right]}{\sum\limits_{k\in\mathcal{K}}\sum\limits_{i\in\mathcal{I}}{p}_{k,i}},
\end{equation}
which captures the expected fraction of user model requests that can be served with the cached models on edge servers within the users' latency requirements. Here, ${\bf{X}}=\left\{x_{m,i} \ \middle| \ m\in\mathcal{M}, i\in\mathcal{I}\right\}$ represents the model caching decisions. Besides, ${\mathbb{I}}_{1}\left(m,k,i\right)$ is an indicator function, which is given by
\begin{equation}\label{eq_i1}
\mathbb{I}_{1}\left(m,k,i\right)={\mathbb{I}}_{\left\{T_{m,k,i}\le {\bar{T}}_{k,i}\right\}},
\end{equation} 
where ${\mathbb{I}}_{1}\left(m,k,i\right) = 1$ if and only if $T_{m,k,i}\le {\bar{T}}_{k,i}$. Consequently, $1-\prod\limits_{m\in\mathcal{M}}\left(1-x_{m,i}\mathbb{I}_{1}\left(m,k,i\right)\right) = 1$ indicates that at least one edge server with model $i$ in $\mathcal{M}$ can serve user $k$ within the latency constraint.

\section{Cache Hit Ratio Maximization Problem}
This section first formulates the cache hit ratio maximization problem under the TrimCaching framework. Then, we map the formulated problem to a known NP-hard problem and show that no polynomial-time algorithm can solve this problem with a constant approximation guarantee.
\subsection{Problem Formulation}
The TrimCaching framework aims to maximize the cache hit ratio for users' model requests by addressing the model placement problem under the storage capacity of edge servers and user latency requirements. %The expected cache hit ratio is 
%$T_{m,k,i}$ is the E2E latency when edge server $m$ delivers model $i$ to user $k$, including model downloading latency and on-device inference latency. If user $k$ can download model $i$ from its associated edge server, we have
% \begin{equation}
% T_{m,k,i}=\frac{D_i}{\bar{C}_{m,k}} + t_{k,i},
% \end{equation}
% where $D_i$ is the size of model $i$ and $t_{k,i}$ is the inference latency of user $k$ with model $i$. If edge server $m$, which is not associated with user $k$, supplies model $i$ for user $k$, we have
% \begin{equation}
% T_{m,k,i}=\mathop{\min}\limits_{m'\in\mathcal{M}_k}\left(\frac{D_i}{C_{m,m'}}+\frac{D_i}{\bar{C}_{m',k}}\right) +t_{k,i}.
% \end{equation}
The problem formulation is given as follows. 
\begin{subequations}
	\begin{equation}
		{\mathcal{P}1.1}:\ \mathop{\max}\limits_{{\bf{X}}}\ U\left({\bf{X}}\right)
	\end{equation}	
	\begin{equation}\label{general_problem_storage}
		{\rm{s.t.}}\ \eqref{const_1},
	\end{equation}	
	%	\begin{equation}\label{general_problem_routing}
		%		\sum\limits_{m\in\mathcal{M}}\sum\limits_{i\in{\mathcal{I}}_s}x_{m,k,i}\le 1,\ \forall k\in{\mathcal{K}},\forall s\in{\mathcal{S}},
		%	\end{equation}	
	\begin{equation}\label{general_problem_binary}
		x_{m,i}\in\left\{0,1\right\},\ \forall m\in{\mathcal{M}},\forall i\in{\mathcal{I}},
	\end{equation}	
\end{subequations}
where \eqref{const_1} ensures that the total storage used at each edge server does not exceed its capacity.

Note that our formulation ignores user mobility because the problem is solved based on a ``snapshot'' of user locations. This is commonly adopted in model placement schemes~\cite{6600983}. In practice, our algorithm can make model placement decisions by solving the above problem, then re-initiate model placement when the performance degrades to a certain threshold. Our simulation results will show that our algorithm is resilient to user mobility over time, thus eliminating the need for frequent model replacement that would consume backbone bandwidth.

\subsection{Problem Mapping}
Solving $\mathcal{P}1.1$ is extremely challenging due to the product of integer decision variables arising from parameter block sharing. In this subsection, we show that the problem can be mapped to a known NP-hard problem and that no polynomial-time algorithm can achieve a constant approximation guarantee.

We begin by introducing the concept of submodularity. A set function $f: 2^{{\bf{W}}}\rightarrow\mathbb{R}$ is called submodular if $f\left({\bf{S}}\right)+f\left({\bf{T}}\right)\ge f\left({\bf{S}}\cup {\bf{T}}\right) +f\left({\bf{S}}\cap {\bf{T}}\right)$ holds for all subsets ${\bf{S}},{\bf{T}}\subseteq {\bf{W}}$, where ${\bf{W}}$ is a finite ground set~\cite{fujishige2005submodular}. An equivalent characterization is that $f$ is submodular if and only if $f\left({\bf{S}}\cup\left\{w\right\}\right)-f\left({\bf{S}}\right)\ge f\left({\bf{T}}\cup\left\{w\right\}\right)-f\left({\bf{T}}\right)$ for all ${\bf{S}}\subseteq {\bf{T}}\subseteq {\bf{W}}$ and $w\in {\bf{W}}\setminus {\bf{T}}$. Conversely, a function $f$ is called supermodular if the reversed inequalities hold for every pair of subsets~\cite{lovasz1983submodular}. %Second, a constraint of the form $\sum\limits_{w\in {\bf{W}}}g\left(w\right)\le b$ is called a knapsack constraint, where $g\left(w\right)$ is a cost function and ${\bf{W}}$ is a finite set. Third, an algorithm for a maximization problem is said to achieve a constant approximation guarantee if its output is at least $\alpha$-times of the optimal solution, where $\alpha\in\left[0,1\right]$ is a constant.

The problem mapping for ${\mathcal{P}1.1}$ is presented in the following proposition.
\begin{proposition}\label{proposition_submodular}
${\mathcal{P}1.1}$ can be mapped to a submodular maximization problem with $M$ submodular constraints. %It can also be equivalently transformed to a supermodular maximization problem with $M$ knapsack constraints, as shown in ${\mathcal{P}1.2}$.
\end{proposition}
\begin{proof}
The proof is provided in Appendix \myref{proof_proposition_submodular} in our supplementary material.
\end{proof}

Next, we characterize the computational hardness of ${\mathcal{P}1.1}$.
\begin{proposition}\label{NPhard}
	${\mathcal{P}1.1}$ is an NP-hard problem. Moreover, no polynomial-time algorithm that solves ${\mathcal{P}1.1}$ with a constant approximation guarantee exists.
\end{proposition}
\begin{proof}
The proof is provided in Appendix \myref{proof_of_NPhard} in our supplementary material.
\end{proof}
%\begin{proposition}\label{Proposition_supermodular_bound}
%	 There does not exist a general polynomial algorithm with a constant-approximation guarantee to solve ${\mathcal{P}1.1}$.
%\end{proposition}
%\begin{proof}
%	$\mathcal{P}1.2$ is a supermodular maximization problem with one knapsack constraint when $M=1$. It is proved that there is not an approximation algorithm with a constant-approximation guarantee to maximize a supermodular function under a single knapsack constraint. 
%\end{proof}
Although ${\mathcal{P}1.1}$ cannot be solved approximately in general, in the following sections, we will first introduce a special case of $\mathcal{P}1.1$, which captures a typical parameter sharing scenario in practice, for which a polynomial-time algorithm can be developed to obtain a solution with a constant approximation guarantee. Subsequently, for the general case, we will propose a greedy algorithm to solve $\mathcal{P}1.1$. Although this algorithm does not offer a constant approximation guarantee, the solution approach is still highly effective.

\section{Special Case: A Small Fixed Number of
Shared Parameter blocks}
This section first presents the assumption for the special case of $\mathcal{P}1.1$, supported by concrete real-world examples. %Then, we propose an algorithm to decompose the original problem into sub-problems, which are solved sequentially. 
Under this special case, we design an algorithm with polynomial-time complexity and a constant approximation guarantee for maximizing the cache hit ratio in the TrimCaching framework.
\subsection{Assumption}
Recall that $\mathcal{J}^{\text{sh}}$ denotes the set of shared parameter blocks. The following assumption formally describes the special case considered in this paper.
\begin{assumption}\label{assumption_1}
    There is a small fixed number of shared parameter blocks among models in $\mathcal{I}$; that is, there exists a constant $C$, satisfying $C\ll \left|\mathcal{I}\right|$ and $\left|\mathcal{J}^{\text{sh}}\right|\le C$.
\end{assumption}
% where there is a small fixed number of shared parameter blocks. 
Such special cases are often observed in practice, as numerous downstream AI models can be derived from a small number of pre-trained models. 
% (e.g., Open-source models pre-trained on ImageNet or foundation models, such as GPT-3). 
For example, transfer learning enables fine-tuning models pre-trained on large-scale datasets for a wide range of specific tasks~\cite{zhuang2020comprehensive,basha2021autotune}. In practice, this is supported by frameworks such as TensorFlow and PyTorch, which provide a few pre-trained CNNs (e.g., ResNet models pre-trained on ImageNet) that can be adapted to various downstream tasks by freezing some layers and updating only task-specific layers~\cite{tenser2023,PyTorch}. %more: Keras
Similarly, in the LLM domain, PEFT methods adapt foundation models by updating only a tiny fraction of parameters, producing a large number of downstream models that share the same backbone models. For instance, Apple Intelligence freezes a base pre-trained model and fine-tunes many lightweight adapter layers for diverse tasks, such as summarization, proofreading, and mail replies~\cite{apple2024}. Moreover, Hugging Face hosts tens of thousands of LoRA modules, which are fine-tuned from a few pre-trained LLMs~\cite{ibm2024,hu2022lora}. Notably, even GPT models alone have hundreds of LoRA modules available. These examples demonstrate that, despite the vast number of models in the considered model library, there might only be a very small number of shared parameter blocks, where a block can be a layer or an entire pre-trained backbone (e.g., a pre-trained GPT model). An illustrative example of the considered special case is shown in Fig. \ref{fig_tree_model}, where all shared parameter blocks originate from 2 pre-trained models.

%Besides, to derive the approximation ratio, we assume that communication latency for model transmission among edge servers can be ignored. \textcolor{red}{
%%3GPP reported that a typical experienced data rate for downloading a $64\ \rm{MB}$ AI model within $1\ \rm{s}$ is $512 \ \rm{Mbps}$ \cite{3gpp.22.261}. 
%Modern wired backhaul infrastructure can handle up to $160$ signals, each providing a bandwidth of $10 \ \rm{Gbps}$ per fiber \cite{gsmabachkaul,8488527}. In contrast, the download speed of wireless links is much lower than that value. It is reported that the fastest average 5G download speed provided by CMHK in Hong Kong in 2022 is about $155 \ \rm{Mbps}$ \cite{cmhkdownload}, which is only $1.55\%$ of the backhaul link transmission speed.} Since edge servers are usually interconnected with high-speed optic fiber, the transmission latency among edge servers is usually significantly lower compared with the downloading latency between edge servers and users.
\par
%Suppose that the model library $\mathcal{I}$ is obtained by using the transfer learning and fine-tuning techniques. 
% Since only models with the same model structure can share frozen parameter blocks, we suppose that parameter blocks cannot be reused in models with different model structures. 
%Since the number of pre-trained models are known and their parameters are finite, the shared parameters are also finite and known. 
% Model 1 and 2 reuse the green parameter block 1-5 from the pre-trained model 1 in total. Similarly, Model 3 and 4 reuse the blue parameter block 8-10 and 12, 13, 15 from the pre-trained model 2. 
%Each model is composed of its corresponding specific parameter blocks and some shared parameter blocks from the pre-trained model. For example, model 1 has parameter blocks 1-5, 18, and 19. Parameter blocks 18 and 19 are specific parameters of model 1. The other parameters are shared parameter blocks reused from the pre-trained model 1. We cluster the shared parameter blocks from the pre-trained models and specific parameter blocks of each model, respectively.  \par 
\begin{figure}[!t]
	\centerline{\includegraphics[width=0.4\textwidth]{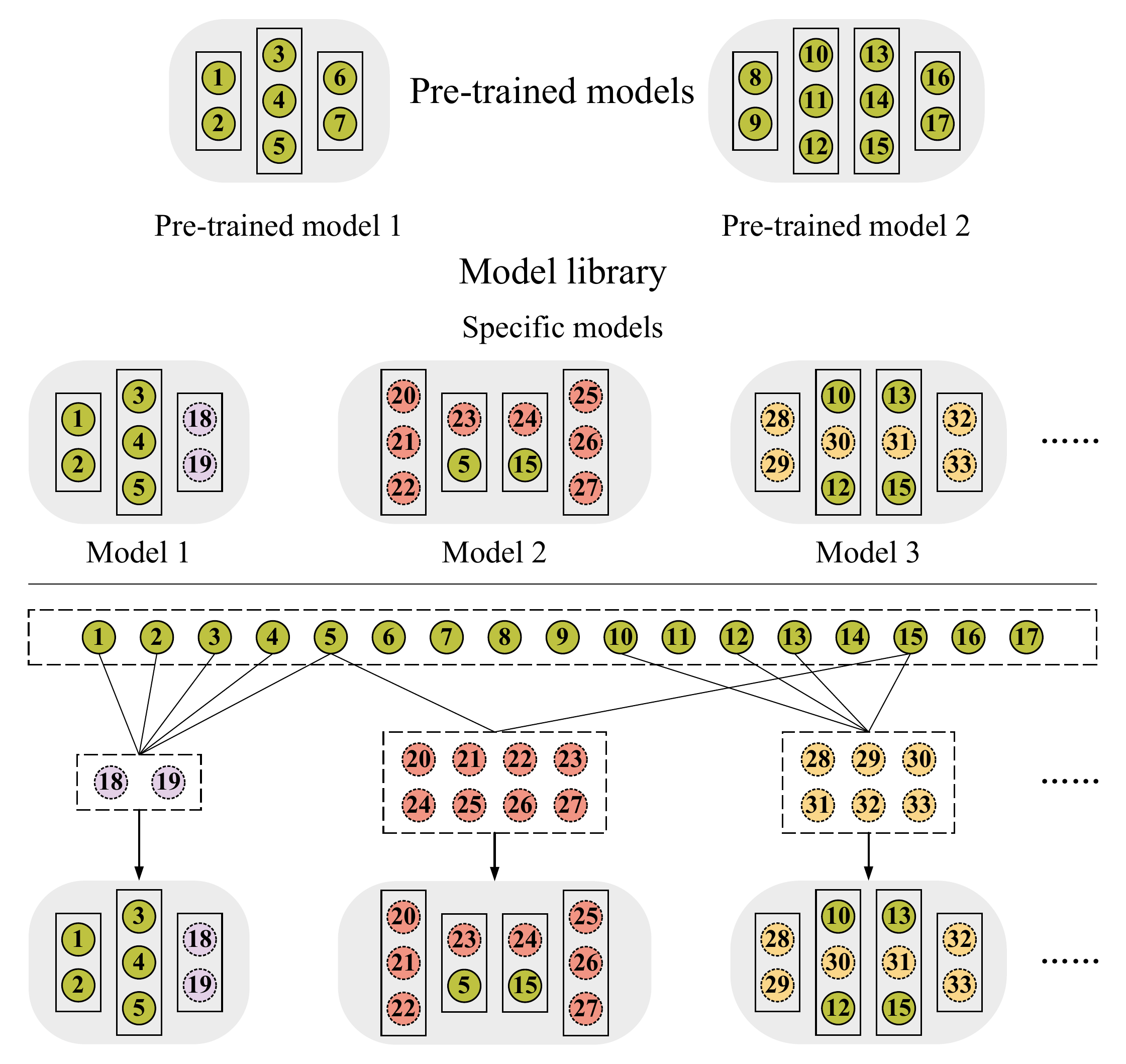}}
	\caption{An example of the special case with a small fixed number of shared parameter blocks. In the figure, regardless of the scale of the model library, the shared parameter blocks (green) come from two pre-trained models. Nodes in other colors represent specific parameter blocks.
 }
	\label{fig_tree_model}
\end{figure}

% First, we propose a successive greedy method to decompose $\mathcal{P}2.1$ into multiple sub-problems and solve sub-problems sequentially. The proposed successive greedy method can ensure a constant-approximation guarantee if the optimal solution of the sub-problem can be achieved. Second, we adopt a DP-based approach to solve each sub-problem. The DP-based algorithm can obtain an $\epsilon$-optimal solution for each sub-problem.
\par 
\subsection{TrimCaching Spec Algorithm}
Under Assumption \ref{assumption_1} in the special case, the number of shared parameter blocks is independent of the model library size, and thus of the problem scale, making it feasible to traverse all the shared parameter blocks. Building on this, our key idea is to develop an approximation algorithm by traversing the combinations of all shared parameter blocks while judiciously selecting specific parameter blocks, resulting in a polynomial-time algorithm and $\frac{1-\epsilon}{2}$ approximation ratio.

We propose a successive greedy-based algorithm to solve $\mathcal{P}1.1$ under the special case. The proposed TrimCaching Spec algorithm is summarized in Algorithm~\ref{algorithm_successive_greedy}. This algorithm operates by decomposing $\mathcal{P}1.1$ into $M$ sub-problems, each corresponding to edge server $m$, and solving sub-problems sequentially in ascending order of server indices. Specifically, in Line~\ref{line:successive_p2} of Algorithm~\ref{algorithm_successive_greedy}, the caching decision $\hat{{\bf{X}}}_{m}$ for edge server $m$ is determined by solving the $m$-th sub-problem, which is formulated as follows.
\begin{subequations}
	\begin{equation}
		{\mathcal{P}2.1}_m:\ \mathop{\max}\limits_{\hat{{\bf{X}}}_m}\ \hat{U}_m\left(\hat{{\bf{X}}}_m\right)
	\end{equation}	
	\begin{equation}
		{\rm{s.t.}} \ \sum\limits_{j\in\mathcal{J}} D'_j\left[1-\prod\limits_{i\in\mathcal{I}_j}\left(1-\hat{x}_{m,i}\right) \right]\le Q_m,\ \forall m\in\mathcal{M},
	\end{equation}	
	\begin{equation} 
		\hat{x}_{m,i}\in\left\{0,1\right\},\ \forall i\in\mathcal{I},
	\end{equation}	
\end{subequations}
%\begin{figure*}[b]
where\footnote{For notation simplicity, $\hat{U}_m\left(\hat{{\bf{X}}}_m\right)$ is used as shorthand for $\hat{U}_m\left(\hat{{\bf{X}}}_m\middle | \bigcup\limits_{m'=1}^{m-1}\hat{{\bf{X}}}_{m'}\right)$.} 
\begin{equation}\label{eq_hat_u}
	\begin{aligned}
	\hat{U}_m\left(\hat{{\bf{X}}}_m\right)
	=\frac{\sum\limits_{k\in\mathcal{K}}\sum\limits_{i\in\mathcal{I}}{p}_{k,i}\hat{x}_{m,i}\mathbb{I}_{1}\left(m,k,i\right)\mathbb{I}_{2}\left(m,k,i\right)}{\sum\limits_{k\in\mathcal{K}}\sum\limits_{i\in\mathcal{I}}{p}_{k,i}}.
	\end{aligned}
\end{equation}
%\end{figure*}
Here, $\mathbb{I}_{2}\left(m,k,i\right)$ is given by
\begin{equation}\label{eq_i2}
\mathbb{I}_{2}\left(m,k,i\right)=\prod\limits_{m'=1}^{m-1}\left(1-\hat{x}_{m',i} \mathbb{I}_{\left\{T_{m',k,i}\le\bar{T}_{k,i}\right\}}\right),
\end{equation}
where $\mathbb{I}_{2}\left(m,k,i\right)=1$ represents that the model request for model $i$ of user $k$ has not been satisfied by any of the first $m-1$ edge servers. For initialization, $\mathbb{I}_{2}\left(m,k,i\right)$ is set to 1 when $m=1$. %Moreover, it is noted that ${\mathcal{P}2.1}_m$ is equivalent to ${\mathcal{P}1.1}$ when there is only a single edge server in the network, i.e., $M = 1$.

%The proposed TrimCaching Spec algorithm is summarized in Algorithm \ref{algorithm_successive_greedy}.  \par 
\begin{algorithm}[!t]
	\caption{TrimCaching Spec Algorithm} %算法的名字
	\label{algorithm_successive_greedy}
	\LinesNumbered
	\KwIn{$\mathcal{I}$, $\mathcal{K}$, and $\mathcal{M}$.}
	\KwOut{$\hat{{\bf{X}}}$.} %算法的结果输出
	%{\bf Initialize:} $\mathcal{I}'=\mathcal{I}$.\\
	{\bf Initialize:} $\hat{{\bf{X}}}_m={\bf{0}}$ and $\mathbb{I}_{2}\left(m,k,i\right)=1$.\\
	\For{$m\in\mathcal{M}$}
	{
		Solve $\mathcal{P}2.1_m$ with Algorithm \ref{algorithm_DP} to obtain $\hat{{\bf{X}}}_m$ for edge server $m$.\label{line:successive_p2}\\ 
		%$\mathcal{I}'=\mathcal{I}'\setminus\left\{i\mid \hat{x}_{m,i}\in\hat{{\bf{X}}}_m \right\}$.\\
		%Update $\mathbb{I}_{2}\left(m,k,i\right)$.\\
	}
        $\hat{{\bf{X}}}=\bigcup\limits_{m\in\mathcal{M}}\hat{{\bf{X}}}_m$.\\
\end{algorithm}

The details of solving $\mathcal{P}2.1_{m}$ will be presented in the next subsection. With $\hat{{\bf{X}}}_m$, the solution to $\mathcal{P}1.1$ produced by Algorithm \ref{algorithm_successive_greedy} is denoted by $\hat{{\bf{X}}}=\bigcup\limits_{m\in\mathcal{M}}\hat{{\bf{X}}}_m$, satisfying the following proposition. % represent the solution to ${\mathcal{P}1.1}$ in the considered special case. Then, the cache hit ratio of $\hat{{\bf{X}}}$ follows from
\begin{proposition}\label{proposition_eq_u_m}
    The cache hit ratio of $\hat{{\bf{X}}}$ is equal to the sum of the cache hit ratios of $\hat{{\bf{X}}}_{m}$ for edge server $m$, i.e.,
\begin{equation}\label{eq_u_m}
U\left(\hat{{\bf{X}}}\right) = U\left(\bigcup\limits_{m\in\mathcal{M}}\hat{{\bf{X}}}_m\right) = \sum\limits_{m\in\mathcal{M}}\hat{U}_m\left(\hat{{\bf{X}}}_m\right).
\end{equation}
\end{proposition}

\begin{proof}
    The proof is shown in Appendix \myref{proof_proposition_eq_u_m} in our supplementary material.
\end{proof}

Furthermore, based on Proposition \ref{proposition_eq_u_m}, we have the following proposition.
\begin{proposition}\label{Successive_greedy_optimal}
	Assume that each sub-problem $\mathcal{P}2.1_m$ can be solved optimally in Algorithm \ref{algorithm_successive_greedy}. Under the considered special case, the TrimCaching Spec algorithm obtains a solution $\hat{{\bf{X}}}$ to $\mathcal{P}1.1$ which is lower bounded by $U\left(\hat{{\bf{X}}}\right)\ge\frac{1}{2}U\left({\bf{X}}^*\right)$, where ${\bf{X}}^*$ is the optimal solution to $\mathcal{P}1.1$.
\end{proposition}
\begin{proof}
The proof is presented in Appendix \myref{proof_Successive_greedy_optimal} in our supplementary material.
\end{proof}
\subsection{Rounding DP Approach and $\epsilon$-optimal Solutions for Sub-problems}
In this subsection, we propose a rounding DP-based algorithm to obtain an $\epsilon$-optimal solution to $\mathcal{P}2.1_m$. To begin with, we address the submodularity of the constraint in $\mathcal{P}2.1_m$ by decoupling the caching decisions for shared and specific parameter blocks. Specifically, under Assumption~\ref{assumption_1}, we determine the placement of shared parameter blocks via the exhaustive search. Once the placement decisions for shared parameter blocks are determined, we develop a rounding DP-based method to make the placement decisions for specific parameter blocks. For ease of presentation, let $\mathcal{A}$ denote the set of all possible combinations of shared parameter blocks in $\mathcal{J}^{\text{sh}}$, where each element $\mathcal{N}\in\mathcal{A}$ represents a specific combination of shared parameter blocks. For any $\mathcal{N}$, let $\mathcal{I}_\mathcal{N}=\left\{i \ \middle| \ \mathcal{J}_{i} \subseteq\mathcal{N}\right\}$ be the set of models whose shared parameter blocks are fully contained in $\mathcal{N}$. Additionally, let $d_{\mathcal{N}}$ denote the total size of the parameter blocks in $\mathcal{N}$, and let $D\left(i\right)$ denote the total size of the specific parameter blocks of model $i$\footnote{For example, in Fig. \ref{fig_tree_model}$, \{1,2,3,4,5,15\}$ and $\{10,12,13,15\}$ are two valid instances of $\mathcal{N}$. For $\mathcal{N}=\{1,2,3,4,5,15\}$, $\mathcal{I}_\mathcal{N}$ includes models 1 and 2. Moreover, $d_{\mathcal{N}}$ is the total size of parameter blocks \{1, 2, 3, 4, 5, 15\}, and $D\left(1\right)$ is the total size of parameter blocks \{18, 19\}.}. 

\subsubsection{Number of cache hits of models in $\mathcal{I}_{\mathcal{N}}$ on an edge server} Given $\mathcal{N}$, the number of cache hits of model $i\in\mathcal{I}_\mathcal{N}$ on edge server $m$ is given by
\begin{equation}\label{eq_utility}
	u\left(m,i\right)=\sum\limits_{k\in\mathcal{K}}{p}_{k,i}\mathbb{I}_{1}\left(m,k,i\right)\mathbb{I}_{2}\left(m,k,i\right),
\end{equation}
%Since the denominator in $\hat{U}_m\left(\hat{{\bf{X}}}_m\right)$ is a constant, maximizing $\hat{U}_m\left(\hat{{\bf{X}}}_m\right)$ is equivalent to maximizing the number of cache hits $\sum\limits_{i\in\mathcal{I}}\hat{x}_{m,i}u\left(m,i\right)$ in $\hat{U}_m\left(\hat{{\bf{X}}}_m\right)$. Therefore, we focus on maximizing $\sum\limits_{i\in\mathcal{I}}\hat{x}_{m,i}u\left(m,i\right)$ in the subsequent DP process.
which is a fixed-point number due to the existence of $p_{k,i}$. We denote the granularity of $u\left(m,i\right)$ of models in $\mathcal{I}_{\mathcal{N}}$ on edge server $m$ by $\delta_{m,\mathcal{N}}$, which reflects the precision of $u\left(m,i\right)$\footnote{For example, given $\mathcal{I}_{\mathcal{N}}=\{1,2\}$, if $u\left(m,1\right) = 0.12$ and $u\left(m,2\right)=0.14$, the precision of both values is two decimal places, and $\delta_{m,\mathcal{N}}=0.01$.}. 
Therefore, the total number of cache hits of models in $\mathcal{I}_\mathcal{N}$ on edge server $m$ has at most $W_{m,\mathcal{N}} + 1$ possible values, where $W_{m,\mathcal{N}}=\frac{\sum\limits_{i\in\mathcal{I}_\mathcal{N}}u\left(m,i\right)}{\delta_{m,\mathcal{N}}}$, and the $w_{m,\mathcal{N}}$-th value is $w_{m,\mathcal{N}}\delta_{m,\mathcal{N}}$, for $w_{m,\mathcal{N}}\in\left\{0,1,\dots,W_{m,\mathcal{N}}\right\}$. Since the DP method traverses all possible values of the total number of cache hits, a smaller $\delta_{m,\mathcal{N}}$ leads to a larger $W_{m,\mathcal{N}}$, thereby increasing the DP complexity. 

To facilitate the DP execution, we introduce a constant factor $\epsilon\in\left[0,1\right]$, and $u\left(m,i\right)$ is rounded to 
\begin{equation}\label{eq_round_u}
\dot{u}\left(m,i\right) = 
\begin{cases}
    \lfloor\frac{u\left(m,i\right)}{\epsilon u_{m,\min}}\rfloor,\ \epsilon >0,\\
    u\left(m,i\right), \ \epsilon=0,
\end{cases}
\end{equation}
where $u_{m,\min} = \mathop{\min}\limits_{i\in\mathcal{I}}u\left(m,i\right)$. %As a result, when determining the maximum cache hit ratio for edge server $m$, some $u\left(m,i\right)$ of models in $\mathcal{I}_{\mathcal{N}}$ can be rounded into the same value, and the number of feasible values of $w_{\mathcal{N}}$ can be decreased, thereby accelerating the DP process for $\mathcal{I}_{\mathcal{N}}$. In the rounding process, 
After rounding, the number of total cache hits of models in $\mathcal{I}_\mathcal{N}$ on edge server $m$ has at most $\dot{W}_{m,\mathcal{N}} + 1$ possible values, where $\dot{W}_{m,\mathcal{N}}=\frac{\sum\limits_{i\in\mathcal{I}_\mathcal{N}}\dot{u}\left(m,i\right)}{\dot{\delta}_{m,\mathcal{N}}}$, $\dot{\delta}_{m,\mathcal{N}}$ is the granularity of $\dot{u}\left(m,i\right)$ for models in $\mathcal{I}_{\mathcal{N}}$ on edge server $m$, and the $\dot{w}_{m,\mathcal{N}}$-th value is $\dot{w}_{m,\mathcal{N}}\dot{\delta}_{m,\mathcal{N}}$, for $\dot{w}_{m,\mathcal{N}}\in\left\{0,1,\dots,\dot{W}_{m,\mathcal{N}}\right\}$.%Similarly, $\mathcal{T}\left(e_{\mathcal{N}},\dot{w}_{\mathcal{N}}\right)$ is the smallest data size of specific parameter blocks that have to be cached on edge server $m$ for achieving the $\dot{w}_{\mathcal{N}}$-th value of the cache hits with the first $e_{\mathcal{N}}$ models in $\mathcal{I}_\mathcal{N}$. 

\subsubsection{Maximum cache hit ratio of models in $\mathcal{I}_{\mathcal{N}}$ on an edge server}Let $\mathcal{T}\left(e_{\mathcal{N}},\dot{w}_{\mathcal{N}}\right)$ represent the minimum data size of specific parameter blocks that must be cached on edge server $m$, which achieves $\dot{w}_{\mathcal{N}}\dot{\delta}_{m,\mathcal{N}}$ cache hits (i.e., the $\dot{w}_{\mathcal{N}}$-th value of the number of cache hits) with the first $e_{\mathcal{N}}$ models in $\mathcal{I}_\mathcal{N}$. The state-transition equation for updating $\mathcal{T}\left(e_{\mathcal{N}},\dot{w}_{\mathcal{N}}\right)$ is given by 
\begin{equation}\label{eq_dp}
  %\resizebox{1\hsize}{!}{$  
  \begin{aligned}	&\mathcal{T}\left(e_{\mathcal{N}},\dot{w}_{\mathcal{N}}\right)=\\
    &\begin{cases}
                    \begin{aligned}
				&{\min\left\{
					\begin{aligned}
					&\mathcal{T}\left(e_{\mathcal{N}}-1,\dot{w}_{\mathcal{N}}\right),\\
					&\mathcal{T}\left(e_{\mathcal{N}}-1,\dot{w}_{\mathcal{N}}-\frac{\dot{u}\left(m,e_{\mathcal{N}}\right)}{\dot{\delta}_{m,\mathcal{N}}}\right)\\
                    &+D\left(e_{\mathcal{N}}\right)
					\end{aligned}
					\right\},\dot{w}_{\mathcal{N}}\ge\frac{\dot{u}\left(m,e_{\mathcal{N}}\right)}{\dot{\delta}_{m,\mathcal{N}}},}\\
				&{\mathcal{T}\left(e_{\mathcal{N}}-1,\dot{w}_{\mathcal{N}}\right), \dot{w}_{\mathcal{N}}<\frac{\dot{u}\left(m,e_{\mathcal{N}}\right)}{\dot{\delta}_{m,\mathcal{N}}},}
                    \end{aligned}
		\end{cases}
    \end{aligned}%$}
\end{equation}
where $e_{\mathcal{N}}\in\left\{0,1,\dots,\left|\mathcal{I}_\mathcal{N}\right|\right\}$, and the initial value of $\mathcal{T}\left(e_{\mathcal{N}},\dot{w}_{\mathcal{N}}\right)$ is
\begin{equation}\label{eq_dp_initial}
\mathcal{T}\left(e_{\mathcal{N}},\dot{w}_{\mathcal{N}}\right) = 
\begin{cases}
    {\infty,\ \text{if } \dot{w}_{\mathcal{N}}\ne0,}\\
		{0,\ \text{if }\dot{w}_{\mathcal{N}}=0 \text{ or } e_{\mathcal{N}}=0.}
\end{cases}
\end{equation}

\begin{figure}[!t]
\centerline{\includegraphics[width=0.48\textwidth]{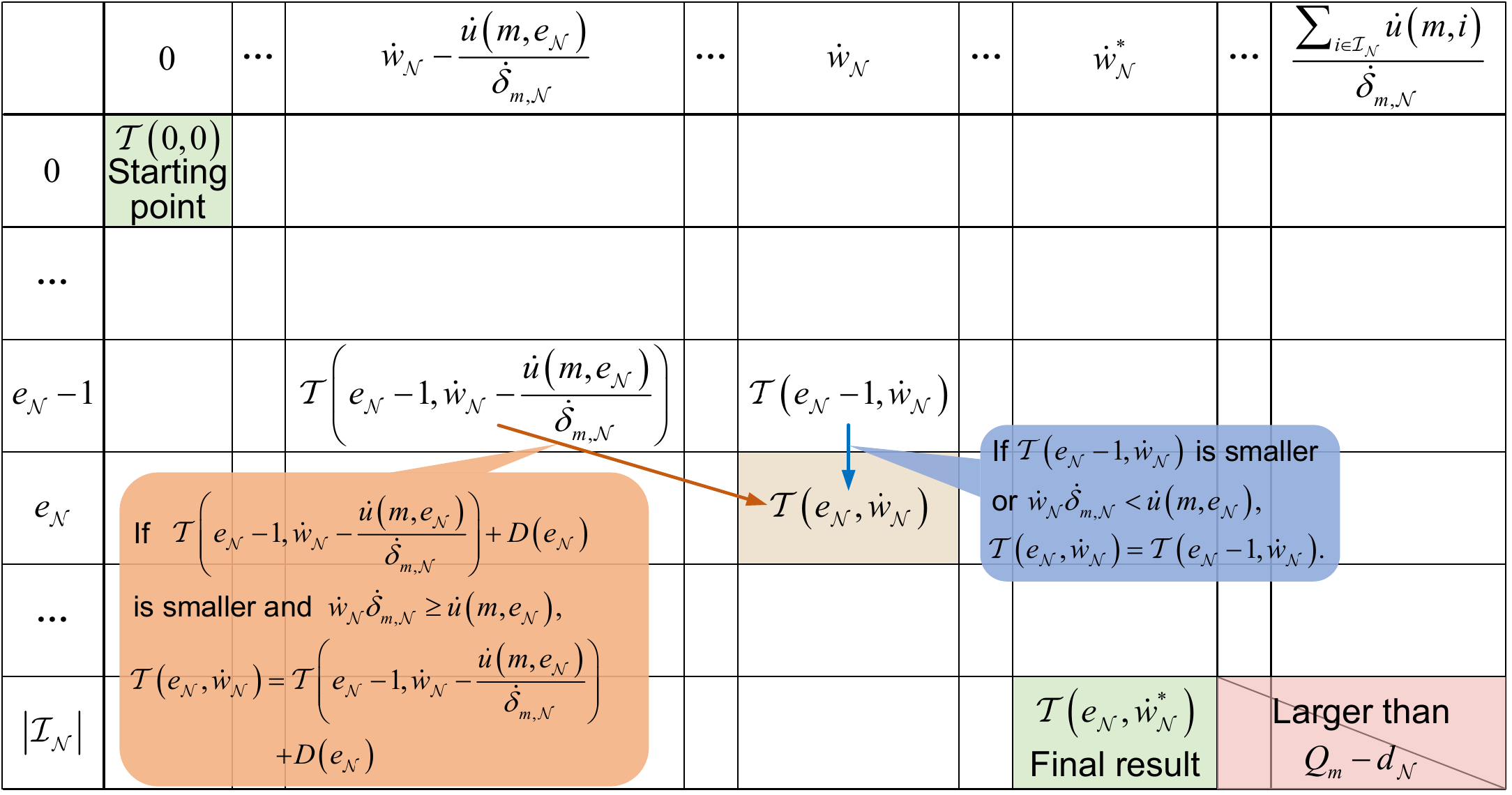}}
	\caption{The process of updating $\mathcal{T}\left(e_{\mathcal{N}},\dot{w}_{\mathcal{N}}\right)$ for edge server $m$, where the first row is the index of the number of cache hits and the first column is the model index in $\mathcal{I}_{\mathcal{N}}$.
 }
	\label{fig_dp}
\end{figure}

The detailed process of updating $\mathcal{T}\left(e_{\mathcal{N}},\dot{w}_{\mathcal{N}}\right)$ is illustrated in Fig. \ref{fig_dp} and summarized below. When calculating $\mathcal{T}\left(e_{\mathcal{N}},\dot{w}_{\mathcal{N}}\right)$, we consider the following two cases. (1) If $\dot{w}_{\mathcal{N}}\ge\frac{\dot{u}\left(m,e_{\mathcal{N}}\right)}{\dot{\delta}_{m,\mathcal{N}}}$, i.e., the number of cache hits of model $e_{\mathcal{N}}$ is less than or equal to the $\dot{w}_{\mathcal{N}}$-th value of the number of cache hits, then $\mathcal{T}\left(e_{\mathcal{N}},\dot{w}_{\mathcal{N}}\right)$ is the minimum of $\mathcal{T}\left(e_{\mathcal{N}} - 1,\dot{w}_{\mathcal{N}}\right)$ and $\mathcal{T}\left(e_{\mathcal{N}}-1,\dot{w}_{\mathcal{N}}-\frac{\dot{u}\left(m,e_{\mathcal{N}}\right)}{\dot{\delta}_{m,\mathcal{N}}}\right) + D\left(e_{\mathcal{N}}\right)$.
\begin{itemize}
\item If $\mathcal{T}\left(e_{\mathcal{N}} - 1,\dot{w}_{\mathcal{N}}\right)$ is smaller, it indicates that a subset of the first $e_{\mathcal{N}}-1$ models can achieve the $\dot{w}_{\mathcal{N}}$-th value of the number of cache hits with a smaller data size. Therefore, the same data size is sufficient to achieve the $\dot{w}_{\mathcal{N}}$-th value of the number of cache hits using the first $e_{\mathcal{N}}$ models.
\item If $\mathcal{T}\left(e_{\mathcal{N}}-1,\dot{w}_{\mathcal{N}}-\frac{\dot{u}\left(m,e_{\mathcal{N}}\right)}{\dot{\delta}_{m,\mathcal{N}}}\right) + D\left(e_{\mathcal{N}}\right)$ is smaller, it represents $\dot{w}_{\mathcal{N}}\dot{\delta}_{m,\mathcal{N}}$ is achieved by including model $e_{\mathcal{N}}$, which contributes $\dot{u}\left(m,e_{\mathcal{N}}\right)$, and using a subset of the first $e_{\mathcal{N}}-1$ models to provide the remaining number of cache hits. Therefore, $D\left(e_{\mathcal{N}}\right)$ is added to $\mathcal{T}\left(e_{\mathcal{N}}-1,\dot{w}_{\mathcal{N}}-\frac{\dot{u}\left(m,e_{\mathcal{N}}\right)}{\dot{\delta}_{m,\mathcal{N}}}\right)$.
\end{itemize}
(2) If $\dot{w}_{\mathcal{N}}<\frac{\dot{u}\left(m,e_{\mathcal{N}}\right)}{\dot{\delta}_{m,\mathcal{N}}}$, i.e., the number of cache hits of model $e_{\mathcal{N}}$ exceeds the $\dot{w}_{\mathcal{N}}$-th value of the number of cache hits, then the $\dot{w}_{\mathcal{N}}$-th value of the number of cache hits must be achieved only by a subset of the first $e_{\mathcal{N}}-1$ models without including model $e_{\mathcal{N}}$. 
  %Finally, after updating all $\mathcal{T}\left(e_{\mathcal{N}},w_{\mathcal{N}}\right)$, the maximum number of cache hits of edge server $m$ with $\mathcal{N}$, i.e., $w^*_{\mathcal{N}}\delta_{m,\mathcal{N}}$, can be determined, where 
% \begin{equation}
% w^*_{\mathcal{N}}=\mathop{\arg\max}\limits_{w_{\mathcal{N}}}\{w_{\mathcal{N}}\mid\mathcal{T}\left(\left|\mathcal{I}_{\mathcal{N}}\right|,w_{\mathcal{N}}\right) \le Q_m-d_{\mathcal{N}}\}.
% \end{equation}

After obtaining $\mathcal{T}\left(e_{\mathcal{N}},\dot{w}_{\mathcal{N}}\right)$ for all $e_{\mathcal{N}}$ and $\dot{w}_{\mathcal{N}}$, the index corresponding to the maximum number of cache hits of models in $\mathcal{I}_{\mathcal{N}}$ on edge server $m$ is given by
\begin{equation}\label{eq_rounded_w}
\dot{w}^*_{\mathcal{N}}=\mathop{\arg\max}\limits_{\dot{w}_{\mathcal{N}}}\{\dot{w}_{\mathcal{N}}\mid\mathcal{T}\left(\left|\mathcal{I}_{\mathcal{N}}\right|,\dot{w}_{\mathcal{N}}\right) \le Q_m-d_{\mathcal{N}}\}.
\end{equation}
%Furthermore, after traversing all feasible $\mathcal{N}$ in $\mathcal{A}$, the index of the maximum number of cache hits of models on edge server $m$ is given by $\dot{w}^*_{\mathcal{N}^*}$, where $\mathcal{N}^*=\mathop{\arg\max}\limits_{\mathcal{N}}\{\dot{w}^*_{\mathcal{N}}\}$ is the corresponding combination of shared parameter blocks. 
With $\dot{w}^*_{\mathcal{N}}$, the maximum cache hit ratio of models in $\mathcal{I}_{\mathcal{N}}$ on edge server $m$ is expressed as
% \begin{equation}\label{eq_sp_optimal}
% \hat{U}_m\left(\hat{{\bf{X}}}_m\right) = \frac{w^*_{\mathcal{N}^*}}{\sum\limits_{k\in\mathcal{K}}\sum\limits_{i\in\mathcal{I}}{p}_{k,i}}=\frac{\sum\limits_{i\in\hat{\mathcal{I}}_{m}}{u}\left(m,i\right)}{\sum\limits_{k\in\mathcal{K}}\sum\limits_{i\in\mathcal{I}}{p}_{k,i}},
% \end{equation} 
% where $\hat{{\bf{X}}}_m$ is the model caching decision corresponding to $w^*_{\mathcal{N}^*}$, and $\hat{\mathcal{I}}_{m} = \left\{i\mid\hat{x}_{m,i}\in\hat{{\bf{X}}}_{m}\right\}$. The details of obtaining $\hat{{\bf{X}}}_m$ will be elaborated in the following paragraph. %Thus, the optimal solution to $\mathcal{P}2.1_m$ is obtained.
% Besides, we define $\dot{w}^*_{\mathcal{N}}=\mathop{\arg\max}\limits_{\dot{w}_{\mathcal{N}}}\{\dot{w}_{\mathcal{N}}\mid\mathcal{T}\left(\left|\mathcal{I}_{\mathcal{N}}\right|,\dot{w}_{\mathcal{N}}\right) \le Q_m-d_{\mathcal{N}}\}$ and $\dot{w}^*_{\mathcal{N}^*}\delta_{m,\mathcal{N}^*}=\mathop{\max}\limits_{\mathcal{N}}\{\dot{w}^*_{\mathcal{N}}\delta_{m,\mathcal{N}}\}$. At last, in $\mathcal{P}2.1_{m}$, the model caching decision corresponding to $\dot{w}^*_{\mathcal{N}^*}$ is denoted by $\hat{{\bf{X}}}_m$, and 
\begin{equation}\label{eq_rounded_u}
	\hat{U}_m\left(\hat{{\bf{X}}}_{m,\mathcal{N}}\right) = \frac{\sum\limits_{i\in\hat{\mathcal{I}}_{m,\mathcal{N}}}{u}\left(m,i\right)}{\sum\limits_{k\in\mathcal{K}}\sum\limits_{i\in\mathcal{I}}{p}_{k,i}},
\end{equation} 
where $\hat{{\bf{X}}}_{m,\mathcal{N}}$ is the model caching decision corresponding to $\dot{w}^*_{\mathcal{N}}$, and $\hat{\mathcal{I}}_{m,\mathcal{N}} = \left\{i\ \middle | \ \hat{x}_{m,i}=1,\hat{x}_{m,i}\in\hat{{\bf{X}}}_{m,\mathcal{N}}\right\}$. The details for determining $\hat{{\bf{X}}}_{m,\mathcal{N}}$ will be presented in the following paragraphs. Moreover, note that the rounding DP-based algorithm described above operates based on $\dot{u}\left(m,i\right)$, %and the realized maximum number of rounded cache hits is $\sum\limits_{i\in\dot{\mathcal{I}}_{m}}\dot{u}\left(m,i\right)$. However, 
whereas we use $u\left(m,i\right)$ to calculate $\hat{U}_m\left(\hat{{\bf{X}}}_{m,\mathcal{N}}\right)$ in \eqref{eq_rounded_u} at last. 

\subsubsection{Optimal model caching of models in $\mathcal{I}_{\mathcal{N}}$ on an edge server}The determination of $\hat{{\bf{X}}}_{m,\mathcal{N}}$ is summarized in Algorithm \ref{algorithm_recursive} and illustrated in Fig. \ref{fig_recursive}. The algorithm determines $\hat{{\bf{X}}}_{m,\mathcal{N}}$ recursively by comparing $\dot{u}\left(m,e_{\mathcal{N}}\right)$ with the remaining number of cache hits $\dot{w}_{\mathcal{N}}\dot{\delta}_{m,\mathcal{N}}$, starting from model $\left|\mathcal{I}_{\mathcal{N}}\right|$ and number of cache hits $\dot{w}^*_{\mathcal{N}}\dot{\delta}_{m,\mathcal{N}}$. In Line \ref{line:recursive_x} of Algorithm \ref{algorithm_recursive}, the $e_{\mathcal{N}}$-th model in $\mathcal{I}_\mathcal{N}$ is placed on edge server $m$ if the following two conditions are satisfied. 
\begin{itemize}
    \item $\dot{w}_{\mathcal{N}}\dot{\delta}_{m,\mathcal{N}}\ge \dot{u}\left(m,e_{\mathcal{N}}\right)$, implying that the number of cache hits of the $e_{\mathcal{N}}$-th model is no greater than the remaining number of cache hits to be satisfied.
    \item $\mathcal{T}\left(e_{\mathcal{N}}-1,\dot{w}_{\mathcal{N}}-\frac{\dot{u}\left(m,e_{\mathcal{N}}\right)}{\dot{\delta}_{m,\mathcal{N}}}\right)+D\left(e_{\mathcal{N}}\right)<\mathcal{T}\left(e_{\mathcal{N}}-1,\dot{w}_{\mathcal{N}}\right)$, representing that the number of cache hits $\dot{w}_{\mathcal{N}}\dot{\delta}_{m,\mathcal{N}}$ is achieved by model $e_{\mathcal{N}}$, contributing $\dot{u}\left(m,e_{\mathcal{N}}\right)$, and a subset of the first $e_{\mathcal{N}}-1$ models, contributing $\dot{w}_{\mathcal{N}}\dot{\delta}_{m,\mathcal{N}}-\dot{u}\left(m,e_{\mathcal{N}}\right)$. Besides, the corresponding total data size, $\mathcal{T}\left(e_{\mathcal{N}}-1,\dot{w}_{\mathcal{N}}-\frac{\dot{u}\left(m,e_{\mathcal{N}}\right)}{\dot{\delta}_{m,\mathcal{N}}}\right)+D\left(e_{\mathcal{N}}\right)$, is less than $\mathcal{T}\left(e_{\mathcal{N}}-1,\dot{w}_{\mathcal{N}}\right)$. 
\end{itemize}
%In each while loop, Algorithm \ref{algorithm_recursive} first determines whether the model combination corresponding to the $\dot{w}'_{\mathcal{N}}$-th value of the number of cache hits includes the $e'_{\mathcal{N}}$-th model in $\mathcal{I}_{\mathcal{N}}$, starting from the pair $\left(\left|\mathcal{I}_{\mathcal{N}}\right|,w^*_{\mathcal{N}}\right)$. 
%In line 3, it first checks whether the number of cache hits of the model $e'_{\mathcal{N}}$ is smaller than $\dot{w}'_{\mathcal{N}}\dot{\delta}_{m,\mathcal{N}}$. If yes, in line 4, it further checks whether the total data size of $\mathcal{T}\left(e'_{\mathcal{N}},\dot{w}'_{\mathcal{N}}-\frac{\dot{u}\left(m,e'_{\mathcal{N}}\right)}{\dot{\delta}_{m,\mathcal{N}}}\right)$ and $D_{\mathcal{N}}\left(e'_{\mathcal{N}}\right)$ is less than $\mathcal{T}\left(e'_{\mathcal{N}},\dot{w}'_{\mathcal{N}}\right)$. These two lines ensure that the $e'_{\mathcal{N}}$-th model must be cached if the first $e'_{\mathcal{N}}-1$ models achieve the number of cache hits $w'_{\mathcal{N}}\delta_{m,\mathcal{N}}-u\left(m,e'_{\mathcal{N}}\right)$. Therefore, $\hat{{\bf{X}}}_{m,{\mathcal{N}}}=\hat{{\bf{X}}}_{m,{\mathcal{N}}}\cup\{x_{m,i}\}$, where model $i$ is the $e'_{\mathcal{N}}$-th model in $\mathcal{I}_\mathcal{N}$. 
After processing model $e_{\mathcal{N}}$, the algorithm repeats the same steps above to process the $e_{\mathcal{N}}-1$-th model until all models in $\mathcal{I}_\mathcal{N}$ have been checked. Note that there may be more than one feasible model caching decision for $\mathcal{I}_{\mathcal{N}}$ that can achieve the same number of cache hits $\dot{w}^*_{\mathcal{N}}\dot{\delta}_{m,\mathcal{N}}$ with the data size $\mathcal{T}\left(\left|\mathcal{I}_\mathcal{N}\right|,w^*_{\mathcal{N}}\right)$. Although Algorithm \ref{algorithm_recursive} produces one feasible $\hat{{\bf{X}}}_{m,{\mathcal{N}}}$, this does not affect the optimality.

% Next, with Algorithm \ref{algorithm_recursive}, after $\mathcal{N}^*$ is determined, $\hat{{\bf{X}}}_{m,{\mathcal{N}^*}}$ can be obtained. 
\begin{figure}[!t]
\centerline{\includegraphics[width=0.48\textwidth]{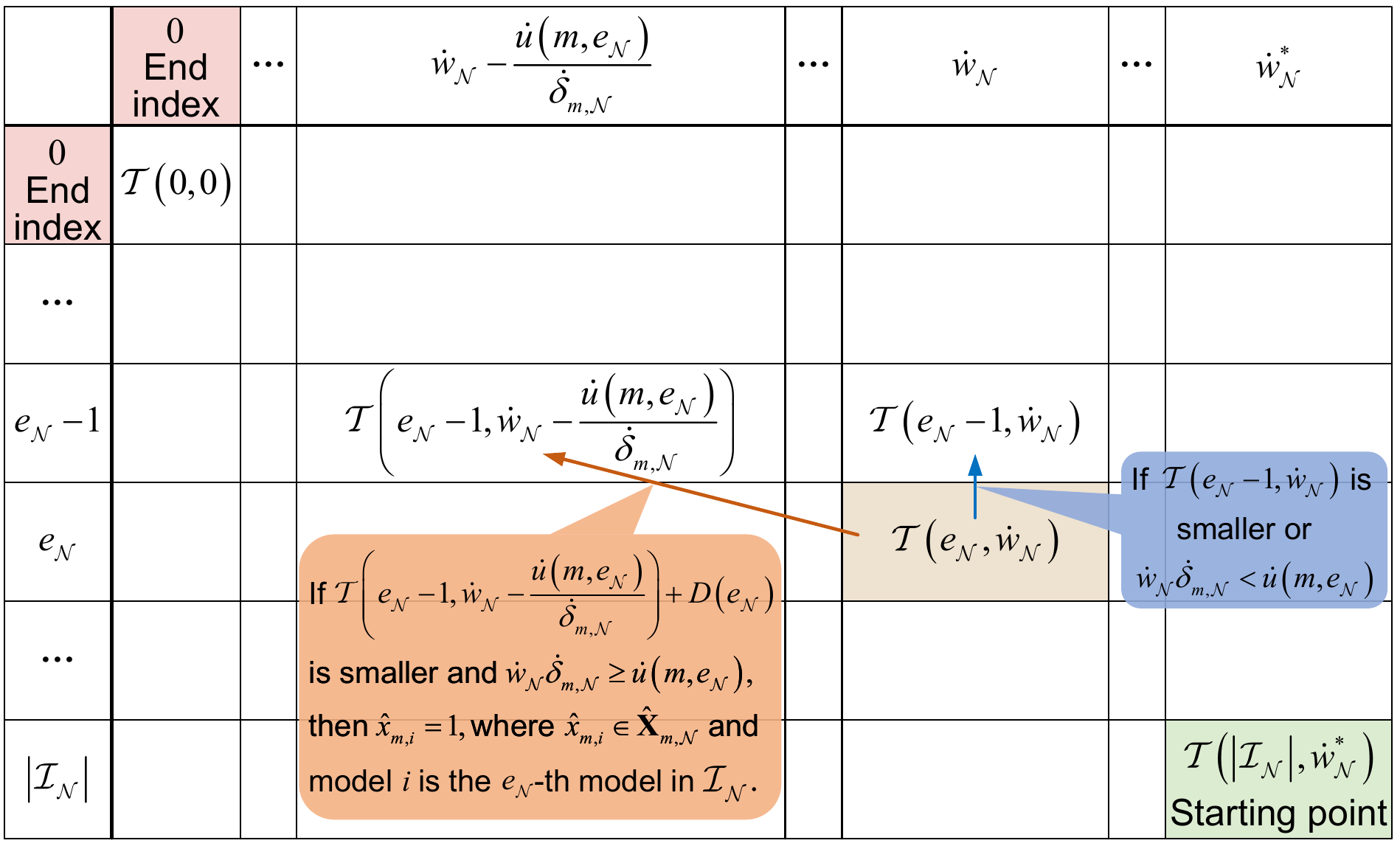}}
	\caption{The process of determining $\hat{{\bf{X}}}_{m,{\mathcal{N}}}$, where the first row and the first column are the same as those in Fig. \ref{fig_dp}.
 }
	\label{fig_recursive}
\end{figure}

\begin{algorithm}[!t]
	\caption{Model Caching Decision Algorithm} %算法的名字
	\label{algorithm_recursive}
	\LinesNumbered
	\KwIn{$m$, $\mathcal{I}_{\mathcal{N}}$, $\dot{\delta}_{m,\mathcal{N}}$, $\dot{u}\left(m,e_{\mathcal{N}}\right)$, and $\mathcal{T}\left(e_{\mathcal{N}},\dot{w}_{\mathcal{N}}\right)$.}
	\KwOut{$\hat{{\bf{X}}}_{m,\mathcal{N}}$.} %算法的结果输出
	{\bf Initialize:} $e_{\mathcal{N}}= \left|\mathcal{I}_{\mathcal{N}}\right|$, $\dot{w}_{\mathcal{N}} = \dot{w}^*_{\mathcal{N}} $, and $\hat{{\bf{X}}}_{m,{\mathcal{N}}}={\bf{0}}$\\
	\While{$e_{\mathcal{N}}\ne 0 $ and $\dot{w}_{\mathcal{N}}\ne 0 $}
        {
            \If{$\dot{w}_{\mathcal{N}}\dot{\delta}_{m,\mathcal{N}}\ge \dot{u}\left(m,e_{\mathcal{N}}\right)$ }
            {
                \If{$\mathcal{T}\left(e_{\mathcal{N}}-1,\dot{w}_{\mathcal{N}}-\frac{\dot{u}\left(m,e_{\mathcal{N}}\right)}{
            \dot{\delta}_{m,\mathcal{N}}}\right)+D\left(e_{\mathcal{N}}\right)<\mathcal{T}\left(e_{\mathcal{N}}-1,\dot{w}_{\mathcal{N}}\right)$}
                {
                    $\hat{x}_{m,i}=1$, where $\hat{x}_{m,i}\in\hat{{\bf{X}}}_{m,{\mathcal{N}}}$ and model $i$ is the $e_{\mathcal{N}}$-th model in $\mathcal{I}_\mathcal{N}$.\label{line:recursive_x}\\
                    $\dot{w}_{\mathcal{N}}=\dot{w}_{\mathcal{N}}-\frac{\dot{u}\left(m,e_{\mathcal{N}}\right)}{\dot{\delta}_{m,\mathcal{N}}}$.\\
                }
            }
            $e_{\mathcal{N}}=e_{\mathcal{N}}-1$.\\
        }
\end{algorithm}\par 

\subsubsection{$\epsilon$-optimal solution to $\mathcal{P}2.1_{m}$ and rounding DP-based algorithm outline}
After traversing all feasible $\mathcal{N}$ in $\mathcal{A}$, the maximum number of cache hits of models on edge server $m$ corresponds to the index $\dot{w}^*_{\mathcal{N}^*}$, where $\mathcal{N}^*=\mathop{\arg\max}\limits_{\mathcal{N}}\{\dot{w}^*_{\mathcal{N}}\}$. Letting $\hat{{\bf{X}}}_{m,\mathcal{N}^{*}}$ be the solution $\hat{{\bf{X}}}_{m}$ to $\mathcal{P}2.1_{m}$, the maximum cache hit ratio of edge server $m$ is given by
\begin{equation}\label{eq_rounded_u_star}
	\hat{U}_m\left(\hat{{\bf{X}}}_{m}\right)=\hat{U}_m\left(\hat{{\bf{X}}}_{m,\mathcal{N}^{*}}\right) = \frac{\sum\limits_{i\in\hat{\mathcal{I}}_{m}}{u}\left(m,i\right)}{\sum\limits_{k\in\mathcal{K}}\sum\limits_{i\in\mathcal{I}}{p}_{k,i}},
\end{equation} 
where $\hat{\mathcal{I}}_{m}=\hat{\mathcal{I}}_{m,\mathcal{N}^{*}}$.
%where $\hat{{\bf{X}}}_{m,\mathcal{N}^{*}}$ is the model caching decision corresponding to $\dot{w}^*_{\mathcal{N}^*}$, which will be elaborated in the following paragraphs, and $\dot{\mathcal{I}}_{m,\mathcal{N}^{*}} = \left\{i\mid\hat{x}_{m,i}\in\hat{{\bf{X}}}_{m,\mathcal{N}^{*}}\right\}$. Note that in the above process, the DP-based algorithm executes based on $\dot{u}\left(m,i\right)$, %and the realized maximum number of rounded cache hits is $\sum\limits_{i\in\dot{\mathcal{I}}_{m}}\dot{u}\left(m,i\right)$. However, 
%while we use $u\left(m,i\right)$ to calculate $\hat{U}_m\left(\hat{{\bf{X}}}_m\right)$ in \eqref{eq_rounded_u} at last. 

The rounding DP-based algorithm is outlined in Algorithm \ref{algorithm_DP}. 
\begin{algorithm}[!t]
	\caption{Rounding DP-based Algorithm} %算法的名字
	\label{algorithm_DP}
	\LinesNumbered
	\KwIn{$m$, $\epsilon$, $\mathcal{K}$, $\mathcal{I}$.}
	\KwOut{$\hat{U}_{m}\left(\hat{{\bf{X}}}_m\right)$, $\hat{{\bf{X}}}_m$.} %算法的结果输出
	{\bf Initialize:} $\hat{U}_{m}\left(\hat{{\bf{X}}}_m\right) = 0$, $\hat{{\bf{X}}}_m={\bf{0}}$, $\mathcal{N}^{*}=\emptyset$, and $\hat{U}_{m}\left(\hat{{\bf{X}}}_{m,\mathcal{N}^{*}}\right)=0$. \\
    % Calculate $u\left(m,i\right)$ and $\dot{u}\left(m,i\right)$ with \eqref{eq_utility} and \eqref{eq_round_u}, respectively.\label{line:dp_u}\\
	Calculate $u\left(m,i\right)$, and $\dot{u}\left(m,i\right)$ with \eqref{eq_utility} and \eqref{eq_round_u}, respectively.\label{line:dp_u}\\
        \For{$\mathcal{N}\in\mathcal{A}$}
	{
%		\textbf{Traverse all potential combinations of shared parameters: }\\
		Calculate $d_{\mathcal{N}}$.\label{line:dp_d}\\
		\If{$d_{\mathcal{N}}>Q_m$}
		{
			\textbf{Continue}.\\
		}
     %        \eIf{$\epsilon = 0$}
	    %   {
		   %    $\dot{u}\left(m,i\right)=u\left(m,i\right)$.\\
	    % }
	    % {
		   %    $\dot{u}\left(m,i\right)=\lfloor\frac{u\left(m,i\right)}{\epsilon u_{m,\min}}\rfloor$.
	    % }
		Calculate $\dot{\delta}_{m,\mathcal{N}}$ of $\dot{u}\left(m,i\right)$ for models in $\mathcal{I}_{\mathcal{N}}$.\label{line:dp_delta}\\
		% \For{$e_{\mathcal{N}}\in\mathcal{I}_\mathcal{N}$}
		% {
		% 	Calculate $D_{\mathcal{N}}\left(e_{\mathcal{N}}\right)$ from \eqref{eq_size}. \\
		% }
            Calculate $D\left(e_{\mathcal{N}}\right)$ for all $e_{\mathcal{N}}$ in $\mathcal{I}_\mathcal{N}$.\label{line:dp_D}\\
		Initialize $\mathcal{T}\left(e_{\mathcal{N}},\dot{w}_{\mathcal{N}}\right)$ using \eqref{eq_dp_initial}.\\
        %, %$\mathcal{T}\left(e_{\mathcal{N}},\dot{w}_{\mathcal{N}}\right) = \left\{ {\begin{array}{*{20}{c}}
		% 		{\infty,\ \text{if } \dot{w}_{\mathcal{N}}\ne0,}\\
		% 		{0,\ \text{if }\dot{w}_{\mathcal{N}}=0.}
		% \end{array}} \right.$\\
		\For{$e_{\mathcal{N}}\in\left\{1,\dots,\left|\mathcal{I}_{\mathcal{N}}\right|\right\}$}
		{
			\For{$\dot{w}_{\mathcal{N}}\in\left\{1,\dots,\dot{W}_{m,\mathcal{N}}\right\}$}
			{
%				\textbf{Update the state-transition equation:}\\
                Calculate $\mathcal{T}\left(e_{\mathcal{N}},\dot{w}_{\mathcal{N}}\right)$ with \eqref{eq_dp}.\\
				% \eIf{$\dot{u}\left(m,e_{\mathcal{N}}\right)\le \dot{w}_{\mathcal{N}}\dot{\delta}_{m,\mathcal{N}}$}
				% {
    %                     $\mathcal{T}\left(e_{\mathcal{N}},\dot{w}_{\mathcal{N}}\right) = \min\left\{
				% 	\begin{aligned}
				% 		&\mathcal{T}\left(e_{\mathcal{N}}-1,\dot{w}_{\mathcal{N}}\right),\\
    %                         &\mathcal{T}\left(e_{\mathcal{N}}-1,\dot{w}_{\mathcal{N}}-\frac{\dot{u}\left(m,e_{\mathcal{N}}\right)}{\dot{\delta}_{m,\mathcal{N}}}\right) \\
    %                         &+ D_{\mathcal{N}}\left(e_{\mathcal{N}}\right)
				% 	\end{aligned}
    %                     \right\}
    %                     $.\\
				% }
				% {
				% 	$\mathcal{T}\left(e_{\mathcal{N}},\dot{w}_{\mathcal{N}}\right) = \mathcal{T}\left(e_{\mathcal{N}}-1,\dot{w}_{\mathcal{N}}\right)$. 
				% }
			}
		}
		Calculate $\dot{w}^*_{\mathcal{N}}$ with \eqref{eq_rounded_w}.\label{line:dp_w}\\
        % $=\mathop{\arg\max}\limits_{\dot{w}_{\mathcal{N}}}\{\dot{w}_{\mathcal{N}}\mid\mathcal{T}\left(\left|\mathcal{I}_{\mathcal{N}}\right|,\dot{w}_{\mathcal{N}}\right) \le Q_m-d_{\mathcal{N}}\}$. \\
        Calculate $\hat{{\bf{X}}}_{m,\mathcal{N}}$ with Algorithm \ref{algorithm_recursive}. \label{line:dp_x}\\
        % $\hat{U}_m\left(\hat{{\bf{X}}}_{m,\mathcal{N}}\right)=\frac{\sum\limits_{i\in\dot{\mathcal{I}}_{m,\mathcal{N}}}u\left(m,i\right)}{\sum\limits_{k\in\mathcal{K}}\sum\limits_{i\in\mathcal{I}}{p}_{k,i}}$, where $\dot{\mathcal{I}}_{m,\mathcal{N}} = \left\{i\mid\hat{x}_{m,i}\in\hat{{\bf{X}}}_{m,\mathcal{N}}\right\}$.\\
        Calculate $\hat{U}_m\left(\hat{{\bf{X}}}_{m,\mathcal{N}}\right)$ with \eqref{eq_rounded_u}.\label{line:dp_U}\\
		\If{$\hat{U}_m\left(\hat{{\bf{X}}}_{m,\mathcal{N}}\right)>\hat{U}_{m}\left(\hat{{\bf{X}}}_{m,\mathcal{N}^{*}}\right)$}
		{
			%$\hat{U}_{m}\left(\hat{{\bf{X}}}_m\right)=\hat{U}_m\left(\hat{{\bf{X}}}_{m,\mathcal{N}}\right)$, and $\hat{{\bf{X}}}_m=\hat{{\bf{X}}}_{m,\mathcal{N}}$.\\
            $\mathcal{N}^{*}=\mathcal{N}$.\\
		}
	}
    $\hat{{\bf{X}}}_{m}=\hat{{\bf{X}}}_{m,\mathcal{N}^{*}}$.\\
\end{algorithm}
Furthermore, we establish the following proposition for Algorithm \ref{algorithm_DP}.
\begin{proposition}\label{DP_epsilon}
	The cache hit ratio produced by Algorithm \ref{algorithm_DP} satisfies $\hat{U}_{m}\left(\hat{{\bf{X}}}_m\right)\ge\left(1-\epsilon\right)\hat{U}_{m}\left(\hat{{\bf{X}}}^*_m\right)$, where $\hat{{\bf{X}}}^*_m$ is the optimal solution to $\mathcal{P}2.1_{m}$.
\end{proposition}
%\begin{proof}
%    The proof is omitted here. A similar proof can be found in \cite{sindelar2011sharing}.
%\end{proof}
 \begin{proof}
 	The proof is shown in Appendix \myref{proof_DP_epsilon} in our supplementary material.
 \end{proof}

\subsection{Analysis of the TrimCaching Spec Algorithm}
\subsubsection{Time complexity analysis}
With Proposition \ref{DP_epsilon}, we first establish the following theorem on the time complexity of the proposed TrimCaching Spec Algorithm under the special case.
\begin{theorem}\label{theorem_spec_time}
    Under Assumption \ref{assumption_1} for the special case of $\mathcal{P}1.1$, where the number of shared parameter blocks is small and fixed, the TrimCaching Spec algorithm has a polynomial-time computational complexity $O\left(MI\right)$.
\end{theorem}
\begin{proof}
The proof is presented in Appendix \myref{proof_theorem_spec_time} in our supplementary material.
\end{proof}
In the special case, the complexity of the TrimCaching Spec algorithm can be further reduced if models are fine-tuned with the bottom-layer freezing technique, which involves freezing a number of the bottom layers of a pre-trained model and only updating the top layers. This approach is commonly used in transfer learning. In this case, we have the following corollary. 
\begin{corollary}\label{corollary_1}
    % When models in the model library are fine-tuned with the bottom layer freezing technique, the total complexity of the TrimCaching Spec algorithm can be significantly reduced, moving from $O\left(2^{\left|\mathcal{J}^{\text{sh}}\right|}\left(\frac{{p}}{\delta_{\min}}+1\right)MI\right)$ down to $O\left(\left(\kappa+1\right)\left(\frac{{p}}{\delta_{\min}}+1\right)MI\right)$, where $\kappa$ is the number of all shared layers, ${p} = \sum\limits_{k\in\mathcal{K}}\sum\limits_{i\in\mathcal{I}}{p}_{k,i}$, and $\delta_{\min}$ is the granularity of $p_{k,i}$.
    When models in the model library are fine-tuned from a single pre-trained model using the bottom-layer freezing technique, where each model shares a sequence of consecutive bottom layers from the pre-trained model, the time complexity of the TrimCaching Spec algorithm can be significantly reduced from $O\left(2^{\left|\mathcal{J}^{\text{sh}}\right|}MI\right)$ to $O\left(\left(\kappa+1\right)MI\right)$, where $\kappa$ denotes the maximum number of bottom layers shared between any model in the library and the pre-trained model.
\end{corollary}
\begin{proof}
    The proof is provided in Appendix \myref{proof_corollary_1} in our supplementary material.
\end{proof} 

\subsubsection{Approximation guarantee analysis}Next, we analyze the approximation guarantee of the proposed TrimCaching Spec Algorithm. The details are summarized in the following theorem.
\begin{theorem}\label{theorem_spec_final}
The TrimCaching Spec algorithm obtains a solution $\hat{{\bf{X}}}$ to $\mathcal{P}1.1$ satisfying $U\left(\hat{{\bf{X}}}\right)\ge\frac{1-\epsilon}{2}U\left({\bf{X}}^*\right)$. 
\end{theorem}
\begin{proof}
    The proof is provided in Appendix \myref{proof_theorem_spec_final} in our supplementary material.
\end{proof}

It is noted that, when $M=1$, $\mathcal{P}1.1$ is equivalent to $\mathcal{P}2.1_m$. In this case, the TrimCaching Spec algorithm can directly use Algorithm \ref{algorithm_DP} to solve the problem, leading to the following corollary.
\begin{corollary}\label{proposition_spec_same}
	% In the single-edge scenario (M=1), the TrimCaching Spec algorithm obtains solution $\bf{X}$ to $\mathcal{P}1.1$, which satisfies $U\left(\bf{X}\right)\ge \left(1-\epsilon\right)U\left({\bf{X}^*}\right)$.
    In the single-edge scenario (M=1), the solution to $\mathcal{P}1.1$ produced by the TrimCaching Spec algorithm satisfies $U\left(\hat{\bf{X}}\right)\ge \left(1-\epsilon\right)U\left({\bf{X}^*}\right)$.
 \begin{proof}
The proof directly follows from Proposition \ref{DP_epsilon}. 
 \end{proof}
\end{corollary}
% \begin{remark}
%     \rev{Both Theorem \ref{theorem_spec_final} and Corollary \ref{proposition_spec_same} hold without Assumption \ref{assumption_1}, which only affects the time complexity analysis and does not impact the approximation guarantee of the proposed TrimCaching Spec Algorithm.}
% \end{remark}
% \begin{corollary}\label{theorem1}
% The number of {training rounds $R$} for achieving target convergence accuracy $\varepsilon$ satisfies
% { \begin{equation}\label{accuracy_cons_corollary}
% \frac{1}{R} \sum_{t=1}^{R} \mathbb{E} [\Vert \nabla_{\bf{w}} f({\mathbf{w}}^{t-1})\Vert^{2}] \le \varepsilon ,
% \end{equation}}
% is given by
% { \begin{equation}\label{lowest_com_num}
% R  \ge \frac{2 \vartheta}{{\gamma \bigg( {\varepsilon  - \frac{{\beta \gamma \sum\limits_{j = 1}^L {\sigma _j^2}}}{N} - 4{\beta ^2}{\gamma ^2}{I^2}\sum\limits_{j = 1}^{{L_c}} {G_j^2}} \bigg)}}.
% \end{equation}}
% \end{corollary}

\section{The General Case: Arbitrary Parameter Sharing}
This section examines the general case of $\mathcal{P}1.1$, where models can arbitrarily share parameter blocks, in contrast to the special case with a small fixed number of shared parameter blocks. Formally speaking, the number of shared parameter blocks among models may increase with the scale of the model library, such that searching for all combinations of shared parameter blocks, as required by the TrimCaching Spec algorithm, results in exponential time complexity in terms of the problem scale, which should be avoided in the general case.

%\subsection{Greed is still good}
To solve $\mathcal{P}1.1$ in the general case, we propose a greedy-based algorithm, TrimCaching Gen, outlined in Algorithm~\ref{algorithm_greedy}. The detailed procedures of Algorithm \ref{algorithm_greedy} are as follows. First, in the $l$-th step, based on the placement decision ${\bf{X}}^{l-1}$ determined in the $l-1$-th step, Algorithm \ref{algorithm_greedy} computes the marginal increase in the cache hit ratio of placing model $i$ on edge server $m$ in Line \ref{line:greedy_delta}, which is expressed as 
\begin{equation}\label{eq_delta_u}
    \Delta U\left({\bf{X}}^{l-1},m,i\right) =  U\left(\rho\left({\bf{X}}^{l-1},m,i\right)\right)-U\left({\bf{X}}^{l-1}\right),
\end{equation}
where
\begin{equation}\label{eq_rho}
    \begin{aligned}
        &\rho\left({\bf{X}}^{l-1},m,i\right)\\
        &=
    \begin{cases}
        {\bf{X}}^{l-1}\setminus \left\{x_{m,i}=0\right\} \cup \left\{x_{m,i}=1\right\}, \text{ if }x_{m,i}=0,\\
        {\bf{X}}^{l-1}, \text{ if }x_{m,i}=1,
    \end{cases}
    \end{aligned}
\end{equation}
represents $x_{m,i}$ in ${\bf{X}}^{l-1}$ is updated to 1 if it was 0, and remains unchanged otherwise. Second, in Line \ref{line:greedy_optimal_m_i}, the algorithm identifies $\{m^*,i^*\}$ which yields the maximum marginal increase in the cache hit ratio, while ensuring that $g_{m^{*}}\left(\rho\left({\bf{X}}_{m^{*}}^{l-1},m^{*},i^{*}\right)\right)$ does not exceed the edge server capacity, where ${\bf{X}}_m^{l-1}$ is the model placement decision of edge server $m$ in the $l-1$-th step based on ${\bf{X}}^{l-1}$, with ${\bf{X}}^{l-1}=\bigcup\limits_{m\in\mathcal{M}}{\bf{X}}^{l-1}_m$. Third, in Line~\ref{line:greedy_x_l}, ${\bf{X}}^{l-1}$ is updated to ${\bf{X}}^l=\rho\left({\bf{X}}^{l-1},m^{*},i^{*}\right)$. Finally, the algorithm repeats the above steps until no feasible pair $\{m,i\}$ can be found that improves the cache hit ratio without violating any capacity constraint.

Next, we establish the following theorems for the TrimCaching Gen algorithm.

\begin{theorem}
    The proposed algorithm solves $\mathcal{P}1.1$ with the time complexity of $O\left(M^2I^2\right)$. 
\end{theorem}
\begin{proof}
    The proof is similar to that of Theorem~\ref{theorem_spec_time}, which is omitted here.
\end{proof}

\begin{algorithm}[!t]
	\caption{TrimCaching Gen Algorithm} %算法的名字
	\label{algorithm_greedy}
	\LinesNumbered
	\KwIn{$\mathcal{K}$, $\mathcal{I}$, and $\mathcal{M}$.}
	\KwOut{${\bf{X}}$ and $U\left({\bf{X}}\right)$.} %算法的结果输出
	{\bf Initialize:}
	$l=0$, ${\bf{X}}^l={\bf{0}}$, and $U\left({\bf{X}}^l\right)=0$.\\
	\While {%There is an edge server that can continue to cache models.
    There exists $\left\{m,i\right\}$ such that $g_m\left(\rho\left({\bf{X}}_{m}^{l},m,i\right)\right)\le Q_m$ and $\Delta U\left({\bf{X}}^{l},m,i\right)\ne 0$.}
	{
		$l=l+1$.\\
		%\textbf{The $l$-th step:}\\
            Calculate $\Delta U\left({\bf{X}}^{l-1},m,i\right)$ with \eqref{eq_delta_u} for $m\in\mathcal{M}$ and $i\in\mathcal{I}$.\label{line:greedy_delta}\\
			$\left\{m^*,i^*\right\} = \mathop{\arg\max}\limits_{m,i}\{\Delta U\left({\bf{X}}^{l-1},m,i\right) \mid \ g_m\left(\rho\left({\bf{X}}_{m}^{l-1},m,i\right)\right)\le Q_m\}$. \label{line:greedy_optimal_m_i}\\
		${\bf{X}}^l=\rho\left({\bf{X}}^{l-1},m^{*},i^{*}\right)$.\label{line:greedy_x_l}\\
	}
	${\bf{X}}={\bf{X}}^l$.\\
\end{algorithm}\par 
\begin{theorem}
%Suppose the solution of the TrimCaching Gen algorithm is ${\bf{X}}=\bigcup\limits_{m\in\mathcal{M}}{\bf{X}}_m$, where ${\bf{X}}_m$ is the model placement decision of edge server $m$. The TrimCaching Gen algorithm can obtain a result which is lower bounded by $U\left({\bf{X}}\right)\ge \left[1-\left(\frac{\Gamma-1}{\Gamma}\right)^{\gamma}\right]U\left({\bf{X}}^*\right)\ge\frac{1}{\Gamma}U\left({\bf{X}}^*\right)$. Here, $\Gamma=\max\{\left|{\bf{X}}\right|:g_m\left({\bf{X}}_m\right)\le Q_m,\forall m\in\mathcal{M}\}$, $\gamma=\min\{\left|{\bf{X}}\right|:g_m\left({\bf{X}}_m\right)\le Q_m\ \& \ g_m\left({\bf{X}}_m\cup\left\{x_{m,i}\right\}\right)\ge Q_m, \forall x_{m,i}\in{\bf{V}}\setminus{\bf{X}}_m,\forall m\in\mathcal{M}\}$, and ${\bf{V}}$ is the ground set of $x_{m,i}$.
The TrimCaching Gen algorithm obtains a solution ${\bf{X}}$ to $\mathcal{P}1.1$ satisfying $U\left({\bf{X}}\right)\ge\frac{1}{\Gamma}U\left({\bf{X}}^*\right)$, where $\Gamma=\max\{\left|{\bf{X}}\right| \mid g_m\left({\bf{X}}_m\right)\le Q_m,\forall m\in\mathcal{M}\}$, and ${\bf{X}}_m$ is the model placement decision of edge server $m$ in ${\bf{X}}$ with ${\bf{X}}=\bigcup\limits_{m\in\mathcal{M}}{\bf{X}}_m$. 
\end{theorem}
\begin{proof}
The proof is omitted here. A similar proof can be found in \cite{iyer2013submodular,iyer2013fast}.
\end{proof}
\begin{remark}
    Since the lower bound decreases with the increasing number of AI models and edge servers, there is no constant approximation guarantee for Algorithm \ref{algorithm_greedy}. This coincides with Proposition \ref{NPhard}, which states that no polynomial-time approximation algorithm exists for the general case.
\end{remark}

% However, it is still a tight bound. Let us take $I=4$, $M=2$, $Q_1=Q+\Delta$, and $Q_2=\Delta$ as an example, where $\Delta$ is a tiny value. Suppose that the model sizes of models 1 and 4 are $Q$, and model $1$ does not share any parameter block with the other 3 models. Suppose the total model size of models 2 and 3 are $Q$, and all parameter blocks of model 4 are included in models 2 and 3. Besides, we assume that the latency in transferring model 1 from edge server 1 to edge server 2 does not meet the transmission latency QoS requirement of users associated with edge server 2, while models 2-4 can be transferred in time. Let the expected cache hits of model 1 be $u$ when model 1 is placed in edge server 1. We also suppose that the expected cache hits of models 2-4 achieved in each edge server are $u$, and the coverage areas of the two edge servers do not overlap. Therefore, the optimal expected number of cache hits is $6u$, where models 2 and 3 are placed on edge server 1, model 4 can be composed of parameters of models 2 and 3, and models 2-4 can be transferred from edge server 1 to 2. However, the TrimCaching Gen algorithm can place model 1 on edge server 1 based on line 4 in Algorithm \ref{algorithm_greedy} in the first step. Therefore, the output expected number of cache hits is only $u$, which is $\frac{1}{6}$-times the optimal result.
\par

\section{Numerical Results}
In this section, we present simulation results for the proposed TrimCaching framework.
\subsection{Simulation Setup}
In the simulation, $K$ users and $M$ edge servers are uniformly distributed in a square area of 1 $\text{km}^\text{2}$. The latency requirements of users for downloading models and performing on-device inference are uniformly distributed over the interval $\left[0.5,1\right]\ \text{s}$ \cite{3gpp.22.874}. The number of users is set to $K=\left\{10,20,30,40,50\right\}$. Each edge server has a coverage radius of 275 m, and the associated user set of edge server $m$ is denoted as $\mathcal{K}_m$. The expected bandwidth and transmit power allocated by edge server $m$ to user $k$ in $\mathcal{K}_m$ are $\bar{B}_{m,k} = \frac{B}{{p}_\text{A}\left|\mathcal{K}_m\right|}$ and $\bar{P}_{m,k} = \frac{P}{{p}_\text{A}\left|\mathcal{K}_m\right|}$, respectively. Here, $B=400$~MHz and $P=43$~dBm~\cite{3gpp.38.104} represent the total bandwidth and transmit power of an edge server, respectively. The user active probability is set to ${p}_\text{A}=0.5$. The communication data rate between edge servers is set to $C_{m,m'}=10\ {\text{Gbps}}$~\cite{gsmabachkaul,8488527}. Besides, we set $\gamma_0$ and $\alpha_0$ in \eqref{eq_communication} to 1 and 4~\cite{wu2024efficient,8886730}, %more: 7414036
respectively. The number of edge servers is set to $M=\left\{6,8,10,12,14\right\}$. Besides, the storage capacity of each edge server is denoted by $Q_m = Q$, which is uniformly set across all servers and varies from 0.5 GB to 5 GB. It is noted that the storage capacity of edge servers can be much larger than 5 GB in reality. However, the model library can also be significantly larger than our constructed parameter-sharing model library, which will be introduced later. Due to our limited computing resources for model fine-tuning, we proportionally reduce the storage capacity of edge servers and the size of the model library, which will not impact the phenomenon observed in the experiments. 

We construct two parameter-sharing model libraries, each derived from a different model family. The first model library is based on the ResNet family, including ResNet-18, ResNet-34, and ResNet-50~\cite{he2016deep}. %For example, the superclass ``fish" includes class ``aquarium fish", ``flatfish", ``ray", ``shark", and ``trout". 
% We fine-tune these three pre-trained models into 20 downstream classification models, respectively, with the sub-dataset of the same superclass. Each model corresponds to each superclass, e.g., a classification model for fishes. 
The second model library is based on the GPT family, including GPT-2 small, medium, and large. 
The parameter sharing among AI models in the special and general cases is constructed as follows. 

\begin{itemize}
    \item \textbf{Special case}: In this case, we design parameter sharing in both the ResNet- and the GPT-2-based model libraries. In the ResNet-based model library, we fine-tune the three pre-trained ResNet models on CIFAR100~\cite{krizhevsky2009learning}. For each model structure, we obtain 100 downstream models, each corresponding to one class in CIFAR100. %, e.g., a classification model for sharks. 
    We adopt the bottom-layer freezing techniques to fine-tune these models. The number of frozen bottom layers (each corresponding to a shared parameter block in our system model) for ResNet-18, ResNet-34, and ResNet-50 falls within the ranges of $\left[29,40\right]$, $\left[49,72\right]$, and $\left[87,106\right]$, respectively. 
    
    For the GPT-2-based model library, we fine-tune the three pre-trained GPT-2 models on datasets, including E2E~\cite{DBLP:journals/corr/NovikovaDR17}, WebNLG~\cite{gardent2017creating}, and DART~\cite{nan2021dart}, to create 100 downstream models for each model structure. Moreover, we adopt LoRA to fine-tune these models, where downstream models fine-tuned from the same pre-trained model share its parameters. The sizes of newly introduced trainable parameters by LoRA for the GPT-2 small, medium, and large are:  \{0, 188,348, 335,804, 630,716, 1,220,540, 2,400,188, 4,759,484 \} Bytes, \{0, 470,487, 863,703, 1,650,135, 3,222,999, 6,368,727, 12,660,247\} Bytes, and \{0, 852,175, 1,589,455, 3,064,015, 6,013,135, 11,911,375, 23,708,047\} Bytes, respectively. Additionally, the GPT-2-based model library also includes models fine-tuned via the full-parameter fine-tuning approach.

     \item  \textbf{General case}: 
     In this case, we design parameter sharing in the ResNet-based model library. To support arbitrary parameter sharing among models in the library (shared parameter blocks do not simply come from a small set of pre-trained models, as in the special case, and the number of shared blocks increases as the library grows), we adopt a two-stage fine-tuning strategy. 
     In the first stage, we select multiple superclasses from CIFAR100. For each selected superclass, all parameters of the pre-trained ResNet-18, ResNet-34, and ResNet-50 are fine-tuned using data from all its classes, resulting in one distinct model per structure for each superclass. 
     In the second stage, each first-stage fine-tuned model is further fine-tuned on datasets from classes within a new superclass that is semantically similar to the one used in the first stage. The pairing of superclasses between stage 1 and stage 2 is provided in Table~\ref{table_ft}. Moreover, the second-stage fine-tuning is conducted using bottom-layer freezing, without restricting the number of frozen layers, where each first-stage model is fine-tuned separately on the data of each class, resulting in a total of 45 models.
     % In this case, we first fine-tune all parameters of the ResNet18, ResNet34, and ResNet50 on CIFAR100 with a few selected superclasses. Then, we fine-tune new models on sub-datasets in CIFAR100 with similar classes based on bottom layer freezing. The fine-tuned models for each class within every superclass in the second round reuse parameter blocks from models fine-tuned for the selected superclass in the first round, creating a large set of shared parameter blocks related to the scale of the model library. The fine-tuning settings are detailed in Table \ref{table_ft}.
     %Taking the first row as an example, we first fine-tune all layers of ResNet18, ResNet34, and ResNet50 on the sub-dataset ``fruit and vegetables''. Then, we fine-tune the newly trained model on the sub-dataset ``flowers'' and ``trees'' with the bottom layer freezing method. The number of frozen bottom layers for the ResNet18, ResNet34, and ResNet50 for all sub-datasets in the second FT falls within the ranges of $\left[31,40\right]$, $\left[55,72\right]$, and $\left[93,106\right]$, respectively.
\end{itemize}
Moreover, the on-device inference latency for the ResNet-based and the GPT-2-based model libraries ranges from 1 ms to 5 ms and from 40 ms to 200 ms, respectively~\cite{3gpp.22.874}. The request probabilities of each user for the models in each model library follow a Zipf distribution~\cite{zipf1929relative}. 
\begin{table}[!t]
	\centering
	\caption{Superclass Selection in Stage 1 and Stage 2}
	\label{table_ft}
	\begin{tabular}{|p{3cm}|p{5cm}|}
		\hline
		\textbf{Stage 1} & \textbf{Stage 2}\\ \hline
		%fish & aquatic mammals\\ \hline
		fruit and vegetables & flowers, trees\\ \hline
		%food containers & household furniture\\ \hline
		medium-sized mammals & large carnivores, large omnivores and herbivores, people, reptiles, small mammals\\ \hline
		vehicles 2 & large man-made outdoor things, vehicles 1\\ \hline
	\end{tabular}
\end{table} \par 

For comparison, three algorithms are considered:
\begin{itemize}
    \item \textbf{TrimCaching Spec}: Algorithm \ref{algorithm_successive_greedy}, developed for the special case. 
     \item  \textbf{TrimCaching Gen}: Algorithm \ref{algorithm_greedy}, developed for the general case but also applicable to the special case.
    \item \textbf{Independent Caching:} AI models are cached independently without considering parameter sharing, referring to traditional content placement schemes \cite{8374917}. 
\end{itemize}
Moreover, $\epsilon$ is set to 0.1 in Algorithm \ref{algorithm_DP} by default. 

All simulation results are averaged from 100 network topologies. For each topology, we further evaluate the cache hit ratio using $10^3$ Rayleigh fading channel realizations. It is noted that the caching decisions are made based on average channel gains, while the performance is examined under Rayleigh fading.

%We perform $10^4$ Monte Carlo simulation runs for each network topology to ensure statistical significance. The cache hit ratio results presented in the figures reflect these cache hit ratios. In each time of Monte Carlo simulation, the downloading data rate when user $k$ downloads AI models from its associated edge server $m$ is $C_{m,k} = B_{m,k}{\rm{log}}_2\left(1+\frac{P_{m,k} \gamma_0  d_{m,k}^{-\alpha_0} \left|h_{m,k}\right|^2}{n_0B_{m,k}}\right)$, where $B_{m,k} = \frac{B}{\left|\mathcal{K}_{m,\text{A}}\right|}$ and $P_{m,k} = \frac{P}{\left|\mathcal{K}_{m,\text{A}}\right|}$ are the bandwidth and transmit power assigned by edge server $m$ for its associated user $k$, respectively, $\mathcal{K}_{m,\text{A}}$ is the active user set associated to edge server $m$, and $h_{m,k}$ is the small-scale fading variable over the Rayleigh channel between user $k$ and edge server $m$, where $\left|h_{m,k}\right|^2\sim \exp\left(1\right)$. 

%The parameter block sharing will be introduced later.
% The result shows that The TrimCaching Spec algorithm outperforms the other two benchmarks. 

\begin{figure*}[!t]
	\centering
	\subfigure[Cache hit ratio vs. $Q$, where $M=10$ and $K=30$. ]{\includegraphics[width=0.23\textwidth]{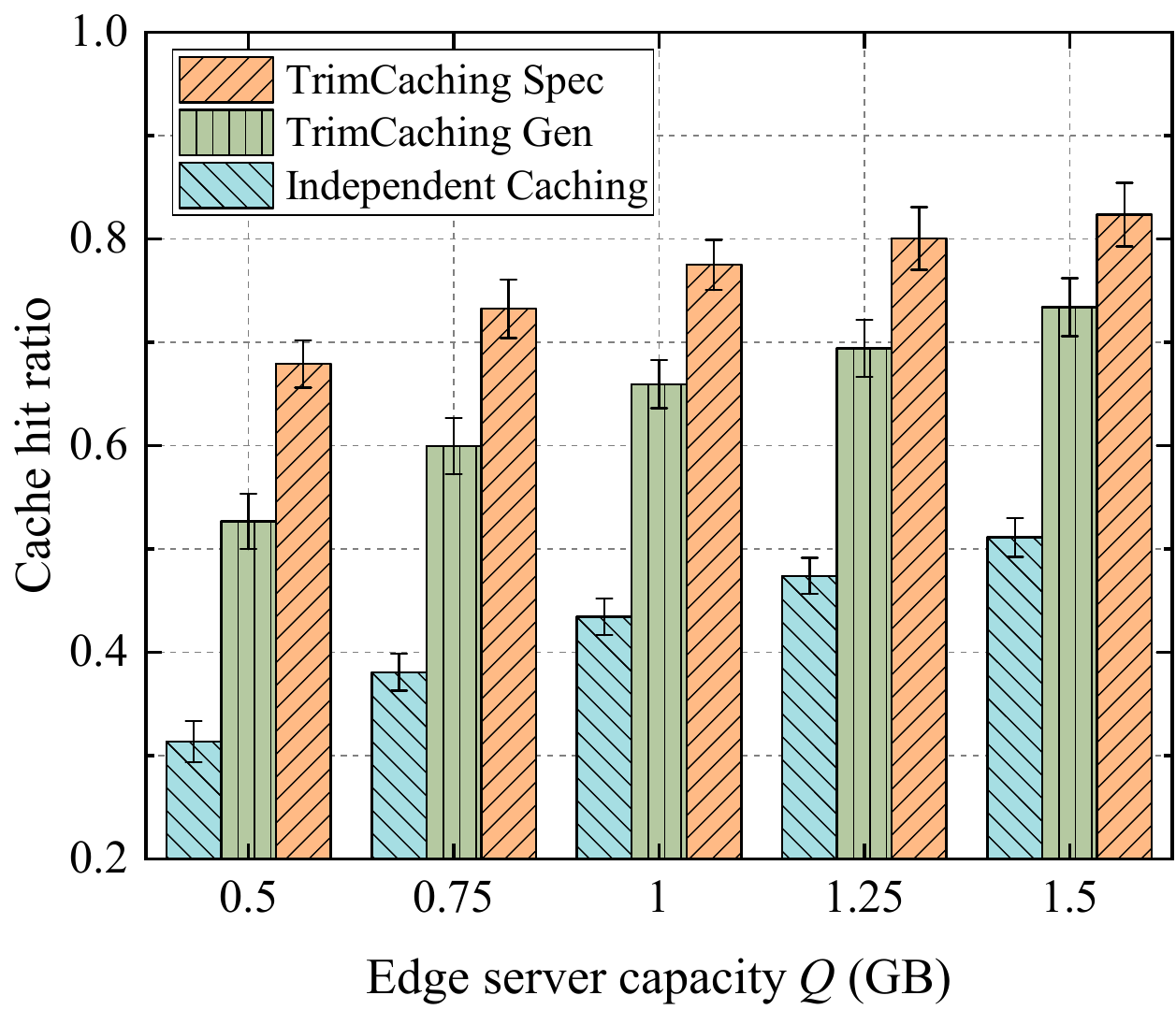}\label{fig_case2_capacity}}
	\quad
	\subfigure[Cache hit ratio vs. $M$, where $Q=1 \ \text{GB}$ and $K=30$.]{\includegraphics[width=0.23\textwidth]{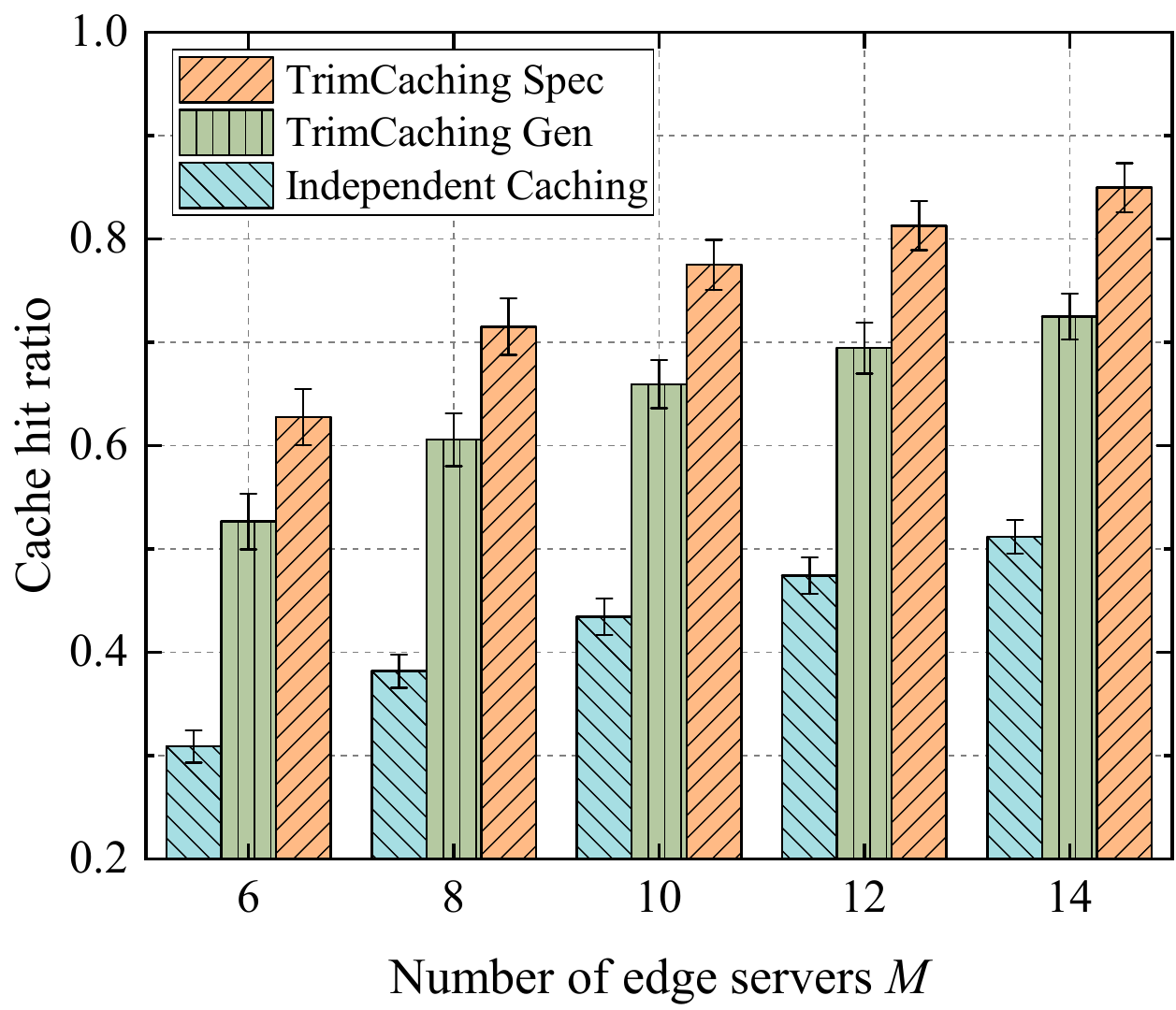}\label{fig_case2_edge}}
	\quad
	\subfigure[Cache hit ratio vs. $K$, where $Q=1 \ \text{GB}$ and $M=10$.]{\includegraphics[width=0.23\textwidth]{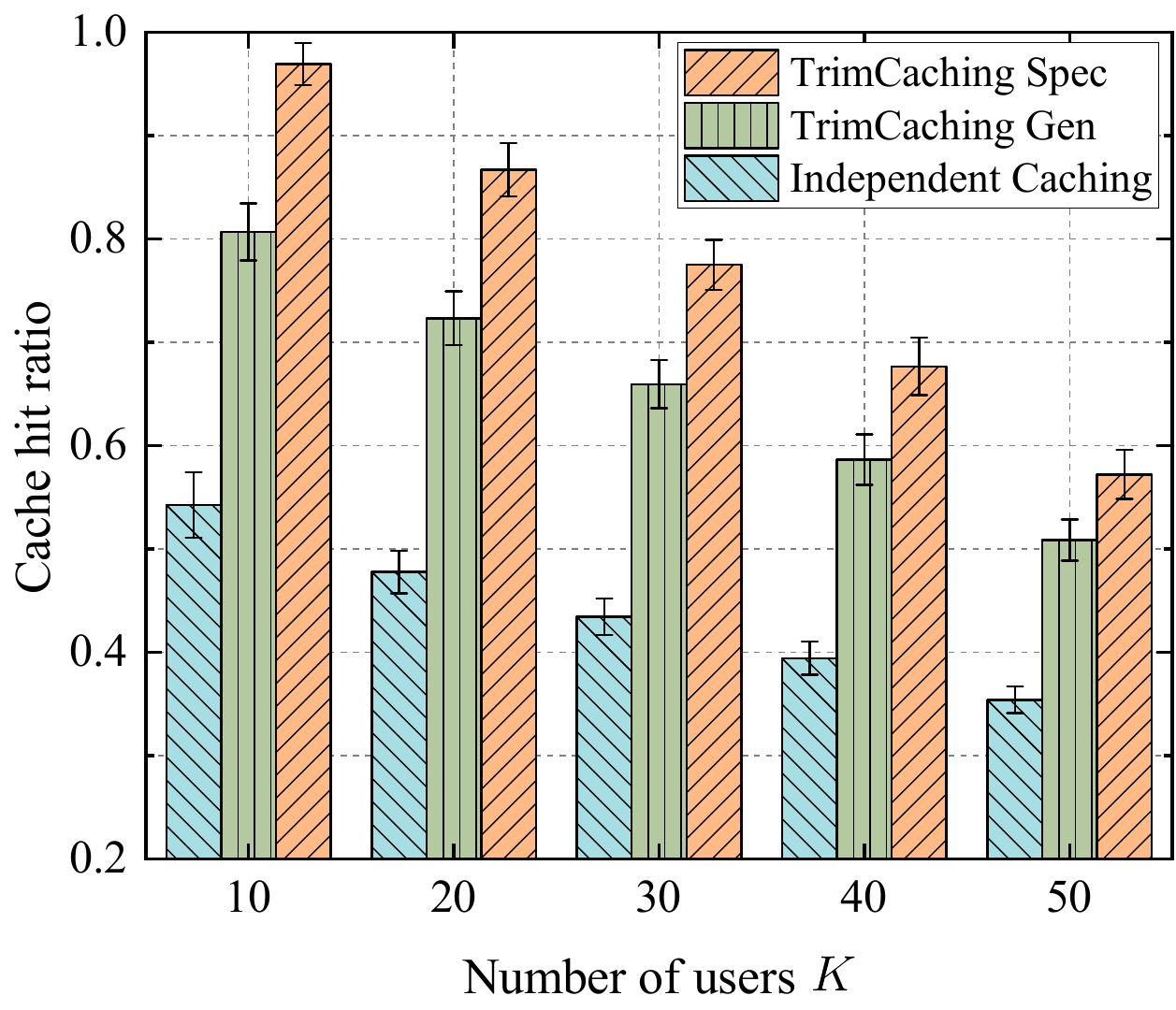}\label{fig_case2_user}}
	\caption{Cache hit ratio in the special case, where a small fixed number of shared parameter blocks is considered, using the ResNet-based model library. The error bar denotes the standard deviation, which is the same for the subsequent figures.}\label{fig_special_case}
\end{figure*}

\begin{figure*}[!t]
	\centering
	\subfigure[Cache hit ratio vs. $Q$, where $M=10$ and $K=30$. ]{\includegraphics[width=0.23\textwidth]{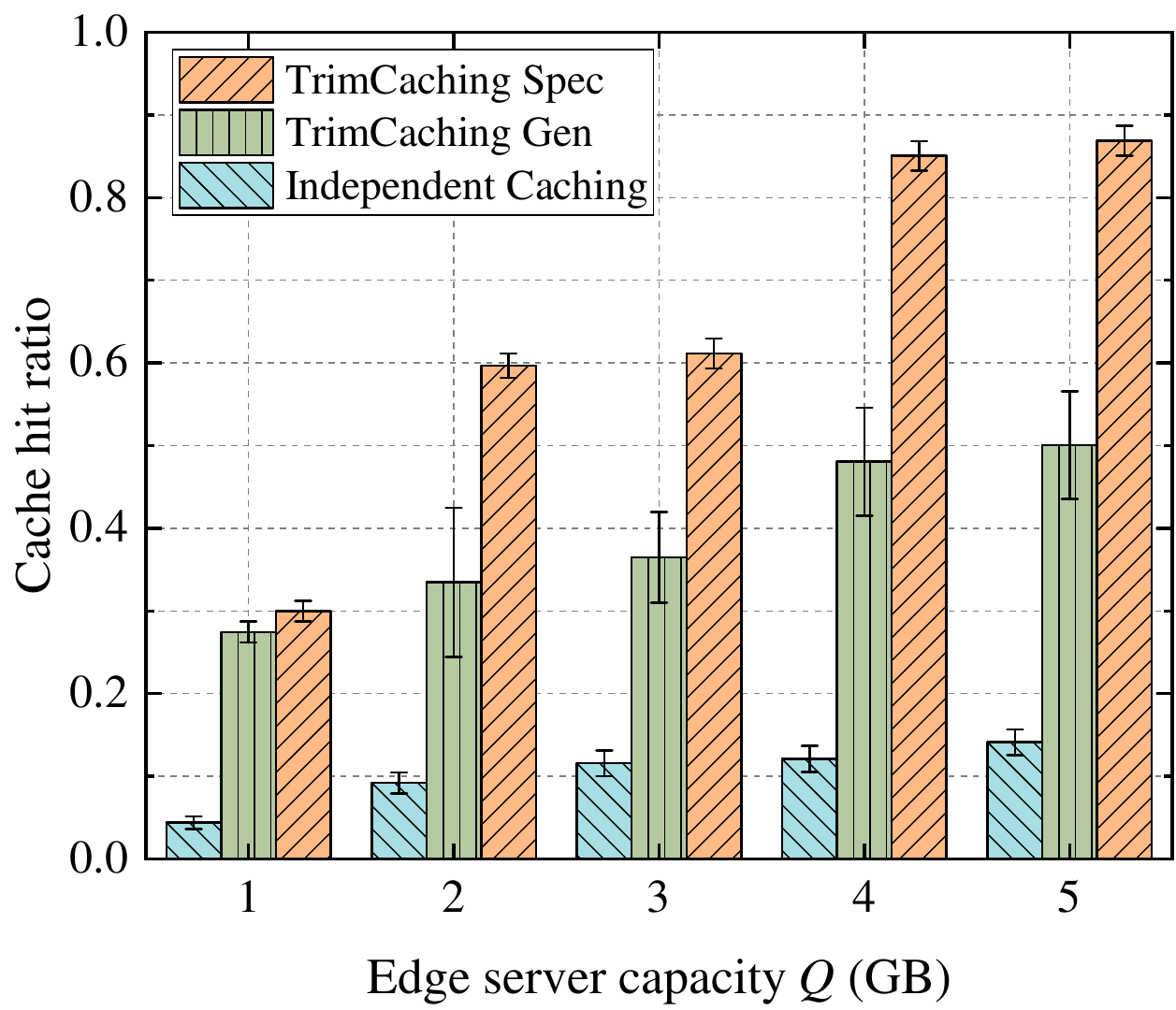}\label{fig_loracase2_capacity}}
	\quad
	\subfigure[Cache hit ratio vs. $M$, where $Q=3 \ \text{GB}$ and $K=30$.]{\includegraphics[width=0.23\textwidth]{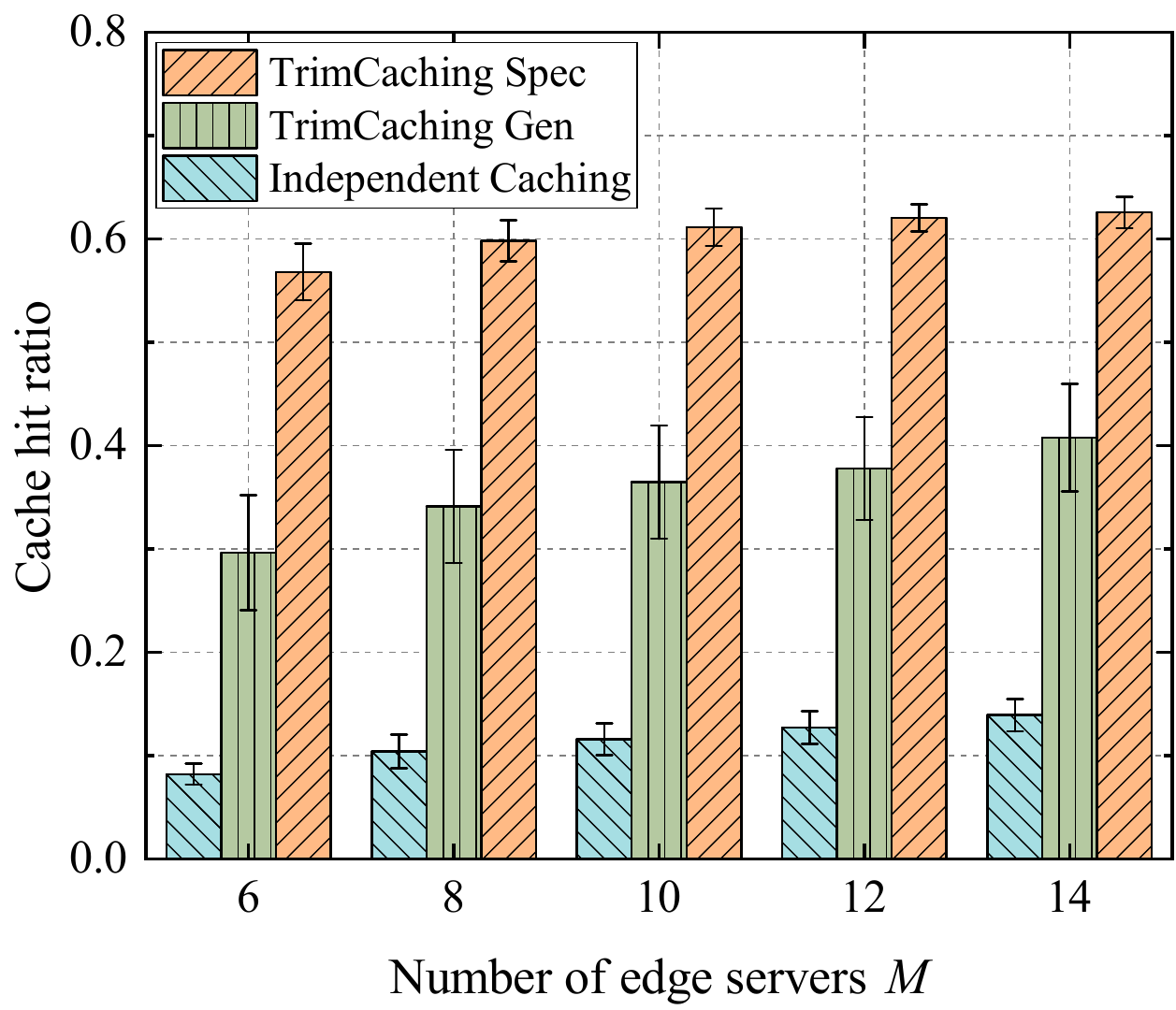}\label{fig_loracase2_edge}}
	\quad
	\subfigure[Cache hit ratio vs. $K$, where $Q=3 \ \text{GB}$ and $M=10$.]{\includegraphics[width=0.23\textwidth]{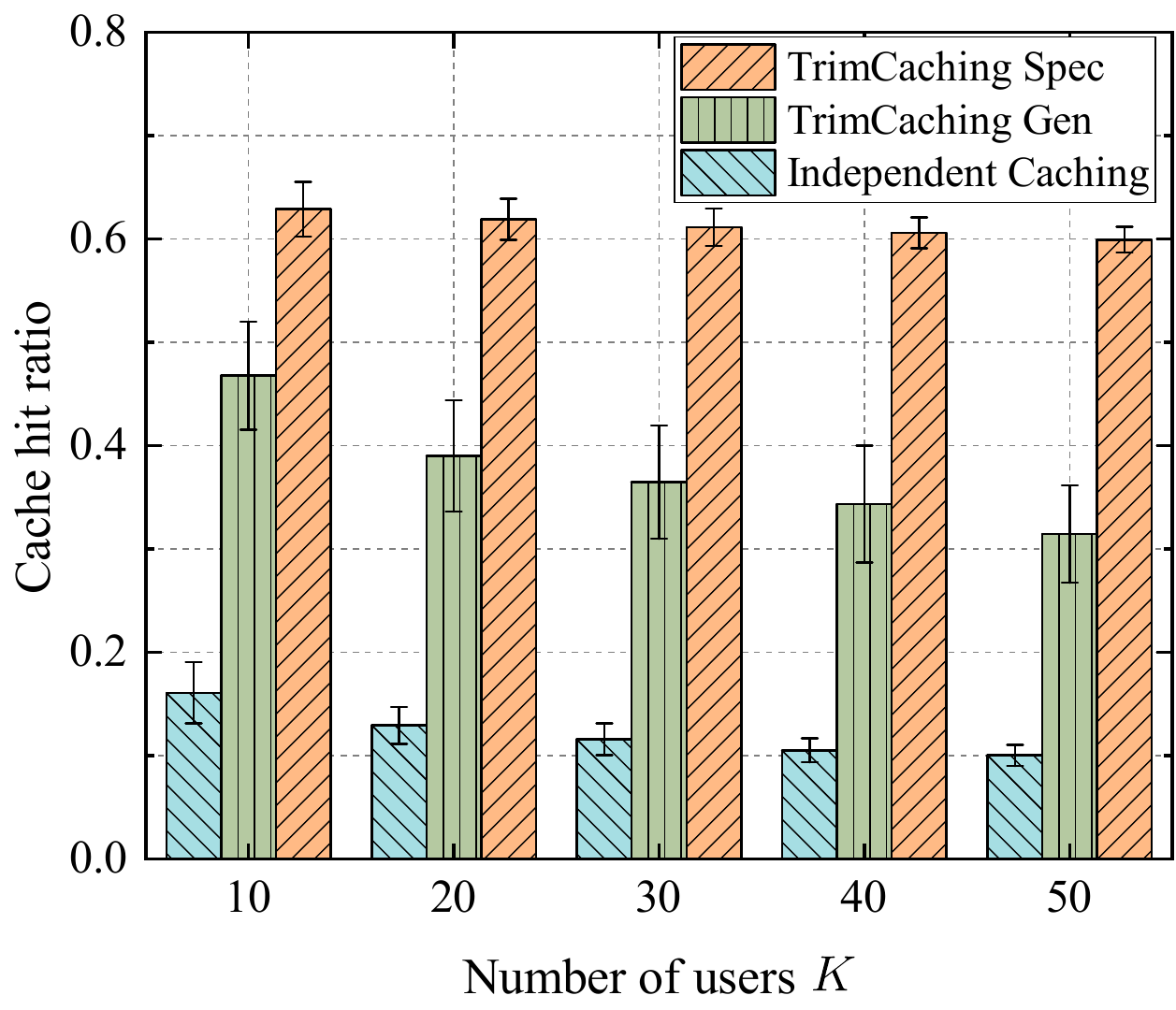}\label{fig_loracase2_user}}
	\caption{Cache hit ratio in the special case using the GPT-2-based model library.}\label{fig_loraspecial_case}
\end{figure*}
\subsection{Performance in the Special Case}
We first compare the algorithms in the special case. Since the TrimCaching Gen algorithm can also be applied to the special case, it is used as a benchmark. Fig. \ref{fig_special_case} presents the cache hit ratio of these algorithms using the ResNet-based model library. As shown in Fig. \ref{fig_case2_capacity}, increasing the storage capacity $Q$ improves the cache hit ratio of all algorithms, as more models can be cached on edge servers. Unsurprisingly, both the TrimCaching Spec and TrimCaching Gen algorithms outperform the Independent Caching framework due to the storage efficiency of parameter-sharing caching. Moreover, the TrimCaching Spec algorithm consistently outperforms the TrimCaching Gen algorithm, as it offers a constant approximation guarantee, whereas the TrimCaching Gen algorithm does not. On average, the TrimCaching Spec algorithm achieves 11.93\% and 33.93\% higher cache hit ratios than the TrimCaching Gen algorithm and the Independent Caching framework, respectively. Moreover, a similar trend can be observed in Fig. \ref{fig_case2_edge}, which varies the number of edge servers.

Fig. \ref{fig_case2_user} evaluates the performance of the algorithms as the number of users increases. As expected, the cache hit ratio decreases as the number of users grows due to reduced transmission data rates. However, the proposed algorithms maintain substantial performance advantages. When $K=50$, the TrimCaching Spec and TrimCaching Gen algorithms improve the cache hit ratio by about 21.81\% and 15.47\%, respectively, over the Independent Caching framework. Besides, the TrimCaching Spec algorithm achieves an average performance gain of 11.50\% compared with the TrimCaching Gen algorithm, demonstrating its superior performance in the special case, while the TrimCaching Gen algorithm remains competitive.  \par 

Fig. \ref{fig_loraspecial_case} demonstrates the cache hit ratio of the algorithms using the GPT-2-based model library, exhibiting trends similar to those shown in Fig. \ref{fig_special_case}. 
The TrimCaching Spec algorithm demonstrates a notable improvement in average cache hit ratios compared with the TrimCaching Gen algorithm and the Independent Caching framework. As illustrated in Fig. \ref{fig_loraspecial_case}, the TrimCaching Spec algorithm achieves higher cache hit ratios than the TrimCaching Gen algorithm and the Independent Caching framework by 25.44\% and 54.26\%, 24.71\% and 49.11\%, 23.68\% and 49.06\%, respectively, under varying storage capacities, numbers of edge servers, and numbers of users.  

\subsection{Performance in the General Case}
\begin{figure*}[!t]
	\centering
	\subfigure[Cache hit ratio vs. $Q$, where $M=10$ and $K=30$. ]{\includegraphics[width=0.23\textwidth]{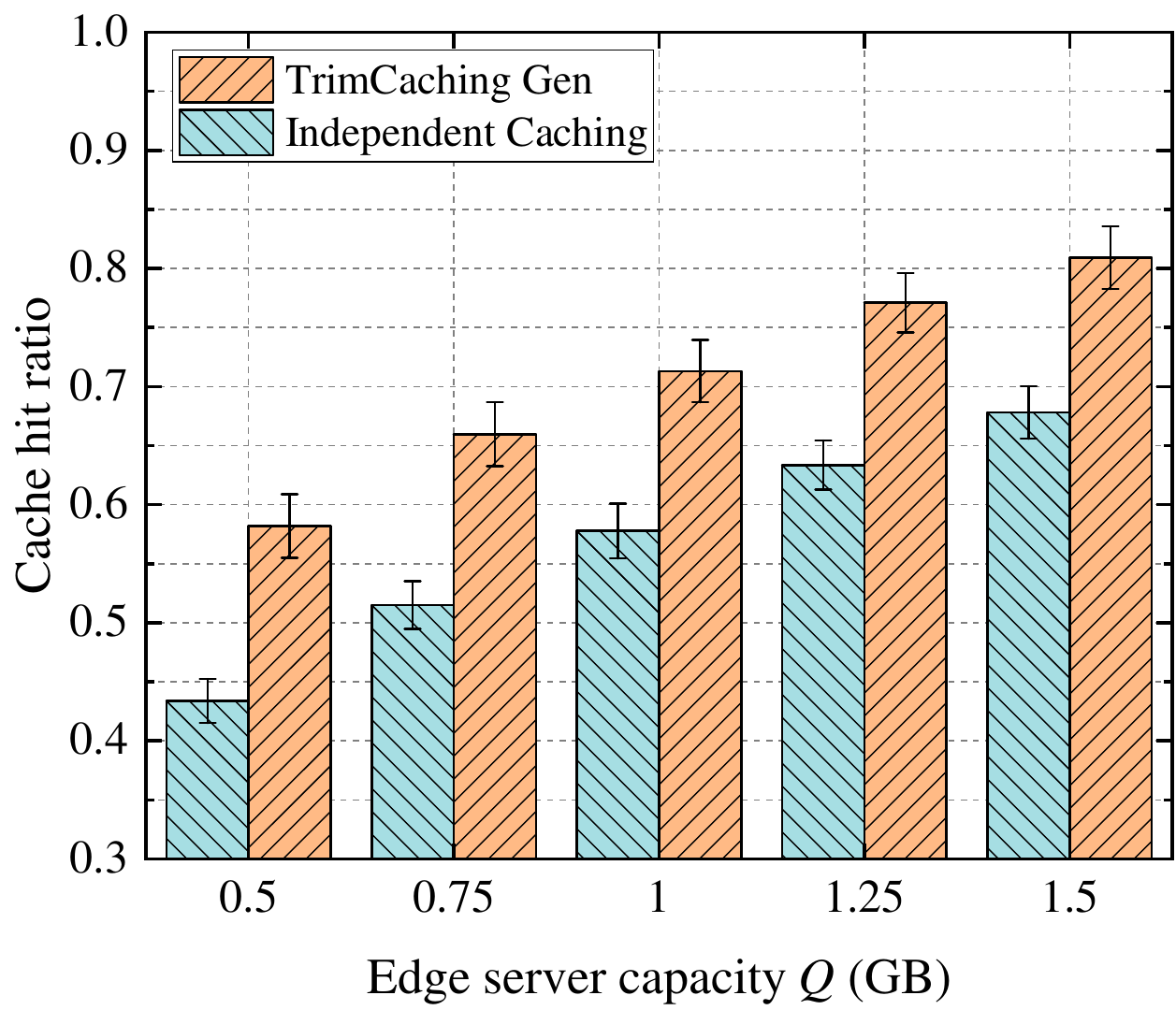}\label{fig_real_case3_capacity}}
	\quad
	\subfigure[Cache hit ratio vs. $M$, where $Q=1 \ \text{GB}$ and $K=30$.]{\includegraphics[width=0.23\textwidth]{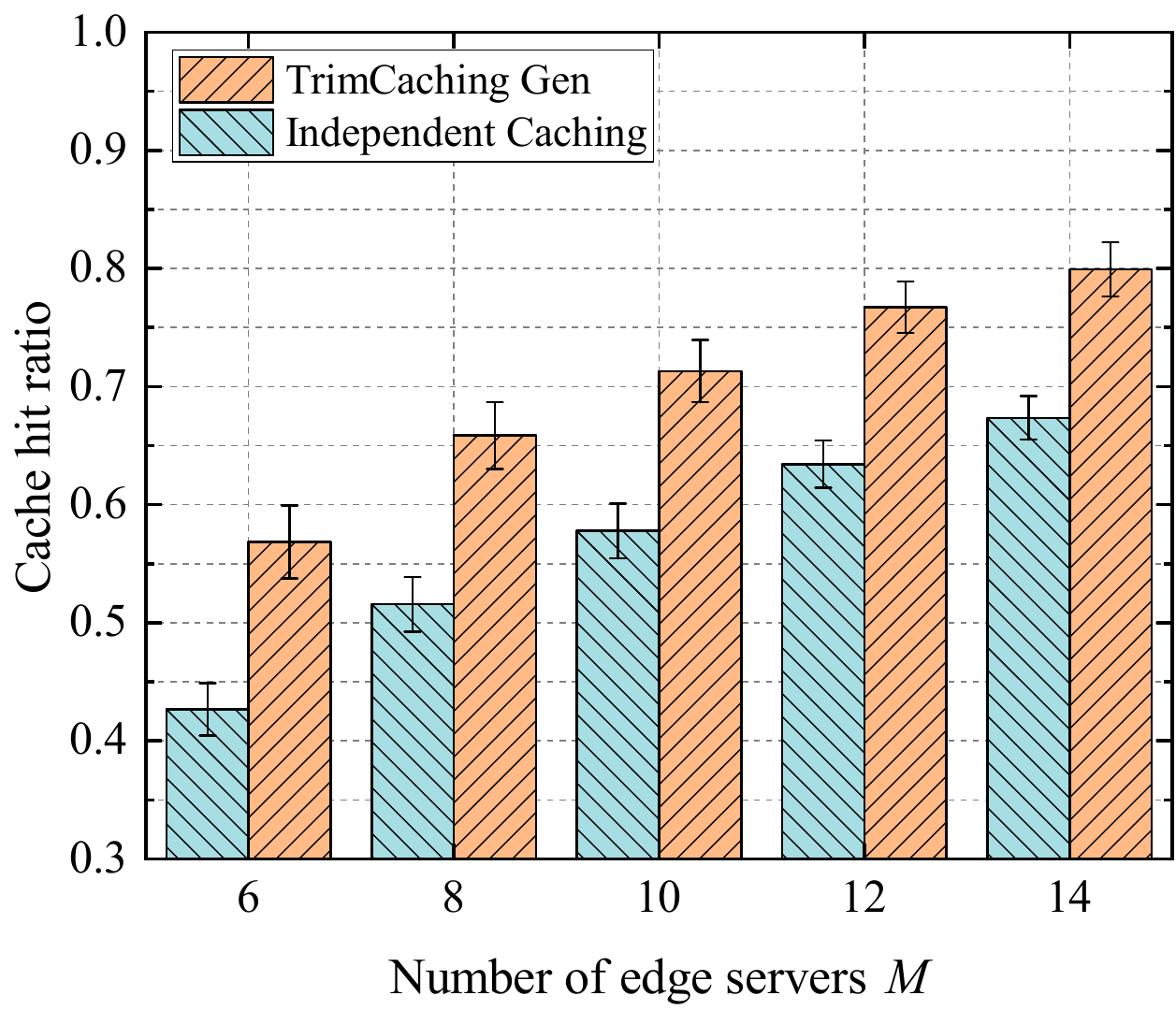}\label{fig_real_case3_edge}}
	\quad
	\subfigure[Cache hit ratio vs. $K$, where $Q=1 \ \text{GB}$ and $M=10$.]{\includegraphics[width=0.23\textwidth]{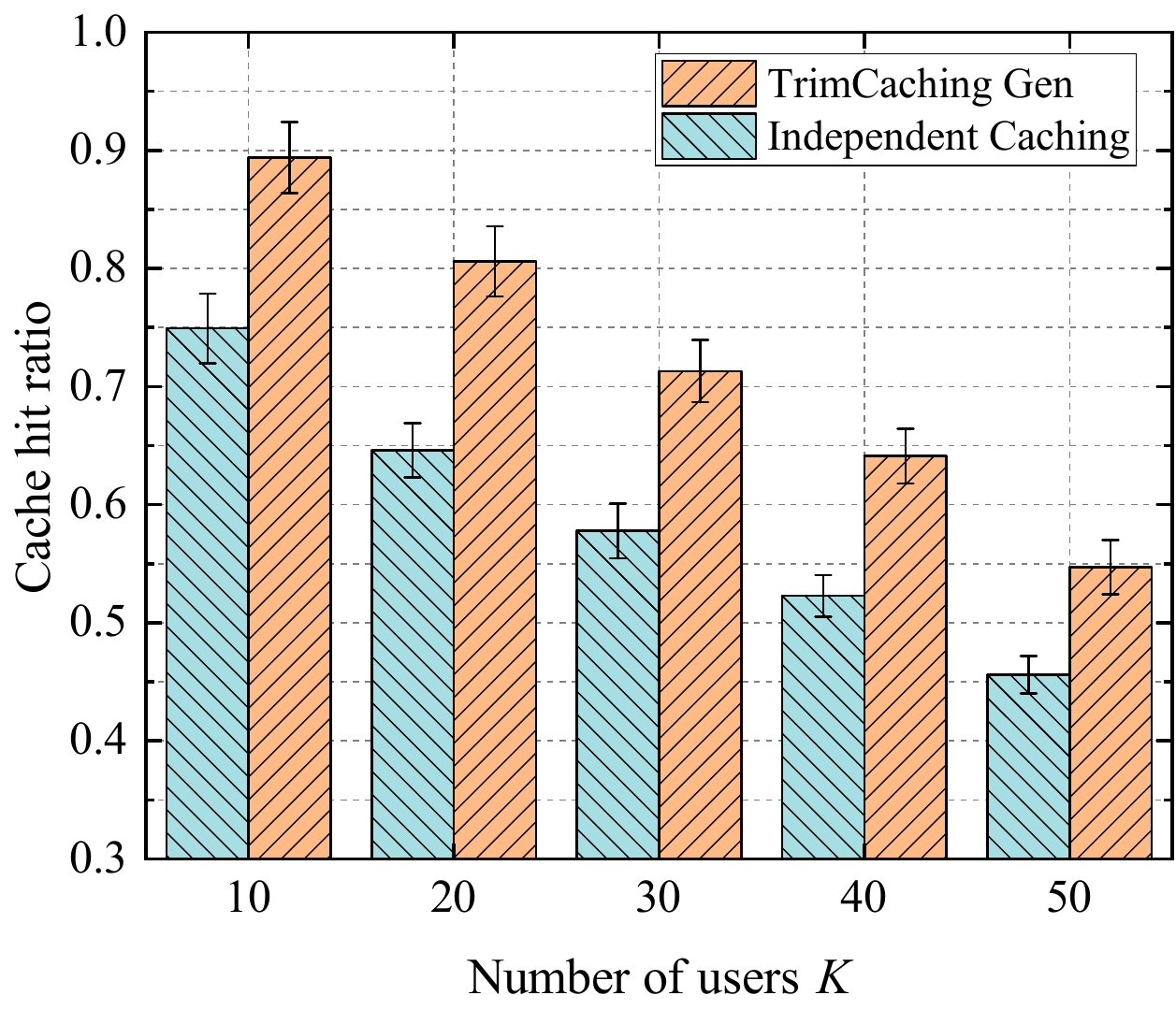}\label{fig_real_case3_user}}
	\quad
	\caption{Cache hit ratio in the general case using the ResNet-based model library.}\label{fig_real_case3}
\end{figure*}
%\begin{figure*}[!ht]
%	\centering
%	\subfigure[Cache hit ratio v.s. $Q$, where $M=10$ and $I=30$. ]{\includegraphics[width=2in]{loracase3capacity.pdf}\label{fig_lorareal_case3_capacity}}
%	\quad
%	\subfigure[Cache hit ratio v.s. $M$, where $Q=3 \ \text{GB}$ and $I=30$.]{\includegraphics[width=2in]{loracase3edge.pdf}\label{fig_lorareal_case3_edge}}
%	\quad
%	\subfigure[Cache hit ratio v.s. $K$, where $Q=3\ \text{GB}$ and $M=10$.]{\includegraphics[width=2in]{loracase3user.pdf}\label{fig_lorareal_case3_user}}
%	\quad
%	\caption{Cache hit ratio for the general case with the GPT-2-based model library.}\label{fig_lorareal_case3}
%\end{figure*}
As shown in Fig. \ref{fig_real_case3}, this subsection compares the cache hit ratio of the TrimCaching Gen algorithm and the Independent Caching framework using the ResNet-based model library in the general case. A trend similar to the special case can be observed. As the storage capacity or the number of edge servers increases, the cache hit ratio improves since more models can be cached. Moreover, increasing the number of users leads to a decline in the cache hit ratio due to the limited spectrum bandwidth. More importantly, the TrimCaching Gen algorithm significantly outperforms the Independent Caching framework in different scenarios. %Moreover, we investigate the cache hit ratio of the above two algorithms with the GPT-2 based model library. The parameter sharing configuration is same as the special case while the backhaul transmission latency is considered all the time. The performance of the TrimCaching Gen algorithm and Independent Caching is similar to the special case. The TrimCaching Gen algorithm performs much better than the Independent Caching. 

% Fig. \ref{fig_real_case3_edge} evaluates the scalability of the frameworks with respect to the number of edge servers. Similar to Fig. \ref{fig_case2_edge}, the increase in the number of edge servers also leads to the improvement in the hit ratio, and the increasing trend slows down when the number of edge servers arrives at $M=14$. The hit ratio of the TrimCaching Gen algorithm is $18\%$ higher than the Independent Caching framework on average. Besides, the TrimCaching Gen algorithm can obtain a $20\%$ higher hit ratio when $M=14$. \par 

% Fig. \ref{fig_real_case3_user} illustrates the performance versus the different numbers of users. The TrimCaching framework maintains a $19\%$ higher hit ratio on average than the Independent Caching. Although an increase in the number of users leads to a decrease in the hit ratio, the TrimCaching Gen algorithm can still satisfy more than $60\%$ requests when $K=50$. The cache hit ratio of the TrimCaching Gen algorithm at $k=50$ is about $17\%$ higher than the Independent Caching.\par 
\subsection{Running Time Comparison}
In this subsection, we compare the running time of the TrimCaching Spec algorithm, the TrimCaching Gen algorithm, and the exhaustive search using the ResNet-based model library. Here, the exhaustive search is employed to obtain the optimal solution, and incurs exponential complexity, i.e., $2^{MKI}$. To run the exhaustive search efficiently, the length and width of the area are reduced to 400 m, and $M$ and $K$ are set to 2 and 6, respectively. Moreover, $\epsilon$ of Algorithm \ref{algorithm_DP} is set to 0 in this subsection. 

% \begin{figure}[!h]
% 	\centering
% 	\subfigure[Hit ratio and average execution time in special case, where the maximal value of $I$ is 18.]{\includegraphics[width=2.5in]{case2optimal.pdf}\label{fig_case2_optimal}}
% 	\quad
% 	\subfigure[Hit ratio and average execution time in general case, where the maximal value of $I$ is 12.]{\includegraphics[width=2.5in]{realcase3optimal.pdf}\label{fig_real_case3_optimal}}
% 	\caption{Comparison with optimal solutions.}\label{fig_optimal}
% \end{figure}
\begin{figure}[!t]
	\centering
        \subfigure[Results in the special case, where $Q=0.1\ \text{GB}$ and each user requests 9 models.]{\includegraphics[width=0.23\textwidth]{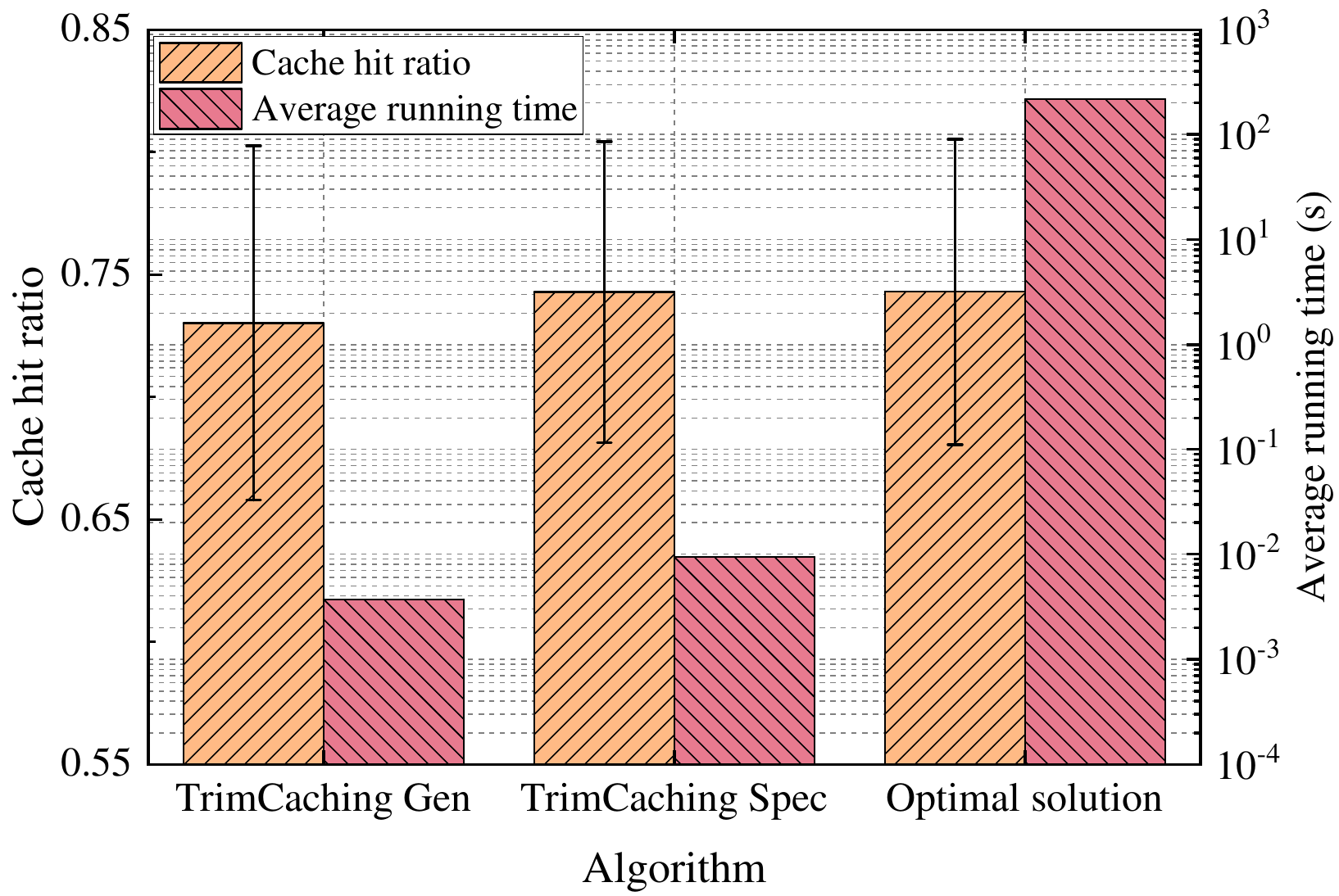}\label{fig_case2_optimal}}
	\quad
	\subfigure[Results in the general case, where $Q=0.2\ \text{GB}$ and each user requests 27 models.]{\includegraphics[width=0.23\textwidth]{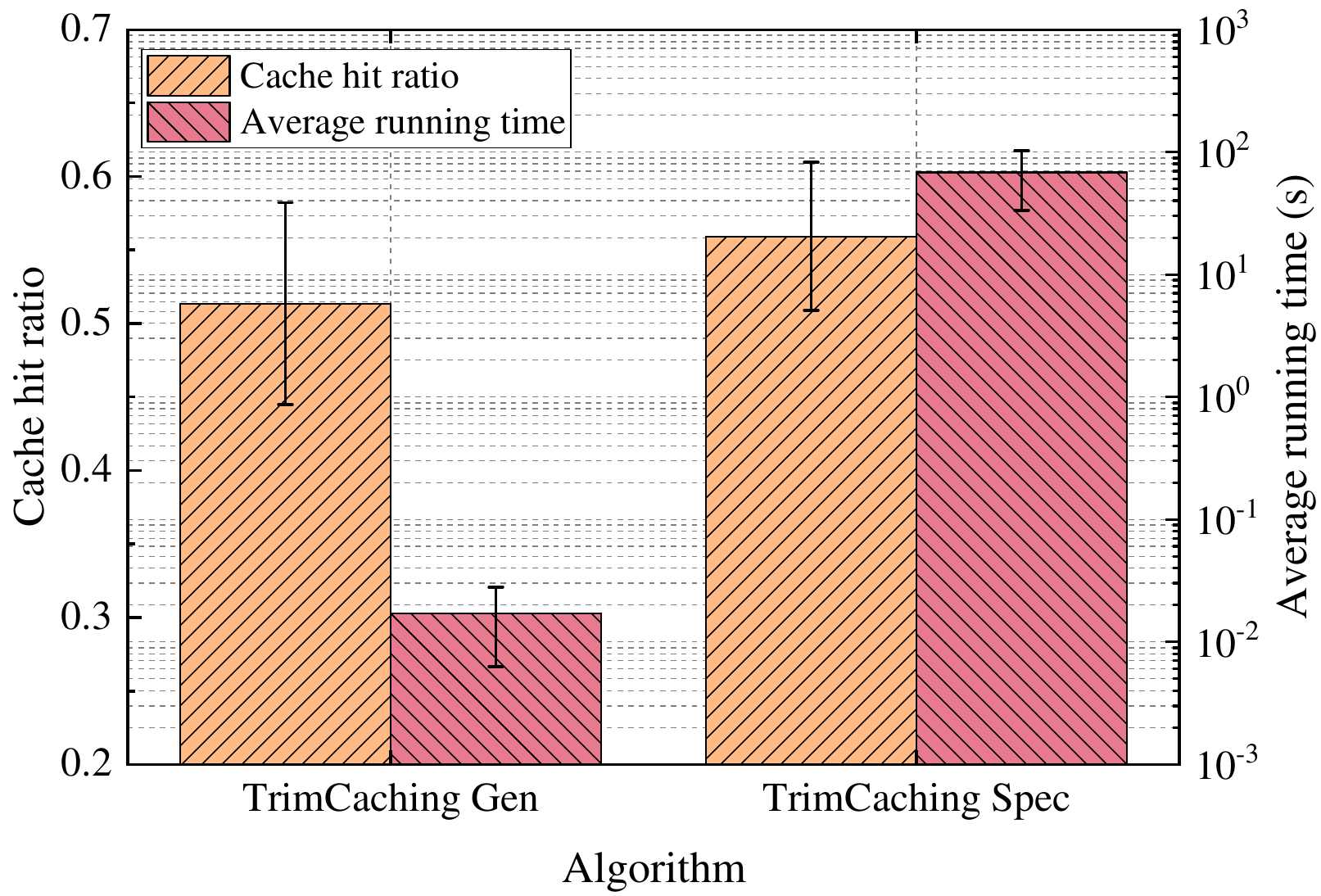}\label{fig_case2_case3}}
 \caption{Cache hit ratio and average running time of different algorithms.}
\end{figure}\label{fig_optimal}
\par 
Fig. \ref{fig_case2_optimal} compares the performance of our algorithms with the exhaustive search in the special case. On the one hand, the TrimCaching Spec algorithm achieves the same cache hit ratio as the optimal solution, while the cache hit ratio of the TrimCaching Gen algorithm is only 1.3\% lower than the optimal solution. On the other hand, the TrimCaching Spec algorithm and the TrimCaching Gen algorithm are, on average, about 22,900 and 58,000 times faster than the exhaustive search, respectively, demonstrating the effectiveness and efficiency of our algorithms in the special case. %However, when applied to the general case, the TrimCaching Spec algorithm has exponential complexity. 
Moreover, Fig. \ref{fig_case2_case3} compares the performance of the TrimCaching Gen algorithm with the TrimCaching Spec algorithm in the general case. The average running time of the TrimCaching Gen algorithm is about 3,900 times faster than the TrimCaching Spec algorithm. This underscores the necessity of designing the TrimCaching Gen algorithm for the general case, since the TrimCaching Spec algorithm exhibits exponential complexity in the general case.
%Fig. \ref{fig_real_case3_optimal} presents the performance of the Algorithm \ref{algorithm_greedy} and exhaustive search algorithm when models in case 3. The hit ratio achieved by Algorithm \ref{algorithm_greedy} is $1\%$ lower than the optimal solution. However, the execution is facilitated by $6600$ times. Fig. \ref{fig_real_case3_optimal} suggests the Algorithm \ref{algorithm_greedy} can reach a near optimal hit ratio with faster speed.\par 
\subsection{Robustness to User Mobility}
In Fig. \ref{fig_mobility}, we demonstrate the robustness of our algorithms over time by considering user mobility. In this subsection, we set $M$, $K$, and $Q$ to 10, 10, and 3 GB, respectively, using the GPT-2-based model library. Moreover, three mobility patterns are considered, representing pedestrians, bikes, and vehicles. The initial speeds of users are randomly drawn from $\left[0.5,1.8\right]\ \text{m/s}$, $\left[2,8\right]\ \text{m/s}$, and $\left[5.5,20\right]\ \text{m/s}$, respectively. In each time slot, the acceleration is selected from $\left[-0.3,0.3\right]\ \text{m/}\text{s}^\text{2}$, $\left[-1,1\right]\ \text{m/}\text{s}^\text{2}$, and $\left[-3,3\right]\ \text{m/}\text{s}^\text{2}$, respectively. The initial movement orientations are uniformly distributed in $\left[0,\pi\right]\ \text{rad}$. The angular velocities range in $\left[-\frac{\pi}{4},\frac{\pi}{4}\right]\ \text{rad/s}$, $\left[-\frac{\pi}{3},\frac{\pi}{3}\right]\ \text{rad/s}$, and $\left[-\frac{\pi}{2},\frac{\pi}{2}\right]\ \text{rad/s}$. Users are assumed to change their speeds/orientations at the beginning of each time slot, and the time slot duration is set to 5 s. \par 
\begin{figure}[!t]
\centerline{\includegraphics[width=0.23\textwidth]{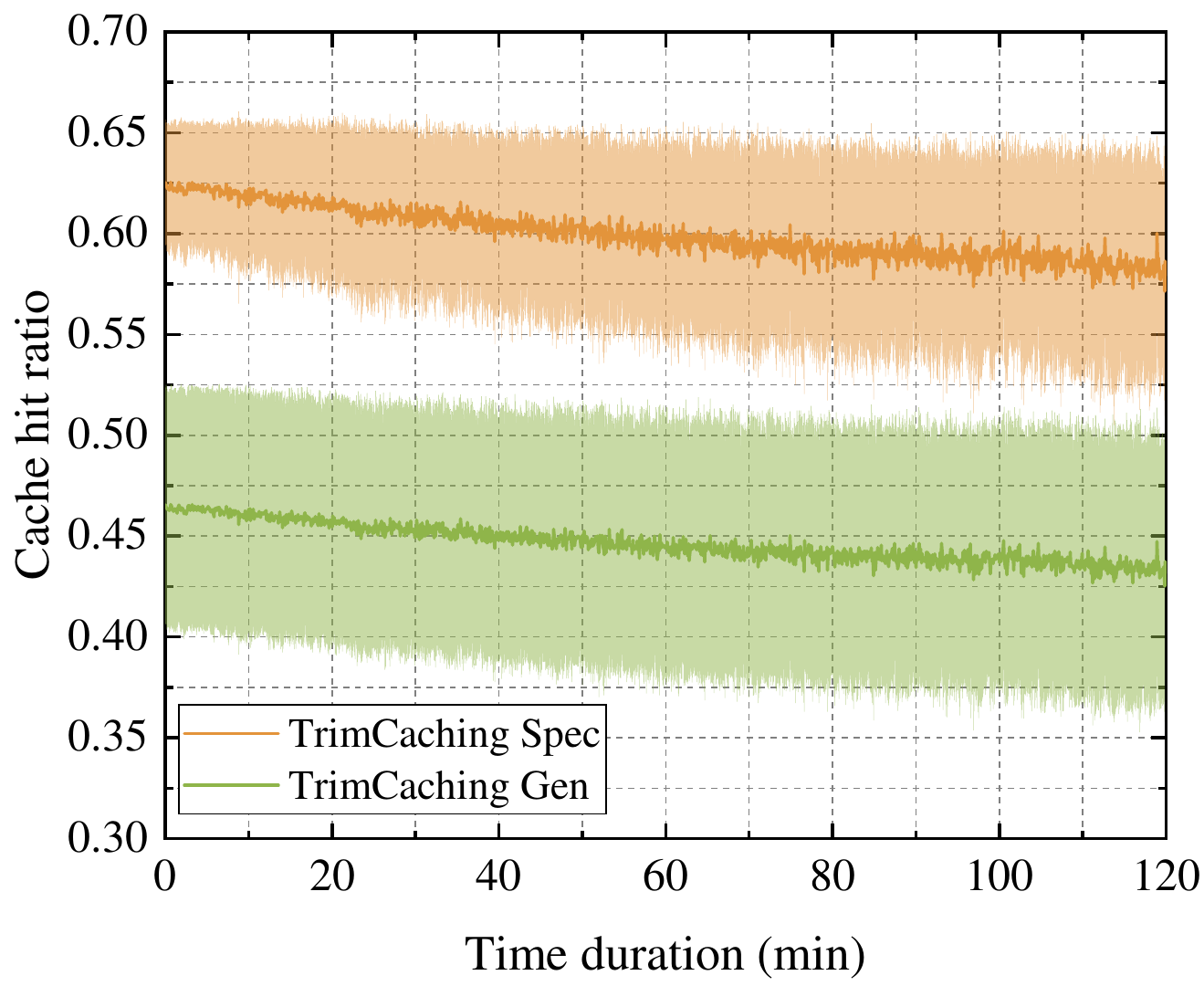}}
	\caption{Impact of user mobility on the cache hit ratio over time.
 }
	\label{fig_mobility}
\end{figure}
Fig. \ref{fig_mobility} illustrates the robustness of the TrimCaching Spec and TrimCaching Gen algorithms in the special case using the GPT-2-based model library. While the cache hit ratio degrades over time due to user mobility, the performance degradation of the TrimCaching Spec and TrimCaching Gen algorithms is only about 3.91\% and 2.83\%, respectively, over 2 hours. This demonstrates the effectiveness of the TrimCaching framework in the long run, even though user mobility has not been explicitly considered in our problem formulation. This also implies that model replacement does not need to be re-conducted frequently, thereby saving backbone bandwidth resources. %these two figures are 6 and 11 in GPT-2-based model library. 
\subsection{Impact of Request Probability Uncertainty on Cache Hit Ratio}

This subsection investigates the effect of the uncertainty of model request probabilities on the cache hit ratio, as the estimated request probabilities used for model caching decision determination may deviate from the true value $p_{k,i}$. Following prior work, we model the estimated request probability of user $k$ for model $i$, denoted by $p'_{k,i}$, as a Beta distribution with mean $p_{k,i}$~\cite{6881261,9053162}. %more: ,6871674,pfuhl2009optimal
In our experiments, the model caching decisions are first determined based on $p'_{k,i}$, while the cache hit ratio is evaluated through Monte Carlo experiments, where each user generates a specific model request according to $p_{k,i}$ in each Monte Carlo experiment. We analyze the cache hit ratio in terms of the coefficient of variation of $p'_{k,i}$.

Fig. \ref{fig_prob} demonstrates that the cache hit ratio of the proposed algorithms decreases slightly as the coefficient of variation of $p'_{k,i}$ increases. Fig. \ref{fig_case2_prob} shows that the cache hit ratio of the TrimCaching Spec and TrimCaching Gen algorithms degrades by 2.20\% and 2.78\%, respectively, as the coefficient of variation of $p'_{k,i}$ increases from 0 to 0.5 in the special case. Besides, in the general case, the cache hit ratio of the TrimCaching Gen algorithm drops by 1.82\% over the same range in Fig.~\ref{fig_case3_prob}. These results demonstrate that the proposed algorithms remain robust with only minor performance degradation even when the estimated request probabilities deviate from their true values substantially.
\begin{figure}[!t]
	\centering
        \subfigure[Results in the special case.]{\includegraphics[width=0.23\textwidth]{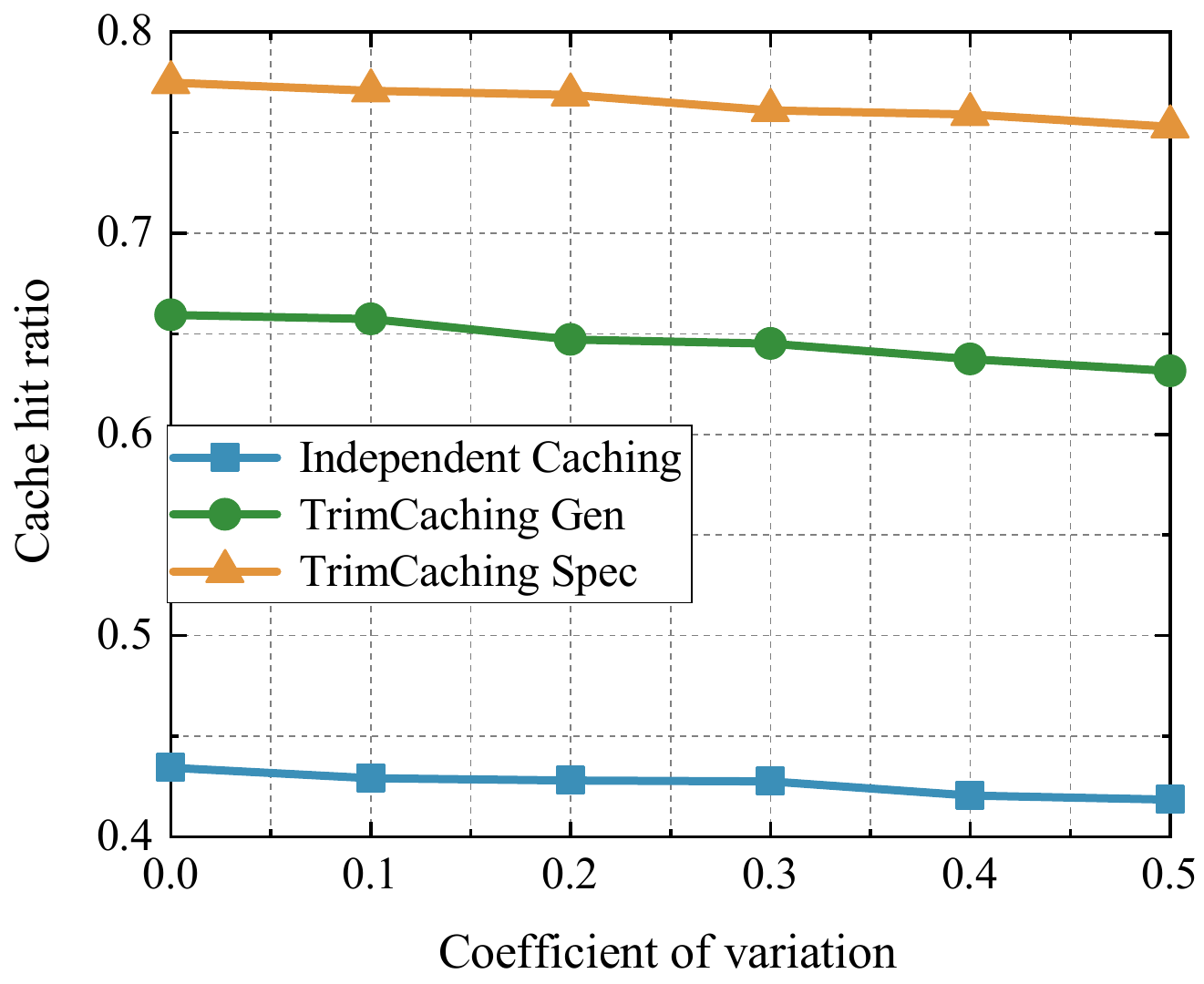}\label{fig_case2_prob}}
	\quad
	\subfigure[Results in the general case.]{\includegraphics[width=0.23\textwidth]{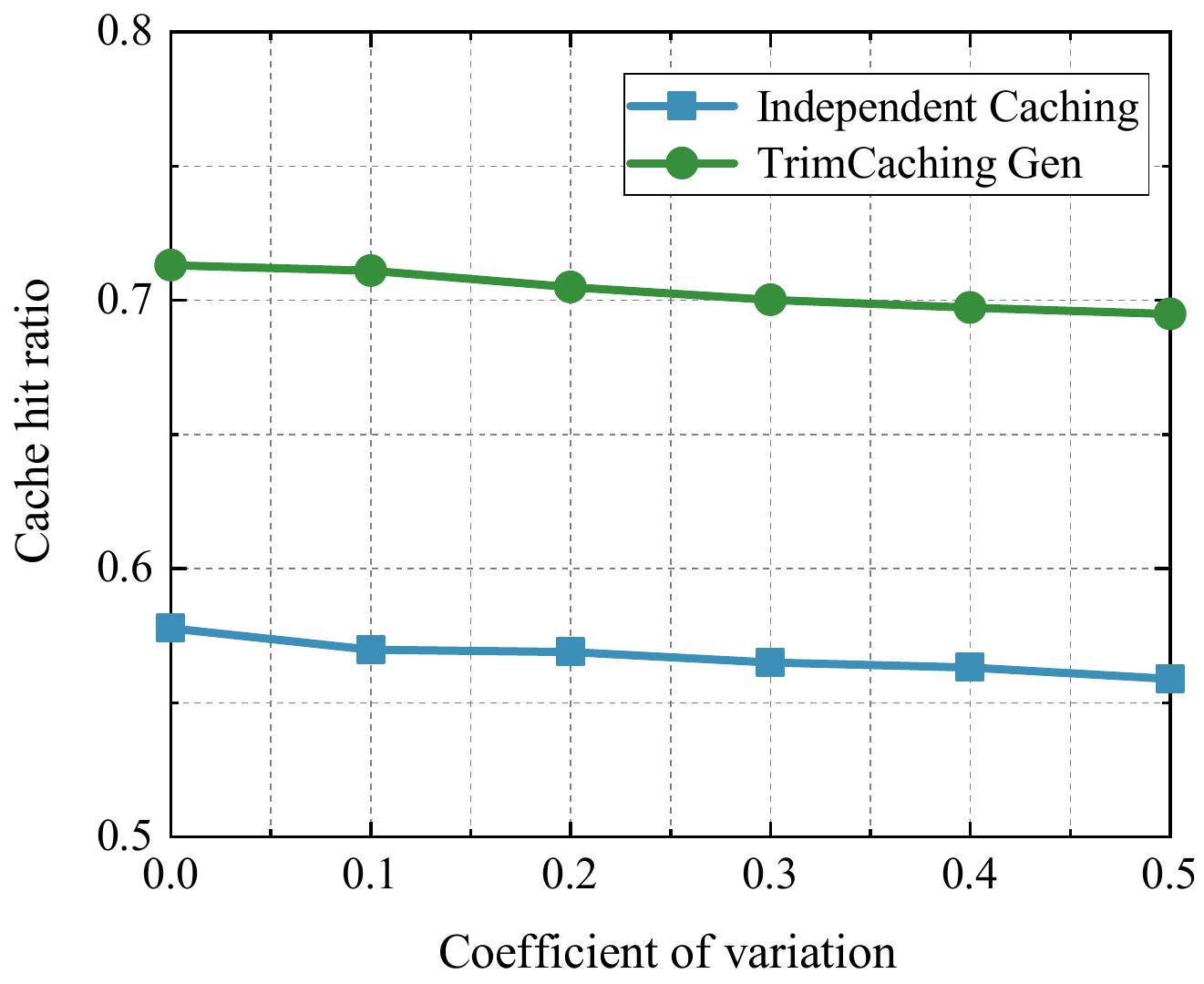}\label{fig_case3_prob}}
 \caption{Cache hit ratio versus the coefficient of variation of estimated request probabilities using the ResNet-based model library, where $Q=$ 1 GB, $K=$ 30, and $M=$ 10.}
 \label{fig_prob}
\end{figure}
\subsection{Comparison with Online Placement Strategies}
In this subsection, we compare the proposed algorithms with two classical online content placement strategies, including the Least Recently Used (LRU) and the Least Frequently Used (LFU) strategies, under both the special and general cases. The experiments are conducted in a dynamic setting, where each user generates one specific model request according to $p_{k,i}$ at the beginning of each 10-second time slot. The placement decisions of the TrimCaching Spec and TrimCaching Gen algorithms are determined once at time 0 based on $p_{k,i}$ and remain unchanged throughout the experiment. In contrast, both LRU and LFU update their placements dynamically according to the user requests at the end of each time slot. Under LRU, models that are least recently requested are removed from the cache, while under LFU, those with the least request frequency are removed~\cite{9187344,8057300}. %more: 8362880
We consider parameter sharing in both LRU and LFU implementations. The cache hit ratio is recorded at the start of each time slot to represent the performance for that slot.

As shown in Fig. \ref{fig_online}, the proposed algorithms consistently outperform the LRU and LFU strategies. Despite using static placement decisions computed only once based on $p_{k,i}$ at the beginning, both the TrimCaching Spec and TrimCaching Gen algorithms achieve significantly higher cache hit ratios. Specifically, after the performance of the LRU and LFU strategies stabilizes (after 30 minutes), the TrimCaching Spec algorithm improves the cache hit ratio by an average of 33.71\% and 21.61\% over LRU and LFU, respectively, while the TrimCaching Gen algorithm achieves improvements of 22.17\% and 10.07\%, respectively, in the special case, as shown in Fig. \ref{fig_case2_online}. Moreover, a similar trend can be observed in Fig. \ref{fig_case3_online}, where the TrimCaching Gen algorithm outperforms LRU and LFU by 26.05\% and 13.20\% on average, respectively, in the general case.
\begin{figure}[!t]
	\centering
        \subfigure[Comparisons in the special case.]{\includegraphics[width=0.23\textwidth]{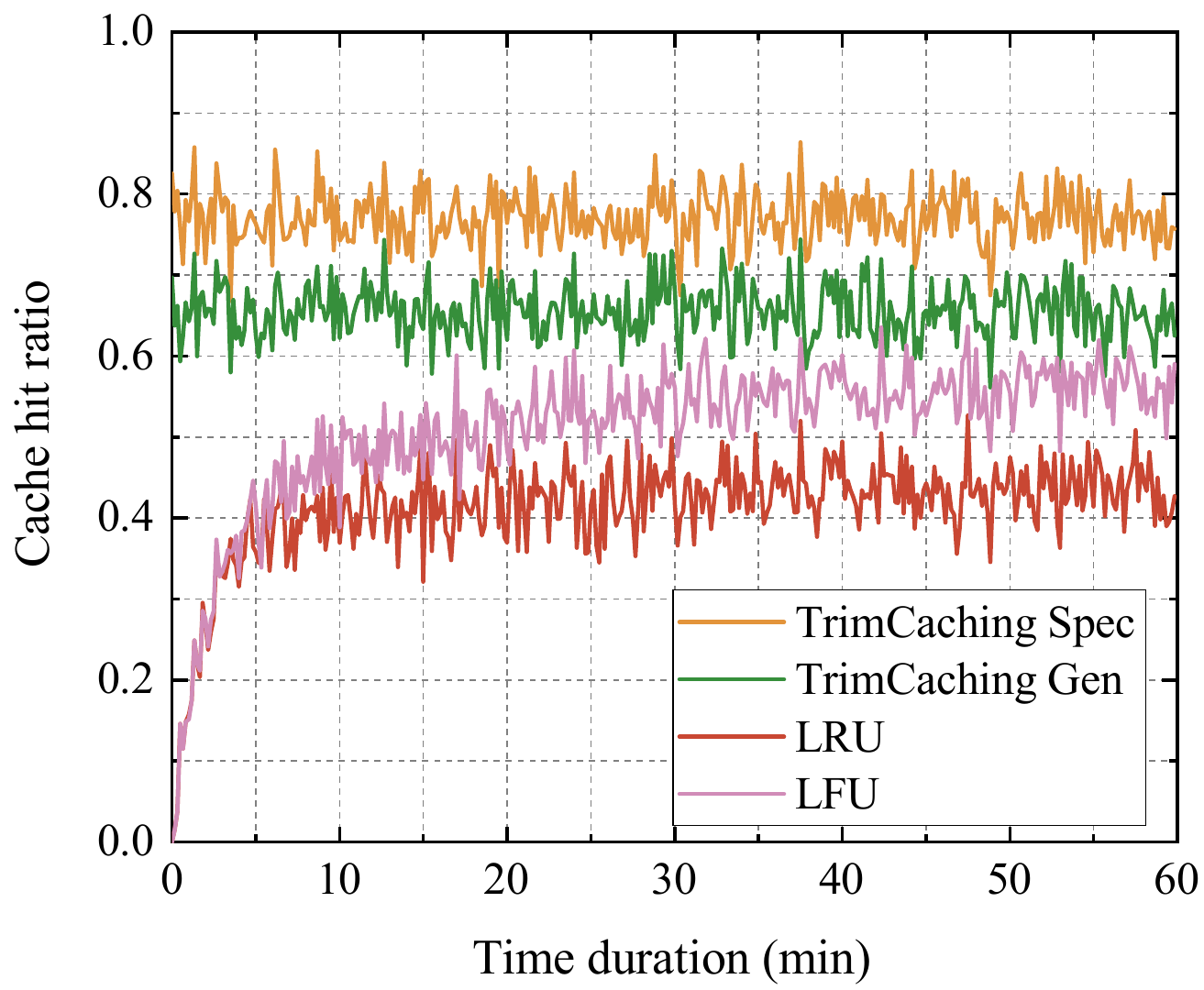}\label{fig_case2_online}}
	\quad
	\subfigure[Comparisons in the general case.]{\includegraphics[width=0.23\textwidth]{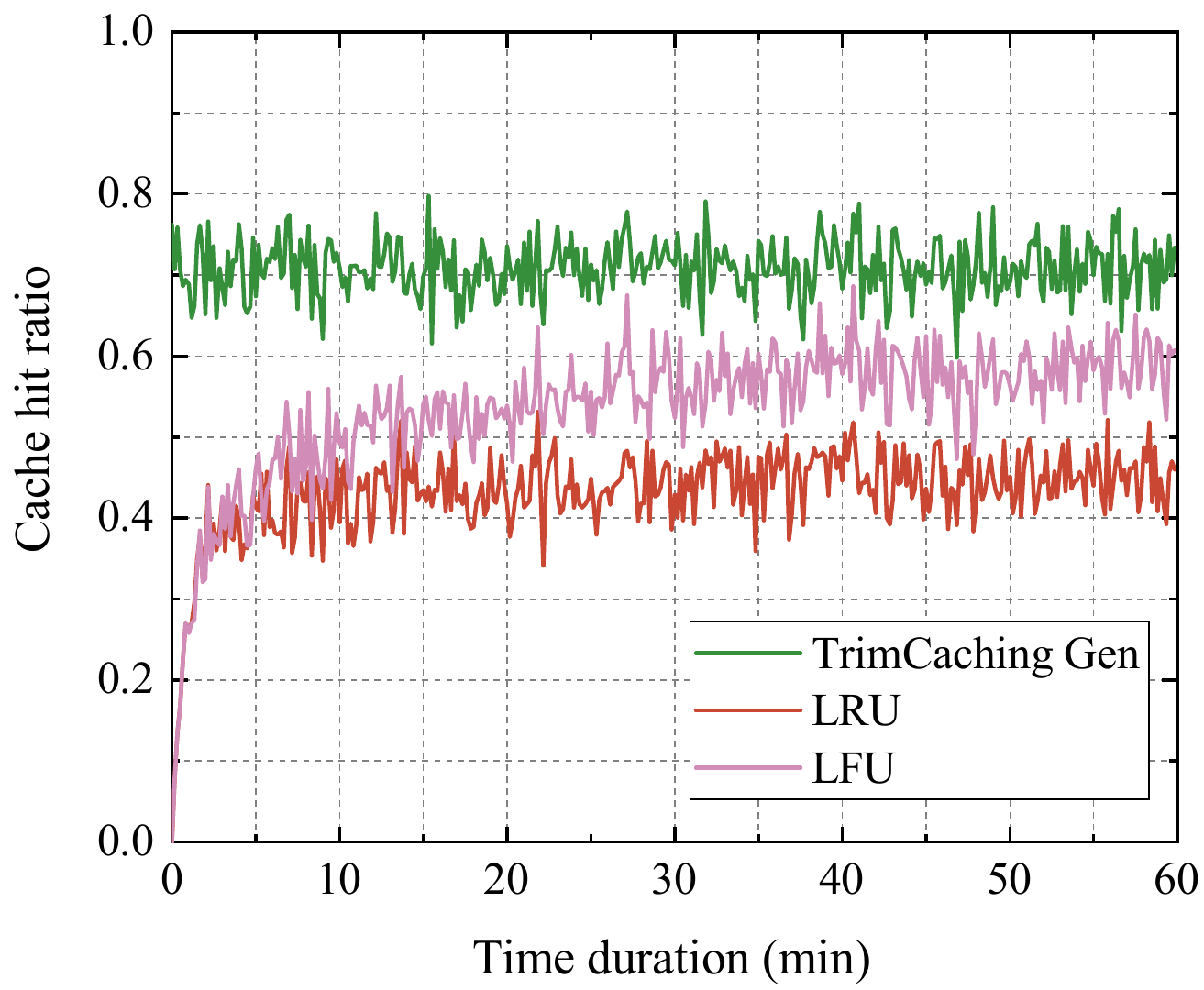}\label{fig_case3_online}}
 \caption{Cache hit ratio comparisons with online placement strategies using the ResNet-based model library, where $Q$, $K$, and $M$ follow those in Fig. \ref{fig_prob}.}
 \label{fig_online}
\end{figure}
\section{Conclusions}
%In this paper, we proposed the TrimCaching framework for edge AI model caching. Specifically, by taking advantage of the shared parameter blocks among AI models for efficient storage, we investigated the parameter-sharing AI model placement problem. Considering multi-edge scenarios, the model placement problem was formulated to maximize the cache hit ratio under the constraints of service latency requirement and storage capacity. Then, we identified this problem as a submodular maximization problem with submodular constraints, revealing that no polynomial-time algorithm can solve it with a constant approximation ratio. To tackle this challenge, we first investigated the special case with a small fixed number of shared parameter blocks independent of the problem scale. We developed a polynomial-time algorithm with $\left(1-\epsilon\right)/2$-approximation guarantee. After that, we addressed the general case by devising a greedy algorithm that is also efficient and effective. It has been shown that the proposed TrimCaching framework significantly improves the cache hit ratio compared with traditional content placement strategies without considering shared parameter blocks among AI models. We hope this work will inspire further research on edge AI model caching, which has the potential to become a vital component in 6G edge intelligence.

In this paper, we proposed the TrimCaching framework for efficient edge AI model caching. By leveraging shared parameter blocks among AI models for efficient storage, we formulated the parameter-sharing AI model placement problem to maximize the cache hit ratio under user latency requirements and storage constraints in multi-edge scenarios. This paper is the first work to formulate such a problem and identify this problem as a submodular maximization problem with submodular knapsack constraints, which is fundamentally different from conventional edge content caching. 
To tackle this challenge, we first investigated the special case with a small fixed number of shared parameter blocks independent of the problem scale. We developed a polynomial-time rounding DP-based algorithm with $\left(1-\epsilon\right)/2$-approximation guarantee. After that, we addressed the general case by devising a greedy algorithm that is also efficient and effective. Our findings underscore the importance of parameter-sharing model caching for enabling low-latency AI model downloading. We hope this work will inspire further research on edge AI model caching, such as model caching in heterogeneous or hierarchical networks, thereby making it a key service of 6G mobile networks.

%show reference used in the appendix in the manuscript
\nocite{KELLERER201764,dawande2000approximation,chekuri2005polynomial}

\bibliographystyle{IEEEtran}
\bibliography{IEEEabrv,reference}

%%%%%%%%%start of the appendix

\newpage
\clearpage
\setcounter{page}{1}

\begin{appendices}
\section{Proof of Proposition \ref{proposition_submodular}}\label{proof_proposition_submodular}
We begin by introducing a few statements. For any feasible ${\bf{X}}$, let $\eta\left({\bf{X}}\right)=\left\{x_{m,i}\ \middle|\ x_{m,i}=1,x_{m,i}\in{\bf{X}}\right\}$ denote the set of model caching decisions with $x_{m,i}=1$, $\mathcal{I}_{{\bf{X}}}=\left\{i\ \middle| \ x_{m,i}=1, x_{m,i}\in{\bf{X}}\right\}$ denote the set of cached models, $\mathcal{M}_{{\bf{X}}}=\left\{m\ \middle| \ x_{m,i}=1, x_{m,i}\in{\bf{X}}\right\}$ denote the set of edge servers that cache models under ${\bf{X}}$.

We first prove that the objective function of $\mathcal{P}1.1$ is a submodular function. Let ${\bf{X}}$ and ${\bf{X}}'$ be two feasible model placement decisions in $\mathcal{P}1.1$ with $\eta\left({\bf{X}}\right)\subseteq\eta\left({\bf{X}}'\right)\subseteq{\bf{V}}$, where ${\bf{V}}$ is the universal set that contains all feasible model placement results for $\mathcal{P}1.1$. For any $x_{m,i}\in{\bf{V}}\setminus\eta\left({\bf{X}}'\right)$, the marginal increase in cache hit ratios of updating $x_{m,i}=0$ to $x_{m,i}=1$ in ${\bf{X}}'$ and ${\bf{X}}$ are 
\begin{equation}
    \begin{aligned}
        \Delta U\left({\bf{X}}',m,i\right) &=U\left({\bf{X}}'\setminus\left\{x_{m,i}=0\right\}\cup\left\{x_{m,i}=1\right\}\right) - U\left({\bf{X}}'\right)\\
        &=U\left(\eta\left({\bf{X}}'\right)\cup\left\{x_{m,i}\right\}\right) - U\left(\eta\left({\bf{X}}'\right)\right),
    \end{aligned}
\end{equation}
and 
\begin{equation}
    \begin{aligned}
        \Delta U\left({\bf{X}},m,i\right) &=U\left({\bf{X}}\setminus\left\{x_{m,i}=0\right\}\cup\left\{x_{m,i}=1\right\}\right) - U\left({\bf{X}}\right)\\
        &=U\left(\eta\left({\bf{X}}\right)\cup\left\{x_{m,i}\right\}\right) - U\left(\eta\left({\bf{X}}\right)\right),
    \end{aligned}
\end{equation}
respectively. Moreover, since $\eta\left({\bf{X}}\right)\subseteq\eta\left({\bf{X}}'\right)$, then we have $\mathcal{M}_{{\bf{X}}}\subseteq\mathcal{M}_{{\bf{X}}'}$ and $\mathcal{I}_{{\bf{X}}}\subseteq\mathcal{I}_{{\bf{X}}'}$, implying that $\mathcal{M}_{{\bf{X}}'}$ can provide at least as many models as $\mathcal{M}_{{\bf{X}}}$ to users. We consider the following three cases for $\Delta U\left({\bf{X}}',m,i\right)$ and $\Delta U\left({\bf{X}},m,i\right)$.
\begin{itemize}
    \item If users can download model $i$ from some edge server in $\mathcal{M}_{{\bf{X}}}$, where the servers in $\mathcal{M}_{{\bf{X}}}$ collectively cache $\mathcal{I}_{{\bf{X}}}$, within latency requirements, then they can also download it from some edge server in $\mathcal{M}_{{\bf{X}}'}$, where the servers in $\mathcal{M}_{{\bf{X}}'}$ collectively cache $\mathcal{I}_{{\bf{X}}'}$, within latency requirements. This implies that $\Delta U\left({\bf{X}}',m,i\right)=\Delta U\left({\bf{X}},m,i\right)=0$.
    \item If users cannot download model $i$ from any edge server in either $\mathcal{M}_{{\bf{X}}}$ or $\mathcal{M}_{{\bf{X}}'}$ within latency requirements, then $\Delta U\left({\bf{X}}',m,i\right)=\Delta U\left({\bf{X}},m,i\right)$.
    \item If users cannot download model $i$ from any edge server in $\mathcal{M}_{{\bf{X}}}$ within latency requirements, while they can download it from some edge server in $\mathcal{M}_{{\bf{X}}'}$ within latency requirements, then $0=\Delta U\left({\bf{X}}',m,i\right)\le \Delta U\left({\bf{X}},m,i\right)$.
\end{itemize}
Therefore, we have $U\left(\eta\left({\bf{X}}\right)\cup\left\{x_{m,i}\right\}\right) - U\left(\eta\left({\bf{X}}\right)\right) \ge U\left(\eta\left({\bf{X}}'\right)\cup\left\{x_{m,i}\right\}\right) - U\left(\eta\left({\bf{X}}'\right)\right)$ when $\eta\left({\bf{X}}\right)\subseteq\eta\left({\bf{X}}'\right)\subseteq{\bf{V}}$. Hence, the objective function of $\mathcal{P}1.1$ is a submodular function.

Next, we prove that the constraint function in $\mathcal{P}1.1$ is a set of submodular functions. Assume that ${\bf{X}}^1_{m}$ and ${\bf{X}}^2_{m}$ are two feasible model placement decisions for edge server $m$ such that $\eta\left({\bf{X}}^1_{m}\right)\subseteq\eta\left({\bf{X}}^2_{m}\right)\subseteq{\bf{V}}_{m}$, where ${\bf{V}}_{m}$ is the universal set of all feasible model placement results for edge server $m$ in $\mathcal{P}1.1$. For any $x_{m,i}\in{\bf{V}}_{m}\setminus\eta\left({\bf{X}}^2_{m}\right)$, the marginal increase in storage cost when updating $x_{m,i}=0$ to $x_{m,i}=1$ in ${\bf{X}}^2_{m}$ and ${\bf{X}}^1_{m}$ is 
\begin{equation}
    \Delta g\left({\bf{X}}^2_{m},i\right)= g_m\left(\eta\left({\bf{X}}^2_{m}\right)\cup\left\{x_{m,i}\right\}\right)-g_m\left(\eta\left({\bf{X}}^2_{m}\right)\right),
\end{equation}
and
\begin{equation}
    \Delta g\left({\bf{X}}^1_{m},i\right)= g_m\left(\eta\left({\bf{X}}^1_{m}\right)\cup\left\{x_{m,i}\right\}\right)-g_m\left(\eta\left({\bf{X}}^1_{m}\right)\right),
\end{equation}
respectively. Since $\eta\left({\bf{X}}^1_{m}\right)\subseteq\eta\left({\bf{X}}^2_{m}\right)$, we have $\mathcal{I}_{{\bf{X}}^1_{m}}\subseteq\mathcal{I}_{{\bf{X}}^2_{m}}$. We consider the following cases for $\Delta g\left({\bf{X}}^2_{m},i\right)$ and $\Delta g\left({\bf{X}}^1_{m},i\right)$.
\begin{itemize}
    \item If $i\in \mathcal{I}_{{\bf{X}}^1_{m}}$, then $i\in \mathcal{I}_{{\bf{X}}^2_{m}}$ holds, implying $\Delta g\left({\bf{X}}^2_{m},i\right)=\Delta g\left({\bf{X}}^1_{m},i\right)=0$.
    \item If $i\notin\mathcal{I}_{{\bf{X}}^1_{m}}$ and $i\notin\mathcal{I}_{{\bf{X}}^2_{m}}$, we can derive that $\Delta g\left({\bf{X}}^2_{m},i\right)=\Delta g\left({\bf{X}}^1_{m},i\right)$.
    \item If $i\notin \mathcal{I}_{{\bf{X}}^1_{m}}$ and $i\in \mathcal{I}_{{\bf{X}}^2_{m}}$, then $0=\Delta g\left({\bf{X}}^2_{m},i\right)\le \Delta g\left({\bf{X}}^1_{m},i\right)$.
\end{itemize}
Therefore, we have $g_m\left(\eta\left({\bf{X}}^1_{m}\right)\cup\left\{x_{m,i}\right\}\right)-g_m\left(\eta\left({\bf{X}}^1_{m}\right)\right)\ge g_m\left(\eta\left({\bf{X}}^2_{m}\right)\cup\left\{x_{m,i}\right\}\right)-g_m\left(\eta\left({\bf{X}}^2_{m}\right)\right)$ when $\eta\left({\bf{X}}^1_{m}\right)\subseteq\eta\left({\bf{X}}^2_{m}\right)\subseteq{\bf{V}}_{m}$. This implies that $g_m\left({\bf{X}}_{m}\right)$ is a submodular function and the constraint function of $\mathcal{P}1.1$ is a set of submodular functions. This completes the proof.

\section{Proof of Proposition \ref{NPhard}}\label{proof_of_NPhard}

Let a binary variable $y_{m,j}$ represent the caching status of parameter block $j$ on edge server $m$, where $y_{m,j}=1$ indicates that parameter block $j$ is cached on edge server $m$. Therefore, the relationships between $x_{m,i}$ and $y_{m,j}$ are given by $x_{m,i}=\prod\limits_{j\in\mathcal{J}_i}y_{m,j}$ and $y_{m,j} = 1-\prod\limits_{i\in\mathcal{I}_j}\left(1-x_{m,i}\right)$. Based on these, $U\left(\bf{X}\right)$ can be equivalently expressed as 
%With this relationship, $U\left({\bf{X}}\right)$ is transformed into 
$U_{1}\left(\bf{Y}\right)=\frac{\sum\limits_{k\in\mathcal{K}}\sum\limits_{i\in\mathcal{I}}p_{k,i}\left[1-\prod\limits_{m\in\mathcal{M}}\left(1-\prod\limits_{j\in\mathcal{J}_i}y_{m,j}{\mathbb{I}}_{1}\left(m,k,i\right)\right)\right]}{\sum\limits_{k\in\mathcal{K}}\sum\limits_{i\in\mathcal{I}}{p}_{k,i}}$, where ${\bf{Y}}=\left\{y_{m,j}\ \middle| \ m\in\mathcal{M}, j\in\mathcal{J}\right\}$. Accordingly, ${\mathcal{P}1.1}$ can be equivalently transformed to a supermodular maximization problem with multiple knapsack constraints, which is expressed as
\begin{subequations}
		\begin{equation}
			{\mathcal{P}1.2}:\ \mathop{\max}\limits_{\bf{Y}}\ U_{1}\left(\bf{Y}\right)
		\end{equation}	
		\begin{equation} \label{eq_1_2_const1}
			{\rm{s.t.}} \ \sum\limits_{j\in\mathcal{J}} D'_jy_{m,j}\le Q_m,\ \forall m\in{\mathcal{M}},
		\end{equation}	
		\begin{equation} 
			y_{m,j}\in\left\{0,1\right\},\ \forall m\in{\mathcal{M}},\forall j\in{\mathcal{J}},
		\end{equation}	
	\end{subequations}
where \eqref{eq_1_2_const1} is a set of knapsack constraints. The proof that the objective function in $\mathcal{P}1.2$ is supermodular is omitted, as it follows directly from the proof of Proposition~\ref{proposition_submodular}.

It has been shown that the supermodular maximization problem with multiple knapsack constraints is NP-hard, and there is no approximation algorithm with a constant approximation guarantee to solve this problem even when $M=1$~\cite{KELLERER201764}. Therefore, no polynomial-time algorithm with a constant approximation guarantee can solve $\mathcal{P}1.2$, and the same hardness result holds for $\mathcal{P}1.1$. This completes the proof.

\section{Proof of Proposition \ref{proposition_eq_u_m}}\label{proof_proposition_eq_u_m}
The second equality in \eqref{eq_u_m} holds due to the following reasons. First, $\mathbb{I}_{2}\left(m,k,i\right)$ ensures that model requests already served by the first $m-1$ edge servers are not counted in the cache hit ratio of edge server $m$. Second, Algorithm \ref{algorithm_successive_greedy} solves each sub-problem $\mathcal{P}2.1_{m}$ sequentially in ascending order of $m$, and the caching decisions of the first $m-1$ edge servers (i.e., $\hat{x}_{m',i}$ in  \eqref{eq_i2}) are known when solving $\mathcal{P}2.1_{m}$ for edge server $m$. %Edge server $m$ can select and cache a subset of models from model library $\mathcal{I}$ after the first $m-1$ edge servers make caching decisions. 
Therefore, the set of model requests served by edge server $m$ and those served by the first $m-1$ edge servers are disjoint, leading to the second equality in \eqref{eq_u_m}. 

\section{Proof of Proposition \ref{Successive_greedy_optimal}}\label{proof_Successive_greedy_optimal}
%The proof is omitted here. A similar proof can be referred to \cite{dawande2000approximation}. 
 	%Suppose $\bf{V}$ is the universal set, and let ${\bf{V}}'={\bf{V}}\setminus\hat{{\bf{X}}}$ be the set of placement decisions which are not adopted at the end of the TrimCaching Spec algorithm. In ${\bf{V}}'$, there can be some placement decisions under the optimal solution ${\bf{X}}^*$. Let ${\bf{Z}}={\bf{V}}'\cap{\bf{X}}^*$ be the placement strategy that should have been conducted. 
    \rev{Let $\tilde{\bf{X}}_m$ denote the set of placement decisions for edge server $m$ made by ${\bf{X}}_{m}^*$ but not by $\hat{{\bf{X}}}_{m}$, i.e., $\tilde{x}_{m,i}\in\ {\tilde{\bf{X}}}_{m}$ satisfies 
    \begin{equation}
        \tilde{x}_{m,i}=
        \begin{cases}
            1,\ \text{if }x^{*}_{m,i}=1\text{ and }\hat{x}_{m,i}=0,\\
            0, \ \text{otherwise,}
        \end{cases}
    \end{equation}
    where ${\bf{X}}_{m}^*$ is the component of the optimal solution ${\bf{X}}^*$ associated with edge server $m$, with ${\bf{X}}^*=\bigcup\limits_{m\in\mathcal{M}}{\bf{X}}_{m}^*$, $x^{*}_{m,i}\in {\bf{X}}_{m}^*$, and $\hat{x}_{m,i}\in\hat{{\bf{X}}}_{m}$. Moreover, based on the definition in \eqref{eq_hat_u}, $\hat{U}_m\left(\tilde{{\bf{X}}}_m\right)$ is given in \eqref{eq_u_m_tilde_x} at the bottom of this page, which is used as shorthand for $\hat{U}_m\left(\tilde{{\bf{X}}}_m\middle | \bigcup\limits_{m'=1}^{m-1}\hat{{\bf{X}}}_{m'}\right)$. Note that ${\tilde{\bf{X}}}_{m}$ and $\hat{U}_m\left(\tilde{{\bf{X}}}_m\right)$ are introduced solely as auxiliary variables to analyze the theoretical performance gap between $\hat{\bf{X}}$ and ${\bf{X}}^{*}$ after all $\hat{\bf{X}}_{m}$ have been obtained.

    \begin{figure*}[b]
    \rev{
    \begin{equation}\label{eq_u_m_tilde_x}
        \begin{aligned}
        \hat{U}_m\left({\tilde{\bf{X}}}_m\right)
        &=\frac{\sum\limits_{k\in\mathcal{K}}\sum\limits_{i\in\mathcal{I}}{p}_{k,i}\tilde{x}_{m,i}\mathbb{I}_{\left\{T_{m,k,i}\le\bar{T}_{k,i}\right\}}\prod\limits_{m'=1}^{m-1}\left(1-\hat{x}_{m',i} \mathbb{I}_{\left\{T_{m',k,i}\le\bar{T}_{k,i}\right\}}\right)}{\sum\limits_{k\in\mathcal{K}}\sum\limits_{i\in\mathcal{I}}{p}_{k,i}}\\
        &=\frac{\sum\limits_{k\in\mathcal{K}}\sum\limits_{i\in\mathcal{I}}{p}_{k,i}\tilde{x}_{m,i}{\mathbb{I}}_{1}\left(m,k,i\right)\prod\limits_{m'=1}^{m-1}\left(1-\hat{x}_{m',i} {\mathbb{I}}_{1}\left(m',k,i\right)\right)}{\sum\limits_{k\in\mathcal{K}}\sum\limits_{i\in\mathcal{I}}{p}_{k,i}}.
        \end{aligned}
    \end{equation}
    }
    \end{figure*}
    
    First, we establish that $\sum\limits_{m\in\mathcal{M}}\hat{U}_m\left(\tilde{{\bf{X}}}_m\right)$ is upper bounded by $U\left(\hat{{\bf{X}}}\right)$, following the approach in~\cite{dawande2000approximation, chekuri2005polynomial}. 
    Let $\mathcal{I}_{{\tilde{\bf{X}}}_{m}}=\left\{i \ \middle| \ \tilde{x}_{m,i}=1,\tilde{x}_{m,i}\in{\tilde{\bf{X}}}_{m}\right\}$ denote the set of cached models corresponding to $\tilde{{\bf{X}}}_m$, which is a subset of $\mathcal{I}$. Since the TrimCaching Spec algorithm obtains ${\hat{\bf{X}}_m}$ by maximizing $\hat{U}_m\left(\hat{{\bf{X}}}_m\right)$ over all models in $\mathcal{I}$ when solving $\mathcal{P}2.1_{m}$, and ${\hat{\bf{X}}_m}$ is the optimal solution to $\mathcal{P}2.1_{m}$, it follows that
    \begin{equation}
        \hat{U}_m\left({\tilde{\bf{X}}}_m\right)\le \hat{U}_m\left(\hat{{\bf{X}}}_{m}\right), \ \forall m\in\mathcal{M}.
    \end{equation}
    Furthermore, summing over all $m$ and invoking Proposition~\ref{proposition_eq_u_m}, we obtain
    \begin{equation}\label{eq_x_tilde_x_hat}
        \sum\limits_{m\in\mathcal{M}}\hat{U}_m\left({\tilde{\bf{X}}}_m\right)\le \sum\limits_{m\in\mathcal{M}}\hat{U}_m\left(\hat{{\bf{X}}}_m\right)=U\left(\hat{{\bf{X}}}\right).
    \end{equation} 

    Second, we show that
    \begin{equation}\label{eq_x_optimal_x_hat_x_tilde}
        U\left({\bf{X}}^{*}\right)\le U\left(\hat{{\bf{X}}}\right)+\sum\limits_{m\in\mathcal{M}}\hat{U}_m\left({\tilde{\bf{X}}}_m\right).
    \end{equation}
    Based on \eqref{eq_u_X}, for each model request $r_{k,i}$ of user $k$ for model $i$, the cache hit indicator under $\bf{X}$ is $1-\prod\limits_{m\in\mathcal{M}}\left(1-x_{m,i}{\mathbb{I}}_{1}\left(m,k,i\right)\right)$, which equals 1 if $r_{k,i}$ is served by $\bf{X}$.
    
    To prove \eqref{eq_x_optimal_x_hat_x_tilde}, we begin by claiming that for any $r_{k,i}$, the following inequality holds. 
    \begin{equation}\label{eq_indicator}
        \begin{aligned}
            &1-\prod\limits_{m\in\mathcal{M}}\left(1-x_{m,i}^{*}{\mathbb{I}}_{1}\left(m,k,i\right)\right)\\
            &\le 1-\prod\limits_{m\in\mathcal{M}}\left(1-\hat{x}_{m,i}\mathbb{I}_{1}\left(m,k,i\right)\right)\\
            &+\sum\limits_{m\in\mathcal{M}}\tilde{x}_{m,i}\mathbb{I}_{1}\left(m,k,i\right)\prod\limits_{m'=1}^{m-1}\left(1-\hat{x}_{m',i}{\mathbb{I}}_{1}\left(m',k,i\right)\right).
        \end{aligned}
    \end{equation}
    We verify this inequality by considering two cases.
    
    (i) If $r_{k,i}$ is not served by $\bf{X}^{*}$, then $1-\prod\limits_{m\in\mathcal{M}}\left(1-x_{m,i}^{*}{\mathbb{I}}_{1}\left(m,k,i\right)\right)=0$. Since the right-hand side of \eqref{eq_indicator} is nonnegative, \eqref{eq_indicator} holds in this case. 
    
    (ii) If $r_{k,i}$ is served by $\bf{X}^{*}$, then $1-\prod\limits_{m\in\mathcal{M}}\left(1-x_{m,i}^{*}{\mathbb{I}}_{1}\left(m,k,i\right)\right)$ on the left-hand side of \eqref{eq_indicator} equals 1. We further consider two sub-cases. 
    (ii-a) If $r_{k,i}$ is also served by $\hat{\bf{X}}$, then $1-\prod\limits_{m\in\mathcal{M}}\left(1-\hat{x}_{m,i}{\mathbb{I}}_{1}\left(m,k,i\right)\right)=1$. Therefore, $1-\prod\limits_{m\in\mathcal{M}}\left(1-x_{m,i}^{*}{\mathbb{I}}_{1}\left(m,k,i\right)\right)=1-\prod\limits_{m\in\mathcal{M}}\left(1-\hat{x}_{m,i}{\mathbb{I}}_{1}\left(m,k,i\right)\right)$, and \eqref{eq_indicator} holds in this sub-case. 
    (ii-b) If $r_{k,i}$ is not served by $\hat{\bf{X}}$, then $1-\prod\limits_{m\in\mathcal{M}}\left(1-\hat{x}_{m,i}{\mathbb{I}}_{1}\left(m,k,i\right)\right)=0$, implying $\prod\limits_{m'=1}^{m-1}\left(1-\hat{x}_{m',i}{\mathbb{I}}_{1}\left(m',k,i\right)\right)=1$ for all $m \in\mathcal{M}$. Since $r_{k,i}$ is served by ${\bf{X}}^{*}$ but not by $\hat{\bf{X}}$, $r_{k,i}$ is served by $\bigcup\limits_{m\in\mathcal{M}}{\tilde{\bf{X}}}_m$, which implies $1-\prod\limits_{m\in\mathcal{M}}\left(1-\tilde{x}_{m,i}{\mathbb{I}}_{1}\left(m,k,i\right)\right)=1$. Therefore, there exists at least one $m\in\mathcal{M}$ such that $\tilde{x}_{m,i}\mathbb{I}_{1}\left(m,k,i\right)=1$, and consequently $\sum\limits_{m\in\mathcal{M}}\tilde{x}_{m,i}\mathbb{I}_{1}\left(m,k,i\right)\ge 1$. Combining this with the fact that, in this sub-case, $\prod\limits_{m'=1}^{m-1}\left(1-\hat{x}_{m',i}{\mathbb{I}}_{1}\left(m',k,i\right)\right)=1$ for all $m \in\mathcal{M}$, we obtain $ \sum\limits_{m\in\mathcal{M}}\tilde{x}_{m,i}\mathbb{I}_{1}\left(m,k,i\right)\prod\limits_{m'=1}^{m-1}\left(1-\hat{x}_{m',i}{\mathbb{I}}_{1}\left(m',k,i\right)\right)\ge 1$. Since $1-\prod\limits_{m\in\mathcal{M}}\left(1-x_{m,i}^{*}{\mathbb{I}}_{1}\left(m,k,i\right)\right)=1$ in this sub-case, \eqref{eq_indicator} holds in this sub-case. Therefore, \eqref{eq_indicator} holds for all $r_{k,i}$.
    
    Multiplying both sides of \eqref{eq_indicator} by $p_{k,i}$, summing over all $r_{k,i}$, and normalizing by $\sum\limits_{k\in\mathcal{K}}\sum\limits_{i\in\mathcal{I}}p_{k,i}$, we obtain \eqref{eq_indicator_p} and \eqref{eq_indicator_u}, which are given at the bottom of the next page. Since the second term on the right-hand side of \eqref{eq_indicator_u} equals $\sum\limits_{m\in\mathcal{M}}\hat{U}_m\left({\tilde{\bf{X}}}_m\right)$, \eqref{eq_x_optimal_x_hat_x_tilde} holds.
    \begin{figure*}[!b]
        \rev{
        \begin{equation}\label{eq_indicator_p}
            \begin{aligned}
                \frac{\sum\limits_{k\in\mathcal{K}}\sum\limits_{i\in\mathcal{I}}p_{k,i}\left[1-\prod\limits_{m\in\mathcal{M}}\left(1-x_{m,i}^{*}{\mathbb{I}}_{1}\left(m,k,i\right)\right)\right]}{\sum\limits_{k\in\mathcal{K}}\sum\limits_{i\in\mathcal{I}}p_{k,i}}
                &\le \frac{\sum\limits_{k\in\mathcal{K}}\sum\limits_{i\in\mathcal{I}}p_{k,i}\left[1-\prod\limits_{m\in\mathcal{M}}\left(1-\hat{x}_{m,i}{\mathbb{I}}_{1}\left(m,k,i\right)\right)\right]}{\sum\limits_{k\in\mathcal{K}}\sum\limits_{i\in\mathcal{I}}p_{k,i}}\\
                &+\frac{\sum\limits_{k\in\mathcal{K}}\sum\limits_{i\in\mathcal{I}}p_{k,i}\left[\sum\limits_{m\in\mathcal{M}}\tilde{x}_{m,i}\mathbb{I}_{1}\left(m,k,i\right)\prod\limits_{m'=1}^{m-1}\left(1-\hat{x}_{m',i}{\mathbb{I}}_{1}\left(m',k,i\right)\right)\right]}{\sum\limits_{k\in\mathcal{K}}\sum\limits_{i\in\mathcal{I}}p_{k,i}}.
            \end{aligned}
        \end{equation}
        \begin{equation}\label{eq_indicator_u}
            U\left({\bf{X}}^{*}\right)
        \le U\left(\hat{{\bf{X}}}\right)
        +\sum\limits_{m\in\mathcal{M}}\frac{\sum\limits_{k\in\mathcal{K}}\sum\limits_{i\in\mathcal{I}}p_{k,i}\tilde{x}_{m,i}\mathbb{I}_{1}\left(m,k,i\right)\prod\limits_{m'=1}^{m-1}\left(1-\hat{x}_{m',i}{\mathbb{I}}_{1}\left(m',k,i\right)\right)}{\sum\limits_{k\in\mathcal{K}}\sum\limits_{i\in\mathcal{I}}p_{k,i}}.
        \end{equation}
        }
    \end{figure*}

    Finally, substituting \eqref{eq_x_tilde_x_hat} into \eqref{eq_x_optimal_x_hat_x_tilde} yields
    \begin{equation}
        U\left({\bf{X}}^*\right)\le U\left(\hat{{\bf{X}}}\right)+\sum\limits_{m\in\mathcal{M}}\hat{U}_m\left({\tilde{\bf{X}}}_m\right)\le 2U\left(\hat{{\bf{X}}}\right).
    \end{equation}
    Therefore,
    \begin{equation}
        U\left(\hat{{\bf{X}}}\right)\ge \frac{1}{2}U\left({\bf{X}}^*\right).
    \end{equation}
    This completes the proof. }
 	%Suppose the optimal solution of $\mathcal{P}2.1$ is $\tilde{{\bf{X}}}^*=\bigcup\limits_{m\in\mathcal{M}}\tilde{{\bf{X}}}^*_m$. If the sequential solution $\tilde{{\bf{X}}}_m$ obtained by the Algorithm \ref{algorithm_2} is equal to $\tilde{{\bf{X}}}^*_m$, we can prove that sequentially solving $M$ sub-problems can achieve the optimal solution of $\mathcal{P}2.1$. In $\mathcal{P}2.1$, we ignore the consensus delay among the network edges. Besides, the DP-based Algorithm \ref{algorithm_2} can obtain the maximal utility with the input sets. Therefore, the sequential solution $\tilde{{\bf{X}}}_m$ is equal to $\tilde{{\bf{X}}}^*_m$.\par 
\section{Proof of Proposition \ref{DP_epsilon}}\label{proof_DP_epsilon}
Let $\hat{{\bf{X}}}^*_m$ be the optimal solution to $\mathcal{P}2.1_{m}$ without rounding, and the corresponding cache hit ratio of $\mathcal{P}2.1_{m}$ is given by 
\begin{equation}\label{eq_proof_5_0}
    \hat{U}_{m}\left(\hat{{\bf{X}}}^*_m\right)=\frac{\sum\limits_{i\in\hat{\mathcal{I}}^*_{m}}u\left(m,i\right)}{\sum\limits_{k\in\mathcal{K}}\sum\limits_{i\in\mathcal{I}}{p}_{k,i}},
\end{equation}
where $\hat{\mathcal{I}}^*_{m}=\left\{i\ \middle| \ \hat{x}^{*}_{m,i}=1, \hat{x}^{*}_{m,i}\in\hat{{\bf{X}}}^*_{m}\right\}$.
    
When $\epsilon>0$, with \eqref{eq_round_u} and \eqref{eq_rounded_u_star}, we can derive that
    \begin{equation}\label{eq_proof_5_1}
        \begin{aligned}
            \hat{U}_{m}\left(\hat{{\bf{X}}}_m\right)
            &=\frac{\sum\limits_{i\in\hat{\mathcal{I}}_{m}}u\left(m,i\right)}{\sum\limits_{k\in\mathcal{K}}\sum\limits_{i\in\mathcal{I}}{p}_{k,i}}\\
            &\ge\frac{\sum\limits_{i\in\hat{\mathcal{I}}_{m}}\epsilon u_{m,\min}\lfloor\frac{u\left(m,i\right)}{\epsilon u_{m,\min}}\rfloor}{\sum\limits_{k\in\mathcal{K}}\sum\limits_{i\in\mathcal{I}}{p}_{k,i}} 
    = \frac{\sum\limits_{i\in\hat{\mathcal{I}}_{m}}\epsilon u_{m,\min}\dot{u}\left(m,i\right)}{\sum\limits_{k\in\mathcal{K}}\sum\limits_{i\in\mathcal{I}}{p}_{k,i}}.
        \end{aligned}
    \end{equation}
    Due to the optimality of the DP approach, $\hat{{\bf{X}}}_m$ is the optimal solution to $\mathcal{P}2.1_{m}$ when it is solved using $\dot{u}\left(m,i\right)$. Therefore, $\dot{w}_{\mathcal{N}^{*}}\dot{\delta}_{m,\mathcal{N}^{*}}=\sum\limits_{i\in\hat{\mathcal{I}}_{m}}\dot{u}\left(m,i\right)$ achieved by $\hat{{\bf{X}}}_m$ is no less than that achieved by $\hat{{\bf{X}}}^*_m$, i.e.,
    \begin{equation}\label{eq_proof_5_2}
        \sum\limits_{i\in\hat{\mathcal{I}}_{m}}\dot{u}\left(m,i\right)\ge\sum\limits_{i\in\hat{\mathcal{I}}_{m}^*}\dot{u}\left(m,i\right).
    \end{equation}
    Based on \eqref{eq_round_u} and \eqref{eq_proof_5_2}, we have
    \begin{equation}\label{eq_proof_5_6}
        \begin{aligned}
            \frac{\sum\limits_{i\in\hat{\mathcal{I}}_{m}}\epsilon u_{m,\min}\dot{u}\left(m,i\right)}{\sum\limits_{k\in\mathcal{K}}\sum\limits_{i\in\mathcal{I}}{p}_{k,i}}
 	&\ge\frac{\sum\limits_{i\in\hat{\mathcal{I}}^*_{m}}\epsilon u_{m,\min}\dot{u}\left(m,i\right)}{\sum\limits_{k\in\mathcal{K}}\sum\limits_{i\in\mathcal{I}}{p}_{k,i}}\\
 	&\ge\frac{\sum\limits_{i\in\hat{\mathcal{I}}^*_{m}}\epsilon u_{m,\min}\left(\frac{u\left(m,i\right)}{\epsilon u_{m,\min}}-1\right)}{\sum\limits_{k\in\mathcal{K}}\sum\limits_{i\in\mathcal{I}}{p}_{k,i}}
        \end{aligned}
    \end{equation}
    Furthermore, with \eqref{eq_proof_5_0}, we can derive that 
    \begin{equation}\label{eq_proof_5_3}
    \begin{aligned}
        &\frac{\sum\limits_{i\in\hat{\mathcal{I}}^*_{m}}\epsilon u_{m,\min}\left(\frac{u\left(m,i\right)}{\epsilon u_{m,\min}}-1\right)}{\sum\limits_{k\in\mathcal{K}}\sum\limits_{i\in\mathcal{I}}{p}_{k,i}}
        =\hat{U}_{m}\left(\hat{{\bf{X}}}^*_m\right) - \frac{\sum\limits_{i\in\hat{\mathcal{I}}^*_{m}}\epsilon u_{m,\min}}{\sum\limits_{k\in\mathcal{K}}\sum\limits_{i\in\mathcal{I}}{p}_{k,i}}\\
 	&=\hat{U}_{m}\left(\hat{{\bf{X}}}^*_m\right)-\frac{\left|\hat{\mathcal{I}}^*_{m}\right|\epsilon u_{m,\min}}{\sum\limits_{k\in\mathcal{K}}\sum\limits_{i\in\mathcal{I}}{p}_{k,i}}.
    \end{aligned}
    \end{equation}
Combining \eqref{eq_proof_5_1}, \eqref{eq_proof_5_6}, and \eqref{eq_proof_5_3}, we have 
\begin{equation}
    \hat{U}_{m}\left(\hat{{\bf{X}}}_m\right)\ge \hat{U}_{m}\left(\hat{{\bf{X}}}^*_m\right)-\frac{\left|\hat{\mathcal{I}}^*_{m}\right|\epsilon u_{m,\min}}{\sum\limits_{k\in\mathcal{K}}\sum\limits_{i\in\mathcal{I}}{p}_{k,i}},
\end{equation}
leading to 
\begin{equation}\label{eq_proof_5_4}
    \frac{\hat{U}_{m}\left(\hat{{\bf{X}}}^*_m\right)-\hat{U}_{m}\left(\hat{{\bf{X}}}_m\right)}{\hat{U}_{m}\left(\hat{{\bf{X}}}^*_m\right)}\le\frac{\left|\hat{\mathcal{I}}^*_{m}\right|\epsilon u_{m,\min}}{\hat{U}_{m}\left(\hat{{\bf{X}}}^*_m\right){\sum\limits_{k\in\mathcal{K}}\sum\limits_{i\in\mathcal{I}}{p}_{k,i}}}.
\end{equation}
Since the number of cache hits of any model in $\hat{\mathcal{I}}^*_{m}$ is no less than $ u_{m,\min}$, we have
 \begin{equation}
     \frac{\left|\hat{\mathcal{I}}^*_{m}\right|u_{m,\min}}{\sum\limits_{k\in\mathcal{K}}\sum\limits_{i\in\mathcal{I}}{p}_{k,i}}\le \hat{U}_{m}\left(\hat{{\bf{X}}}^*_m\right),
 \end{equation}
 which implies 
 \begin{equation}\label{eq_proof_5_5}
     \frac{\left|\hat{\mathcal{I}}^*_{m}\right|\epsilon u_{m,\min}}{\hat{U}_{m}\left(\hat{{\bf{X}}}^*_m\right){\sum\limits_{k\in\mathcal{K}}\sum\limits_{i\in\mathcal{I}}{p}_{k,i}}}\le\epsilon.
 \end{equation}
 Therefore, substituting \eqref{eq_proof_5_5} into \eqref{eq_proof_5_4}, we can conclude that 
    \begin{equation}
        \frac{\hat{U}_{m}\left(\hat{{\bf{X}}}^*_m\right)-\hat{U}_{m}\left(\hat{{\bf{X}}}_m\right)}{\hat{U}_{m}\left(\hat{{\bf{X}}}^*_m\right)}\le\epsilon,
    \end{equation}
which yields
\begin{equation}
    \hat{U}_{m}\left(\hat{{\bf{X}}}_m\right)\ge\left(1-\epsilon\right)\hat{U}_{m}\left(\hat{{\bf{X}}}^*_m\right).
\end{equation}

When $\epsilon=0$, we have $\dot{u}\left(m,i\right)=u\left(m,i\right)$, according to \eqref{eq_round_u}. In this case, Algorithm \ref{algorithm_DP} obtains the optimal solution to $\mathcal{P}2.1_m$ as it traverses all possible combinations of shared parameter blocks and leverages the optimality of the DP approach. Therefore, it follows that $\hat{U}_{m}\left(\hat{{\bf{X}}}_m\right)=\hat{U}_{m}\left(\hat{{\bf{X}}}^*_m\right)$, completing the proof.

\section{Proof of Theorem \ref{theorem_spec_time}}\label{proof_theorem_spec_time}
First, the time complexity of Line \ref{line:dp_u} in Algorithm \ref{algorithm_DP} is $O\left(I\right)$. Second, Algorithm \ref{algorithm_DP} requires at most $2^{\left|\mathcal{J}^{\text{sh}}\right|}$ iterations to traverse all combinations of shared parameter blocks. Third, in each iteration of Algorithm \ref{algorithm_DP}, the complexity of updating $\mathcal{T}\left(e_{\mathcal{N}},\dot{w}_{\mathcal{N}}\right)$ is $O\left(I\frac{{p}}{\delta_{\min}}\right)$, where ${p} = \sum\limits_{k\in\mathcal{K}}\sum\limits_{i\in\mathcal{I}}{p}_{k,i}$ and $\delta_{\min}$ is the granularity of $p_{k,i}$. Moreover, the time complexity of Lines \ref{line:dp_d}, \ref{line:dp_delta}, \ref{line:dp_D}, \ref{line:dp_w}, and \ref{line:dp_U} is $O\left(1\right)$, $O\left(1\right)$, $O\left(I\right)$, $O\left(\frac{{p}}{\delta_{\min}}\right)$, and $O\left(1\right)$, respectively. Finally, the complexity of Algorithm \ref{algorithm_recursive}, which is invoked in Line \ref{line:dp_x} of Algorithm \ref{algorithm_DP}, is $O\left(I\right)$. Therefore, the time complexity of Algorithm \ref{algorithm_DP} is $O\left(I+2^{\left|\mathcal{J}^{\text{sh}}\right|}\left(I\frac{{p}}{\delta_{\min}}+I\right)\right)$. Since Algorithm \ref{algorithm_successive_greedy} executes Algorithm~\ref{algorithm_DP} for each of the $M$ edge servers, the time complexity of the TrimCaching Spec algorithm is $O\left(MI+M2^{\left|\mathcal{J}^{\text{sh}}\right|}\left(I\frac{{p}}{\delta_{\min}}+I\right)\right)=O\left(MI+2^{\left|\mathcal{J}^{\text{sh}}\right|}\left(\frac{{p}}{\delta_{\min}}+1\right)MI\right)=O\left(2^{\left|\mathcal{J}^{\text{sh}}\right|}MI\right)$. Under Assumption \ref{assumption_1}, $\left|\mathcal{J}^{\text{sh}}\right|$ is small and fixed, which is a constant independent of the problem scale. Therefore, the time complexity of the TrimCaching Spec algorithm is $O\left(MI\right)$. This completes the proof.

\section{Proof of Corollary \ref{corollary_1}}\label{proof_corollary_1}
%With the bottom-layer freezing technique, we can bypass numerous combinations of shared layers in the traversing process. 
Given that shared layers originate from a single pre-trained model and that fine-tuning is performed using the bottom-layer freezing technique, caching any shared layer on an edge server requires that all its preceding layers are also cached. This inherent relationship significantly narrows down the range of layer combinations that need to be traversed. Based on the proof of Theorem \ref{theorem_spec_time}, the time complexity of the TrimCaching Spec algorithm can be reduced from $O\left(2^{\left|\mathcal{J}^{\text{sh}}\right|}MI\right)$ to $O\left(\left(\kappa+1\right)MI\right)$, completing the proof. 

\section{Proof of Theorem \ref{theorem_spec_final}}\label{proof_theorem_spec_final}
% First,  with \eqref{eq_u_m}, we have
% \begin{equation}
%     U\left(\bigcup\limits_{m\in\mathcal{M}}\hat{{\bf{X}}}^*_m\right) = \sum\limits_{m\in\mathcal{M}}\hat{U}_m\left(\hat{{\bf{X}}}^*_m\right).
% \end{equation}
First, based on Proposition \ref{proposition_eq_u_m} and Proposition \ref{DP_epsilon}, we can derive that 
\begin{equation}\label{eq_proof_2_1}
    \begin{aligned}
        U\left(\hat{{\bf{X}}}\right)=\sum\limits_{m\in\mathcal{M}}\hat{U}_m\left(\hat{{\bf{X}}}_m\right)
        &\ge\sum\limits_{m\in\mathcal{M}}\left(1-\epsilon\right)\hat{U}_m\left(\hat{{\bf{X}}}^*_m\right)\\
        &=\left(1-\epsilon\right)U\left(\bigcup\limits_{m\in\mathcal{M}}\hat{{\bf{X}}}_m^*\right).
        %=\left(1-\epsilon\right)U\left(\hat{{\bf{X}}}^*\right),
    \end{aligned}
\end{equation}
%where $\hat{{\bf{X}}}^* = \bigcup\limits_{m\in\mathcal{M}}\hat{{\bf{X}}}_m^*$. %is the optimal solution to $\mathcal{P}1.1$, determined by employing the successive greedy algorithm. 
Next, from Proposition \ref{Successive_greedy_optimal}, it follows that 
\begin{equation}\label{eq_proof_2_2}
    U\left(\bigcup\limits_{m\in\mathcal{M}}\hat{{\bf{X}}}_m^*\right)\ge\frac{1}{2}U\left({\bf{X}}^*\right).
\end{equation}
Finally, combining \eqref{eq_proof_2_1} and \eqref{eq_proof_2_2}, we have
\begin{equation}
    U\left(\hat{{\bf{X}}}\right)\ge\frac{1-\epsilon}{2}U\left({\bf{X}}^*\right),
\end{equation}
which completes the proof. 
\end{appendices}

%%%%%%%%%end of the appendix
\end{document}